\documentclass[11pt]{article}
\usepackage{amsmath,amsthm,bm, amscd}
\usepackage{fullpage}
\usepackage[nofullpage,full,titlepage,nousetoc,nouselot,nouselof,hylinks,final]{boaz}

\usepackage{graphicx}

\newcommand{\supp}{\mathsf{supp}}

\providecommand{\lnorm}[2]{\|{#1}\|_{\ell^{#2}}}

\providecommand{\lone}[1]{\lnorm{#1}{1}}
\providecommand{\linfty}[1]{\lnorm{#1}{\infty}}
\providecommand{\ltwo}[1]{\lnorm{#1}{2}}

\newcommand{\zo}{\bits}
\newcommand{\adv}{\calA}
\newcommand{\cd}{\calD}
\newcommand{\df}{\textbf{F}}
\newcommand{\dl}{\textbf{L}}
\newcommand{\dr}{\textbf{R}}
\newcommand{\dq}{\textbf{Q}}
\newcommand{\dbq}{\overline{\dq}}
\newcommand{\mc}{\mathsf{Count}}

\newcommand{\lrmac}{\mathsf{lrMAC}}
\def\calA{{\mathcal A}}

\def\calD{{\mathcal D}}

\def\dbC{{\mathbb C}}\def\dbE{{\mathbb E}}

\def\dbR{{\mathbb R}}
\def\dbZ{{\mathbb Z}}
\DeclareMathOperator{\nm}{\nmExt} % {nmExt}
\DeclareMathOperator{\expect}{E}

\newcommand{\purify}{\mathsf{purify}}
\newcommand{\laext}{\mathsf{laExt}}

\theoremstyle{definition}

\newcommand{\set}[1]{\left \{{#1} \right \}}

\newcommand{\eps}{\epsilon}

\newcommand{\Cond}{\mathsf{Cond}}

\newcommand{\Supp}{\mathsf{Supp}}
\newcommand{\Raz}{\mathsf{Raz}}

\newcommand{\Ext}{\mathsf{Ext}}

\newcommand{\nmExt}{\mathsf{nmExt}}
\newcommand{\mac}{\mathsf{MAC}}
\newcommand{\IP}{\mathsf{IP}}
\newcommand{\Enc}{\mathsf{Enc}}

\newcommand{\TExt}{\mathsf{TExt}}
\newcommand{\zuc}{\Cond}
\newcommand{\scond}{\mathsf {Scond}}

\newcommand{\BI}{\begin{itemize}}
\newcommand{\EI}{\end{itemize}}
\newcommand{\BE}{\begin{enumerate}}
\newcommand{\EE}{\end{enumerate}}

\newtheorem{thm}{Theorem}      % A counter for Theorems etc
\newcommand{\BT}{\begin{theorem}}   \newcommand{\ET}{\end{theorem}}
\newcommand{\BD}{\begin{definition}}   \newcommand{\ED}{\end{definition}}
\newcommand{\BCR}{\begin{corollary}} \newcommand{\ECR}{\end{corollary}}
\newtheorem{constr}[thm]{Construction}
\newcommand{\BCT}{\begin{constr}} \newcommand{\ECT}{\end{constr}}
\newcommand{\BL}{\begin{lemma}}   \newcommand{\EL}{\end{lemma}}

\newcommand{\BP}{\begin{proposition}}   \newcommand{\EP}{\end{proposition}}
\newcommand{\BCM}{\begin{claim}}   \newcommand{\ECM}{\end{claim}}
\newcommand{\BF}{\begin{fact}}   \newcommand{\EF}{\end{fact}}
\newcommand{\BA}{\begin{assumption}}   \newcommand{\EA}{\end{assumption}}

\def\eps{\varepsilon}

\def\le{\leqslant} \def\ge{\geqslant}

\makeatletter
\def\ExtendSymbol#1#2#3#4#5{\ext@arrow 0099{\arrowfill@#1#2#3}{#4}{#5}}
\def\RightExtendSymbol#1#2#3#4#5{\ext@arrow 0359{\arrowfill@#1#2#3}{#4}{#5}}
\def\LeftExtendSymbol#1#2#3#4#5{\ext@arrow 6095{\arrowfill@#1#2#3}{#4}{#5}}
\makeatother
\newcommand\llrightarrow[2][]{\RightExtendSymbol{-}{-}{\rightarrow}{#1}{#2}}

\newcommand\llleftarrow[2][]{\RightExtendSymbol{\leftarrow}{-}{-}{#1}{#2}}

\newcommand{\hinf}{H_\infty}
\newcommand{\thinf}{\widetilde{H}_\infty}

\begin{document}

\begin{titlepage}
\def\thepage{}

%\date{}
\title{
Non-Malleable Extractors, Two-Source Extractors and Privacy Amplification 
}

\author{
Xin Li\thanks{ Partially supported by
NSF Grants CCF-0634811, CCF-0916160, THECB ARP Grant 003658-0113-2007, and a Simons postdoctoral fellowship.}\\
Department of Computer Science\\
University of Washington\\
Seattle, WA 98905, U.S.A.\\
lixints@cs.washington.edu\\
%\and David Zuckerman\footnotemark[1]
% \thanks{Partially supported by NSF Grants CCF-0634811 and CCF-0916160 and THECB ARP Grant 003658-0113-2007.}
%\\
%Department of Computer Science\\
%University of Texas at Austin\\
%1616 Guadalupe, Suite 2.408\\
%Austin, TX 78701, U.S.A.\\
%diz@cs.utexas.edu
}

\maketitle \thispagestyle{empty}

\begin{abstract}
Dodis and Wichs \cite{DW09} introduced the notion of a non-malleable extractor to study the problem of privacy amplification with an active adversary. A non-malleable extractor is a much stronger version of a strong extractor. Given a weakly-random string $x$ and a uniformly random seed $y$ as the inputs, the non-malleable extractor $\nm$ has the property that $\nm(x,y)$ appears uniform even given $y$ as well as $\nm(x,\adv(y))$, for an arbitrary function $\adv$ with $\adv(y) \neq y$. Dodis and Wichs showed that such an object can be used to give optimal privacy amplification protocols with an active adversary.

Previously, there are only two known constructions of non-malleable extractors \cite{DLWZ11, CRS11}. Both constructions only work for $(n, k)$-sources with $k>n/2$. Interestingly, both constructions are also two-source extractors. 

In this paper, we present a strong connection between non-malleable extractors and two-source extractors. The first part of the connection shows that non-malleable extractors can be used to construct two-source extractors. If the non-malleable extractor works for small min-entropy and has a short seed length with respect to the error, then the resulted two-source extractor beats the best known construction of two-source extractors. This partially explains why previous constructions of non-malleable extractors only work for sources with entropy rate $>1/2$, and why explicit non-malleable extractors for small min-entropy may be hard to get. 

The second part of the connection shows that certain two-source extractors can be used to construct non-malleable extractors. Using this connection, we obtain the first construction of non-malleable extractors for $k < n/2$. Specifically, we give an unconditional construction for min-entropy $k=(1/2-\delta)n$ for some constant $\delta>0$, and a conditional (semi-explicit) construction that can potentially achieve $k=\alpha n$ for any constant $\alpha>0$.

We also generalize non-malleable extractors to the case where there are more than one adversarial seeds, and show a similar connection between the generalized non-malleable extractors and two-source extractors.

Finally, despite the lack of explicit non-malleable extractors for arbitrarily linear entropy, we give the first 2-round privacy amplification protocol with asymptotically optimal entropy loss and communication complexity for $(n, k)$ sources with $k=\alpha n$ for any constant $\alpha>0$. This dramatically improves previous results and answers an open problem in \cite{DLWZ11}.
 %that simultaneously achieves optimal round complexity (2-round), optimal entropy loss and optimal communication complexity
\end{abstract}
\end{titlepage}

\section{Introduction}
The broad area of \emph{randomness extraction} studies the problem of converting a weakly random source into a distribution that is close to the uniform distribution in statistical distance. Over the past decades extensive research has been conducted in this area. Among which, a long line of research (\cite{SrinivasanZ99, Trevisan01, RazRV02, LuRVW03, GuruswamiUV09, DvirW08, DvirKSS09} to name a few) studies the so called ``seeded extractors", as defined by Nisan and Zuckerman \cite{NisanZ96}. Besides its original motivation in computing with imperfect random sources, seeded extractors have found applications in coding theory, cryptography, complexity and many other areas. We refer the reader to \cite{FortnowS02, Vadhan02, Shatiel11} for a survey on this subject. Nowadays we have nearly optimal constructions of seeded extractors \cite{LuRVW03, GuruswamiUV09, DvirW08, DvirKSS09}.

Another line of research focuses on the problem of extracting random bits from several independent sources \cite{ChorG88, BarakIW04, BarakKSSW05, Raz05, Bourgain05, Rao06, BarakRSW06, Li11b}. In this case, however, the best known construction is far from optimal. Specifically, the probabilistic method shows that there exists an extractor for two independent sources on $n$ bits with each having roughly $\log n$ bits of entropy, while the best two-source extractor to date can only achieve entropy slightly below $n/2$ \cite{Bourgain05}. The best known extractor for small entropy $k$ requires $O(\log n/ \log k)$ independent sources \cite{Rao06, BarakRSW06}. Moreover, it seems hard to improve these results. Especially in the two-source case, after decades of efforts the entropy requirement only drops from anything above $n/2$ \cite{ChorG88} to slightly below $n/2$ \cite{Bourgain05}. 

Recently, a new kind of seeded extractors, called \emph{non-malleable extractors} were introduced in \cite{DW09} to give protocols for the problem of privacy amplification with an active adversary. We now give the definition of a non-malleable extractor below. As a comparison, we also give the definition of a strong seeded extractor.

\smallskip\noindent
{\bf Notation.} We let $[s]$ denote the set $\set{1,2,\ldots,s}$.
For $\ell$ a positive integer,
$U_\ell$ denotes the uniform distribution on $\zo^\ell$, and
for $S$ a set,
$U_S$ denotes the uniform distribution on $S$.
When used as a component in a vector, each $U_\ell$ or $U_S$ is assumed independent of the other components.
We say $W \approx_\eps Z$ if the random variables $W$ and $Z$ have distributions which are $\eps$-close in variation distance.

\begin{definition}
The \emph{min-entropy} of a random variable~$X$ is
\[ H_\infty(X)=\min_{x \in \supp(X)}\log_2(1/\Pr[X=x]).\]
For $X \in \zo^n$, we call $X$ an $(n,H_\infty(X))$-source, and we say $X$ has
\emph{entropy rate} $H_\infty(X)/n$. We say $X$ is a flat source if it is the uniform distribution over some subset $S \subset \bits^n$.
\end{definition}
 
\begin{definition}\label{def:strongext}
A function $\Ext : \bits^n \times \bits^d \rightarrow \bits^m$ is  a \emph{strong $(k,\eps)$-extractor} if for every source $X$ with min-entropy $k$
and independent $Y$ which is uniform on $\zo^d$,
\[ (\Ext(X, Y), Y) \approx_\eps (U_m, Y).\]
% where $Y$ is the uniform distribution on $d$ bits independent of $X$.
% , and $U_m$ is the uniform distribution on $m$ bits independent of $Y$. 
\end{definition}

\begin{definition}\footnote{Following \cite{DLWZ11}, we define worst case non-malleable extractors, which is slightly different from the original definition of average case non-malleable extractors in \cite{DW09}. However, the two definitions are essentially equivalent up to a small change of parameters.}
\label{nmdef}
A function $\nm:\bits^n \times \bits^d \to \bits^m$ is a $(k,\eps)$-non-malleable extractor if,
for any source $X$ with $\hinf(X) \geq k$ and any function $\adv:\bits^d \to \bits^d$ such that $\adv(y) \neq y$ for all~$y$,
the following holds.
When $Y$ is chosen uniformly from $\bits^d$ and independent of $X$,
\[
(\nm(X,Y),\nm(X,\adv(Y)),Y) \approx_\eps (U_m,\nm(X,\adv(Y)),Y).
\]
\end{definition}

As we can see from the definitions, a non-malleable extractor is a stronger version of the strong extractor, in the sense that it requires the output to be close to uniform even conditioned on both the seed $Y$ and the output $\nm(X, \adv(Y))$ on a different but arbitrarily correlated seed $\adv(Y)$.

The motivation to study a non-malleable extractor, the privacy amplification problem, is a fundamental problem in symmetric cryptography that has been studied by many researchers. Bennett, Brassard, and Robert introduced this problem in \cite{bbr}. The basic setting is that, two parties (Alice and Bob) share an $n$-bit secret key~$X$, which is weakly random. This could happen because the secret comes from a password or biometric data, which are themselves weakly random, or because an adversary Eve managed
to learn some partial information about an originally uniform secret, for example via side channel attacks. We measure the entropy of $X$ by the min-entropy defined above. The goal is to have Alice and Bob communicate over a public
channel so that they can convert $X$ into a nearly uniform secret key. Generally, we also assume that Alice and Bob have local private uniform random bits. The problem is the presence of the adversary Eve, who can see every message transmitted in the channel and may or may not change the messages. We assume that Eve has unlimited computational power.

The case where Eve is \emph{passive}, i.e., cannot change the messages, can be solved simply by using the above mentioned strong seeded extractors. The case where Eve is \emph{active} (i.e., can change the messages in arbitrary ways), on the other hand, is much more difficult. Historically, Maurer and Wolf \cite{MW97} gave the first non-trivial protocol in this case. Their protocol takes one round and works when the entropy rate of the weakly-random secret $X$ is bigger than $2/3$. Dodis, Katz, Reyzin, and Smith \cite{dkrs} later improved this result to give protocols that work for entropy rate bigger than $1/2$. One drawback in both cases is that the final secret key $R$ is much shorter than the min-entropy of $X$. Later, Dodis and Wichs \cite{DW09} showed that no one-round protocol exists for entropy rate less than $1/2$. The first protocol that breaks the $1/2$ entropy rate barrier is due to Renner and Wolf \cite{RW03}, where they gave a protocol that works for essentially any entropy rate. However their protocol takes $O(s)$ rounds and only achieves entropy loss $O(s^2)$, where $s$ in the security parameter of the protocol. Kanukurthi and Reyzin \cite{kr:agree-close} simplified their protocol, but the parameters remain essentially the same.

In \cite{DW09}, Dodis and Wichs showed that explicit non-malleable extractors can be used to give privacy amplification protocols that take an optimal 2 rounds and achieve optimal entropy loss $O(s)$. They showed that non-malleable extractors exist when $k>2m+3\log(1/\eps) + \log d + 9$ and $d>\log(n-k+1) + 2\log (1/\eps) + 7$. However, they only constructed weaker forms of non-malleable extractors and they gave a protocol that takes 2 rounds but that still has entropy loss $O(s^2)$. Chandran, Kanukurthi, Ostrovsky and Reyzin \cite{ckor} improved the entropy loss to $O(s)$ but the number of rounds becomes $O(s)$ as well.

Dodis, Li, Wooley and Zuckerman \cite{DLWZ11} constructed the first explicit non-malleable extractor. Their construction works for entropy $k>n/2$, but they use a large seed length $d=n$ and the efficiency when outputting more than $\log n$ bits relies on an unproven assumption. Cohen, Raz, and Segev \cite{CRS11} later gave an alternative construction that also works for $k>n/2$, but uses a short seed length and does not rely on any unproven assumption. The construction in \cite{CRS11} also allows multiple adversarial functions $\{\adv_i\}$. By using the non-malleable extractors, these two papers thus gave 2-round privacy amplification protocols that achieve optimal entropy loss $O(s)$. However, since both constructions of non-malleable extractors are only shown to work for entropy $k>n/2$,\footnote{We remark that the 1-bit case construction in \cite{DLWZ11} is a special case of the construction in \cite{CRS11}. Also, it is possible that the construction in \cite{DLWZ11} can work for entropy $k \leq n/2$ (but until now nobody can prove it), but the construction in \cite{CRS11} in general cannot work for entropy $k \leq n/2$.} the protocols also only work for $k>n/2$. For any constant $\delta>0$, \cite{DLWZ11} also gave a protocol for $k=\delta n$ than runs in $\poly(1/\delta)$ rounds and achieves optimal entropy loss $O(s)$. Recently, Li \cite{Li12} introduced the notion of a non-malleable condenser, which is a relaxation of a non-malleable extractor. He showed that non-malleable condensers for $(n, k)$ sources also give privacy amplification protocols that take an optimal 2 rounds and achieve optimal entropy loss $O(s)$. However, the non-malleable condensers constructed in \cite{Li12} also only work for $k > n/2$. Thus the natural open question is whether we can construct non-malleable extractors or condensers for smaller min-entropy, and whether there are 2-round privacy amplification protocols with optimal entropy loss for smaller min-entropy.

One interesting aspect of the two known constructions of non-malleable extractors is that they are also both two-source extractors. Indeed, the construction in \cite{DLWZ11} is in fact one of the two-source extractors introduced in \cite{ChorG88}, which requires the sources to have min-entropy $>n/2$, and the construction in \cite{CRS11} is in fact the two-source extractor in \cite{Raz05}, which requires at least on of the sources to have min-entropy $>n/2$. Coincidently, when used as non-malleable extractors, both of these constructions also require the weak source to have min-entropy $>n/2$. These facts suggest that, despite the fact that these two kinds of extractors seem quite different, there may be some connections between them. However, before this work, no such connection is known. 

\subsection{Our results}
In this paper, we present a strong connection between non-malleable extractors and two-source extractors. First, we show that non-malleable extractors can be used to construct two-source extractors. If the non-malleable extractor works for small min-entropy and has a short seed length (w.r.t. $\log(1/\e)$ where $\e$ is the error of the extractor), then the resulted two-source extractor beats the best known construction of two-source extractors. 

\BT 
Assume that for any $\e>0$, we have explicit constructions of $(k, \e)$-non-malleable extractors with seed length $d=2\log(1/\e)+o(n)$ and output length $m$. Then there exists a constant $\delta>0$ and an explicit construction of two source extractors that take as input an $(n, (1/2-\delta)n)$ source and an independent $(n,k)$ source, and output $m$ bits with error $2^{-\Omega(n)}$.
\ET

Note that if $k$ is small, say $k=n/3$ then this already beats the best known two-source extractors, but better results can be achieved if we have explicit constructions of generalized non-malleable extractors. We have the following definition (which already appears in \cite{CRS11}).

\BD
A function $\nm:\bits^n \times \bits^d \to \bits^m$ is a $(r, k,\eps)$-non-malleable extractor if,
for any source $X$ with $\hinf(X) \geq k$ and any $r$ function $\adv_i:\bits^d \to \bits^d, i=1, \cdots, r$ such that $\adv_i(y) \neq y$ for all~$i$ and $y$,
the following holds.
When $Y$ is chosen uniformly from $\bits^d$ and independent of $X$,
\[
(\nm(X,Y),\{\nm(X,\adv_i(Y))\},Y) \approx_\eps (U_m,\{\nm(X,\adv_i(Y))\},Y).
\]
\ED

Here $r$ is the number of adversarial seeds. Note that traditional non-malleable extractors are just $(1, k, \e)$-non-malleable extractors according to our definition. In \appendixref{sec:nmexist} we show that for any constant $r$, $(r, k,\eps)$-non-malleable extractors exist with seed length $d > \frac{3}{2}\log(n-k)+3\log(1/\e)+O(1)$. Now we have the following theorem.

\BT 
For any constant $b>2$ and any constant $0<\delta<1$, there exists a constant $C=C(\delta)=\poly(1/\delta)$ such that the following holds. Assume that for any $\e>0$ there exists an explicit construction of $(C, k, \e)$-non-malleable extractors with seed length $d=b\log(1/\e)+o(n)$ and output length $m$. Then there exists an explicit construction of two source extractors that take as input an $(n, \delta n)$ source and an independent $(n,k)$ source, and output $m$ bits with error $2^{-\Omega(n)}$.
\ET

Note that if we have a $(C, k, \e)$-non-malleable extractor for $k=\delta n$ and some constant $C=C(\delta)=\poly(1/\delta)$ then this will give us a two-source extractor for $(n, \delta n)$ sources. If $\delta$ is small this will be a big breakthrough for two-source extractors. This also implies that, given current techniques, the $(r, k, \e)$-non-malleable extractor in \cite{CRS11} is probably the best that we can achieve.

Next, we show that in the opposite direction, certain two-source extractors can be used to construct non-malleable extractors. The two-source extractors we will use are those that are constructed based on the inner product function. More specifically, we will consider two-source extractors of the form $\TExt=\IP(f(X), Y)$, where $\IP$ is the inner product function over $\F_2$ and $f(X)$ stands for some function (encoding) of the source $X$. We have the following theorem.

%\BT 
%Assume that we have a two-source extractor $\TExt=\IP(f(X), W)$ such that when given an $(n, k)$-source $X$ and an independent $(n_2, n_2/2-\ell)$-source $W$, $\TExt$ outputs 1 bit with error $\e$. Then there exists an explicit construction of $(k, \e')$-non-malleable extractors with error $\e'=O(2^{-\ell}+\e)$.
%\ET

%Similarly, we can generalize this theorem to $(r, k, \e)$-non-malleable extractors.

\BT 
Given two integers $r, \ell$ such that $\ell > r$. Assume that we have a two-source extractor $\TExt=\IP(f(X), W)$ such that when given an $(n, k)$-source $X$ and an independent $(n_2, n_2/(r+1)-\ell)$-source $W$, $\TExt$ outputs 1 bit with error $\e$. Then there exists an explicit construction of $(r, k, \e')$-non-malleable extractors that output 1 bit with error $\e'=O(r2^{r-\ell}+2^{\frac{3r}{2}}\e)$.
\ET

Using this theorem, and by combining known two-source extractors, we obtain new and improved constructions of non-malleable extractors.  We give the first explicit constructions of non-malleable extractors that work for min-entropy $k<n/2$. One of them is unconditional and works for $k=(1/2-\delta)n$ for some universal constant $\delta>0$. The other is conditional but can potentially work for $k=\delta n$ for any constant $\delta>0$. Specifically, we have the following theorems.

%\BT
%There exist constants $0< \delta, \gamma<1$ such that for any $n \in \N$, $k = (1/2-\delta)n$ and any $\e>2^{-\gamma n}$, there exists an explicit $(k, \e)$-non-malleable extractor $\nm: \bits^n \times \bits^d \to \bits$ with $d=O(\log n+\log (1/\e))$.
%\ET

\BT \label{thm:intmain}
There exists a constant $0< \delta<1$ and an explicit $(k, \e)$-non-malleable extractor $\nm: \bits^n \times \bits^n \to \bits^m$ with $k = (1/2-\delta)n$, $m=\Omega(n)$ and $\e=2^{-\Omega(n)}$.
\ET

Our conditional result needs to use an affine extractor and an assumption from additive combinatorics, as used in \cite{BenSZ11}. Thus we first define affine extractors and state the assumption.

\BD An $[n, m, \rho, \e]$ affine extractor is a deterministic function $f: \bits^n \to \bits^m$ such that whenever $X$ is the uniform distribution over some affine subspace over $\F^n_2$ with dimension $\rho n$, we have that for every $z \in \bits^m$, 

\[|\Pr[f(X)=z]-2^{-m}| < \e.\]
\ED

Note that we bound the error by the $\ell^{\infty}$ norm instead of the traditional $\ell^1$ norm, as in \cite{BenSZ11}. We will let $\lambda$ denote the entropy loss rate, i.e., $\lambda=1-\frac{m}{\rho n}$. We note that it is straightforward to show by the probabilistic method that such extractors exist for any constant $\rho, \lambda>0$. However the state of art constructions only achieve $\lambda$ bigger than $1/2$. % \approx \frac{3}{4}$. 

\cite{BenSZ11} also introduced the Approximate Duality conjecture (ADC), which basically says that if two independent sources $X, Y$ with linear entropy are such that $\IP(X, Y)$ is not close to uniform, then there exist two subsources $X' \subset X, Y' \subset Y$ with small deficiency such that $\IP(X', Y')$ is constant. In \cite{BenSZ11} it is shown that ADC is implied by the well-known Polynomial Freiman-Ruzsa Conjecture in additive combinatorics. For a formal definition, see \sectionref{sec:constant}. We now have the following theorems.

%\BT
%Assume the ADC conjecture and that we have an explicit $[n, m, \frac{2}{3}, 2^{-m}]$ affine extractor with $m=(1-\lambda) \frac{2}{3}n$, then there exists a constant $0<\gamma<1$ such that for any $n \in \N$, $k =\frac{3 \lambda}{1+2\lambda} n$ and any $\e> 2^{-\gamma n}$, there exists a semi-explicit $(k, \e)$-non-malleable extractor $\nm: \bits^n \times \bits^d \to \bits$ with $d=O(\log n+\log (1/\e))$.
%\ET

\BT
Assume the ADC conjecture and we have an explicit $[n, m, \frac{2}{3}, 2^{-m}]$ affine extractor with $m=(1-\lambda) \frac{2}{3}n$. Then there exists a semi-explicit $(k, \e)$-non-malleable extractor $\nm: \bits^n \times \bits^d \to \bits^m$ with $k =\frac{3 \lambda}{1+2\lambda} n$, $d=\frac{3}{2+4\lambda}n$, $m=\Omega(n)$ and $\e=2^{-\Omega(n)}$.
\ET

\BT
Given a constant integer $r$, assume the ADC conjecture and we have an explicit $[n, m, \frac{r+1}{r+2}, 2^{-m}]$ affine extractor with $m=(1-\lambda) \frac{r+1}{r+2}n$. Then there exists a semi-explicit $(r, k, \e)$-non-malleable extractor with $k =\frac{(r+2) \lambda}{1+(r+1)\lambda} n$, seed length $d=\frac{r+2}{r+1+(r+1)^2\lambda}n-1$ and $\e=2^{-\Omega(n)}$.
\ET

\begin{remark}
Here we use the ``\emph{semi-explicit}" to mean that the construction may run in time $2^n$. It is semi-explicit in the sense that the running time is polynomial in the length of the extractor's truth table (note that an exhaustive search takes time $2^{2^n}$). If we have affine extractors with large output size such that $\lambda \to 0$, then we can essentially achieve $k=\alpha n$ for any constant $\alpha>0$. 
\end{remark}

%\begin{remark}
%It is also shown in \cite{BenSZ11} that a weaker form of the ADC conjecture is true.  
%\end{remark}

Finally, we give a new privacy amplification protocol for min-entropy $k=\delta n$ for any constant $\delta>0$. Although we don't have explicit non-malleable extractors or condensers for such small $k$, our protocol simultaneously achieves optimal round complexity (2 rounds), asymptotically optimal entropy loss and asymptotically optimal communication complexity. This is the first optimal privacy amplification protocol for arbitrarily linear min-entropy. We have the following theorem.

\BT 
For any constant $0<\delta<1$ there exists a constant $0<\beta<1$ such that as long as $s \leq \beta n$, there is an efficient 2-round privacy amplification protocol for any $(n, \delta n)$ weak secret $X$ with security parameter $s$, entropy loss $O(s+\log n)$ and communication complexity $O(s+\log n)$.
\ET

Thus, in the case where $k=\delta n$, our result dramatically improves all previous results. Especially, it improves the round complexity of the protocol in \cite{DLWZ11} from $\poly(1/\delta)$ to 2, and thus answers an open problem in \cite{DLWZ11}.

\section{Overview of The Constructions and Techniques}
In this section we give an overview of our constructions and the techniques used. In order to give a clean description, we shall be informal and imprecise sometimes.
 
 \subsection{From non-malleable extractors to two-source extractors}
 Given a $(k, \e)$ non-malleable extractor $\nm$ with seed length $d=2\log(1/\e)+o(n)$, here is how we can get a two-source extractor. Assume that we have an $(n, k)$ source $X$ and an independent $(n, (1/2-\delta)n)$ source $Y$ for some constant $\delta>0$. Our first step is to use the 1-bit condenser in \cite{Zuc07} to convert $Y$ into two sources $\bar{Y}_1, \bar{Y}_2$ such that each of them has $l=\Omega(n)$ bits and one of them has min-entropy at least $(1/2+\delta)l$. Note that for an appropriately chosen $\delta$ this is indeed possible. Without loss of generality assume that $\bar{Y}_1$ has min-entropy at least $(1/2+\delta)l$. 
 
 Our key observation here is that $\bar{Y}_2$ can now be viewed as a function of $\bar{Y}_1$. More precisely, we show that the source $Y$ is a convex combination of sources $\{Y^i\}$ such that for each $Y^i$, the corresponding $\bar{Y^i_1}$ also has min-entropy at least $(1/2+\delta)l$, and $\bar{Y^i_2}$ is a deterministic function of $\bar{Y^i_1}$. Now this looks like the setting of a non-malleable extractor, where we have one seed and another correlated seed. However, there is a small problem: $\bar{Y^i_1}$ and $\bar{Y^i_2}$ may be equal sometimes. To solve this, we let $Y_1=\bar{Y_1} \circ 0$ and $Y_2=\bar{Y_2} \circ 1$. In this way we guarantee that $Y^i_1$ and $Y^i_2$ are different, and $Y^i_2$ is still a function of $Y^i_1$. Finally, this only increases the length of the seed by 1. 
 
Now we are all set, and we can take the two-source extractor to be $\TExt(X, Y)=\nm(X, Y_1) \oplus \nm(X, Y_2)$. Note that the seed $Y_1$ here is not uniform. However, a simple argument shows that a non-malleable extractor with seed length $d$ and error $\e$ remains a non-malleable extractor even if the seed only has min-entropy $k'$, with error increased to $2^{d-k'}\e$. In our case, with seed length $d=l+1=\Omega(n)=2\log(1/\e)+o(n)$ and $k'=(1/2+\delta)l$, the error is $\e'=2^{d-k'}\e \approx 2^{(1/2-\delta)l}2^{-l/2}=2^{-\Omega(n)}$. By the non-malleability of $\nm$, we get that $\TExt(X, Y)$ is $2^{-\Omega(n)}$-close to uniform.

We note that given any strong extractor with error $\e$ and seed length $d$, it remains an extractor even if the seed only has min-entropy $k'$, with error increased to $2^{d-k'}\e$. However, since the seed length is at least $d=\log(n-k)+2\log(1/\e)-O(1)$, this will never be able to get $k'$ below $d/2$. On the other hand, if the extractor is non-malleable as in our case, then it does allow us to break the entropy rate $1/2$ barrier, and get a two-source extractor for an $(n, (1/2-\delta)n)$ source and an $(n, k)$ source. This shows that non-malleability is a highly non-trivial property of a seeded extractor.

Similarly, if we have $(r, k, \e)$-non-malleable extractors for larger $r$, then we can afford to have more correlated seeds $Y_i$, or equivalently, more sources in the output of the condenser. Thus we can deal with smaller entropy in $Y$. For example, for any constant $\delta>0$, condensers based on sum-product theorems \cite{BarakKSSW05, Raz05, Zuc07} allow us to convert an $(n, \delta n)$ source into a constant $D$ number of sources such that each of them has $l=\Omega(n)$ bits and one of them has min-entropy at least $0.9l$. If we have $(D-1, k, \e)$-non-malleable extractors with suitable parameters, then we can get two-source extractors or an $(n, \delta n)$ source and an $(n, k)$ source.
 
 \subsection{From two-source extractors to non-malleable extractors}
As stated before, we focus on two-source extractors of the form $\IP(f(X), Y)$, where $\IP$ is the inner product function. First consider the simplest function $\IP(X, Y)$. Note that it is a good two-source extractor. For two independent sources on $n$ bits, it works as long as the sum of the entropies of the two sources is greater than $n$. However, at first this function does not seem to be a good candidate for a non-malleable extractor. To see this, consider the inner product function over $\F_2$. Let $X$ be a source that is obtained by concatenating the bit $0$ with $U_{n-1}$, and let $Y$ be an independent uniform seed over $\bits^n$. Now for any $y \in \bits^n$, let $\adv(y)$ be $y$ with the first bit flipped. Thus we see that for all $x$ in the support of $X$, one has $\angles{x, y}=\angles{x, \adv(y)}$. Therefore, the inner product function is not a non-malleable extractor even for weak sources with min-entropy $k=n-1$. 

In the above example, we have that for all $x$ in the support of $X$, $\IP(x, y)=\IP(x, \adv(y))$. Or equivalently, $\IP(x, y+\adv(y))=0$. How does this happen? Looking closely at this example, our key observation is that this is because the range of $Y$ is \emph{too large}. Indeed, in this example the range of $Y$ is the entire $\bits^n$, thus for any $y$ the adversary can choose a different $\adv(y)$ such that $y+\adv(y)=10 \cdots 0$ so that $\forall x \in \Supp(X), \IP(x,  y+\adv(y))=0$. 

This observation suggests that we should choose the range of $Y$ to be a subset $S \subset \bits^n$, so that for some $y$'s, the adversary will be unable to choose the appropriate $\adv(y)$ from $S$. Equivalently, we take a shorter seed length $l$, choose a uniform $y \in \bits^l$ and map $y$ to an element in $\bits^n$. This is essentially an encoding. Now let us see what properties we need the encoding to have.

We start with a construction for min-entropy $k>n/2$. Assume that we have an $(n, k)$ source $X$ with $k=(1/2+\delta)n$ for some constant $\delta>0$. We take an independent and uniform $y \in \bits^l$ and encode $y$ to $\bar{y} \in \bits^n$. For any function $\adv$, let $\bar{y}'$ be the encoding of $\adv(y)$. We will use an injective encoding, so that $\forall y, \bar{y}' \neq \bar{y}$. The output of the non-malleable extractor is then $\IP(X, \bar{Y})$. 

To show that $\IP(X, \bar{Y})$ is a non-malleable extractor, it suffices to show that $\IP(X, \bar{Y})$ is close to uniform, and that $\IP(X, \bar{Y}) \oplus \IP(X, \bar{Y}')$ is close to uniform. The first part is easy. If $X$ has min-entropy $k>n/2$, then we can take $Y$ to be the uniform distribution over some $l \geq n/2$ bits. Since the encoding is injective, $\bar{Y}$ will have min-entropy $l \geq n/2$. Thus $\IP(X, \bar{Y})$ is close to uniform. For the second part, note that $\IP(X, \bar{Y}) \oplus \IP(X, \bar{Y}')=\IP(X, \bar{Y}+\bar{Y}')$. Thus now we need $\bar{Y}+\bar{Y}'$ to have large min-entropy. Indeed, in the above counterexample where $l=n$, the adversary can choose $\adv$ such that $\bar{Y}+\bar{Y}'$ is always equal to $10 \cdots 0$ and thus has entropy 0. Now when we take $l < n$ and map $\bits^l$ to $S \subset \bits^n$, we want $\bar{Y}+\bar{Y}'$ to have a large support size. 

The ideal case would be that $\bar{Y}+\bar{Y}'$ also has support size $|S|=2^l$. This can be achieved if the encoding has the following property: for every two different $y_1, y_2$, we have that $\bar{y}_1+\bar{y}'_1 \neq \bar{y}_2+\bar{y}'_2$, or equivalently, $\bar{y}_1+\bar{y}'_1 + \bar{y}_2+\bar{y}'_2 \neq 0$. Indeed, if this is true then $\bar{Y}+\bar{Y}'$ also has min-entropy $l \geq n/2$, and thus $\IP(X, \bar{Y}) \oplus \IP(X, \bar{Y}')$ is close to uniform. Looking carefully at this property, we see that it can be ensured (at least almost ensured, as we will explain shortly) if we have another property: the elements in $S$ (when viewed as vectors in $\F^n_2$) are 4-wise linearly independent. Indeed, assume that the elements in $S$ are 4-wise linearly independent. Then if $\bar{y}_1+\bar{y}'_1 + \bar{y}_2+\bar{y}'_2 = 0$, the only possible situation is that $\bar{y}'_1=\bar{y}_2$ and $\bar{y}'_2=\bar{y}_1$. Thus there cannot be three different $y_1, y_2, y_3$ such that $\bar{y}_1+\bar{y}'_1 =  \bar{y}_2+\bar{y}'_2 = \bar{y}_3+\bar{y}'_3$. Thus the min-entropy of $\bar{Y}+\bar{Y}'$ is at least $l-1$. 

So now the question is to explicitly find a large subset $S \subset \bits^n$ such that the elements in $S$ are 4-wise linearly independent. Note that in particular this implies that the sum of any two different pairs of elements in $S$ cannot be the same. Thus we have $\binom{|S|}{2} \leq 2^n$. Therefore $|S|$ can be at most roughly $2^{n/2}$. On the other hand, in order to work for any min-entropy $k>n/2$, we will need $l \geq n/2$ and thus $|S| =2^l \geq 2^{n/2}$. These are very tight upper and lower bounds. Luckily, we have explicit constructions that meet these bounds. We will think of the elements in $S$ as columns in a parity check matrix of some binary linear code. Thus we basically need a code with block length $2^{n/2}$ and message length $2^{n/2}-n$. The 4-wise linearly independent property basically is equivalent to saying that the code has distance at least 5. This is precisely the $[2^{n/2}, 2^{n/2}-n, 5]$-BCH code. Note that although the parity check matrix has $2^{n/2}$ columns, each column is $(a, a^3)$ for a different element $a \in \F^*_{2^{n/2}}$. Thus the encoding from $y$ to $\bar{y}$ can be computed efficiently. 

Once we have the encoding, we can choose $l =n/2$ and we know that $\bar{Y}$ has min-entropy $l$ and $\bar{Y}+\bar{Y}'$ has min-entropy $l-1$. Now it is straightforward to show that both $\IP(X, \bar{Y})$ and $\IP(X, \bar{Y}+\bar{Y}')$ are close to uniform. Thus we obtain a non-malleable extractor for entropy $k>n/2$.

Thinking about the above encoding for a moment, one realizes that the same encoding can be used in any two-source extractor of the form $\IP(f(X), Y)$. Specifically, assume that $\IP(f(X), Y)$ is a two-source extractor for an $(n, k)$ source $X$ and an independent $(n, n/2-1)$ source $Y$. Then by the same argument above, if we choose the seed $Y$ to be the uniform distribution over $\bits^{n/2}$ and encode $Y$ to $\bar{Y}$ like before, we will have that both $\bar{Y}$ and $\bar{Y}+\bar{Y}'$ have min-entropy at least $n/2-1$. Thus both $\IP(f(X), \bar{Y})$ and $\IP(f(X), \bar{Y}) \oplus \IP(f(X), \bar{Y}')$ are close to uniform. Therefore we get a non-malleable extractor for min-entropy $k$.

Similarly, if we have a two-source extractor $\IP(f(X), Y)$ for an $(n, k)$ source $X$ and an independent $(n, k')$ source $Y$ with $k' \approx n/(r+1)$, then we can use a BCH code with distance $2r+3$ to construct a $(r, k, \e)$-non-malleable extractor. We choose the seed $Y$ to be the uniform distribution over $\bits^{n/(r+1)}$ and encode $Y$ to $\bar{Y}$ using the parity check matrix, i.e., $\bar{Y}=(Y, Y^3, \cdots, Y^{2r+1})$ when $Y$ is viewed as an element in $\F^*_{2^{n/(r+1)}}$. Since the columns of the parity check matrix are $2(r+1)$-wise linearly independent, we can show that for any subset $S \subseteq [r]$, $\bar{Y} \oplus \bigoplus_{i \in S} \overline{\adv_i(Y)}$ has min-entropy roughly $n/(r+1)$. Thus $\IP(f(X), \bar{Y}) \oplus \bigoplus_{i \in S} \IP(f(X), \overline{\adv_i(Y)})$ is close to uniform. Therefore we get a $(r, k, \e)$-non-malleable extractor.

\subsection{Non-malleable extractors for min-entropy $k<n/2$}
We give the first construction of non-malleable extractors for min-entropy $k<n/2$ by observing that the encoding of sources in \cite{Bourgain05} gives a function $f$ such that $\IP(f(X), Y)$ is a two-source extractor for an $(n, (1/2-\delta)n)$ source $X$ and an independent $(n, k')$ source $Y$ with $k' \approx n/2$.

Specifically, let $X$ be a distribution over some vector space $\F^n_q$ and let $cX$ be the distribution obtained by sampling $x_1, x_2, \cdots, x_c$ from $c$ independent copies of $X$ and computing $\sum x_i$. By Fourier analysis and the Cauchy-Schwarz inequality one can show that in order to prove $\IP(X, Y)$ is close to uniform, it suffices to prove that $\IP(cX, Y)$ is close to uniform with a smaller error, for some integer $c>1$. In \cite{Bourgain05}, Bourgain showed that for a weak source $X$ with min-entropy rate $1/2-\delta$ for some constant $\delta>0$, one can encode $X$ to $\Enc(X)$ such that $3\Enc(X)$ is close to having min-entropy rate $1/2+\delta$. Thus $\IP(\Enc(X), Y)$ is a two-source extractor that meets our needs. Therefore we obtain our non-malleable extractors for min-entropy $k=(1/2-\delta)n$.

\subsection{Non-malleable extractors for any constant min-entropy rate}
In \cite{BenSZ11}, Ben-Sasson and Zewi showed that affine extractors with large output size can be used to construct two source extractors for min-entropy rate $<1/2$. Their ``preimage construction" can potentially achieve any constant min-entropy rate. We observe that their encoding gives a function $f$ such that $\IP(f(X), Y)$ is a two-source extractor for two independent sources with min-entropy rate $\delta$ for any constant $\delta>0$. Specifically, they showed that if we have an affine extractor with large output size, then there is an injective mapping $F: \bits^n \to \bits^{n'}$ that maps $\bits^n$ into the preimage of a certain output of the affine extractor, such that for any weak source $X$ with min-entropy $\delta n$, $F(\Supp(X))$ is not contained in any affine subspace of dimension say $(1-\delta/2)n'$. Thus when $Y$ is a $(n', \delta n')$ source, we have that $\IP(F(X), Y)$ is non-constant. Next, similar as in \cite{BenSZ11}, the ADC conjecture implies that in fact $\IP(F(X), Y)$ is close to uniform. Thus $\IP(F(X), Y)$ is a two-source extractor that meets our needs. Therefore we obtain a non-malleable extractor (and even a $(r, k, \e)$-non-malleable extractor) for min-entropy $k=\delta n$.

\subsection{Increasing output size}
We can also increase the output size to $\Omega(n)$ for all our constructions with 1 bit output. To do this, note that we encode the seed $Y$ by using the columns of a parity check matrix of a BCH code. Equivalently, the encoding is that $\bar{Y}=(Y, Y^3)$ when we use a field $\F_{2^l}$ with $l=\Theta(n)$ and $Y$ is viewed as an element in $\F^*_{2^l}$. Now treat $\F_{2^l}$ as the vector space $\F^l_2$ and take $l$ elements $b_1, \cdots, b_l \in \F_{2^l}$ that corresponds to a basis of $\F^l_2$. For each $b_i$ we define one bit $Z_i=\IP(f(X), b_i \bar{Y})$. 

We then show that $\{Z_i\}$ satisfy the conditions of a non-uniform XOR lemma, \lemmaref{special-case}. Specifically, let $Z'_i=\IP(f(X), b_i \bar{Y'})$ where $Y'=\adv(Y)$. For any non-empty subset $S_1 \subset [l]$ and any subset $S_2 \subset [l]$, by the linearity of the inner product function, the xor of $Z_i$'s where $i \in S_1$ and $Z'_j$'s where $j \in S_2$ is of the form $\IP(f(X), t_1\bar{Y}+t_2\bar{Y'})$, with $t_1, t_2 \in \F_{2^l}$. Since $S_1$ is non-empty we have $t_1 \neq 0$. We then show that $t_1\bar{Y}+t_2\bar{Y'}$ roughly has the same min-entropy as $Y$ (at least the min-entropy of $Y$ minus $\log 3$). Thus $\IP(f(X), t_1\bar{Y}+t_2\bar{Y'})$ is close to uniform. We further show that the error is $2^{-\Omega(n)}$. Thus by \lemmaref{special-case} we can output $m=\Omega(n)$ bits with error $2^{-\Omega(n)}$.

\subsection{Reducing seed length}
In all the constructions where we encode the seed $Y$ by a parity check matrix, the seed length is linear in the source length. However the error is also $2^{-\Omega(n)}$. If we only need to achieve a bigger error, we can reduce the seed length by using the parity check matrix of a BCH code with larger distance. Specifically, when the distance is $2t+1$ the seed length is roughly $n/t$. However we need to guarantee something else. For example, in the construction for min-entropy $k>n/2$, we need to show that both $\IP(X, \bar{Y})$ and $\IP(X, (\bar{Y}+\bar{Y}'))$ are still close to uniform. This can be shown as follows. Since now the columns of the parity check matrix are $2t$-wise linearly independent, both $\frac{t}{2}\bar{Y}$ and $\frac{t}{2}(\bar{Y}+\bar{Y}')$ will now have min-entropy roughly $\frac{t}{2}H_{\infty}(Y) = n/2$. Thus we can conclude that both $\IP(X, \frac{t}{2}\bar{Y})$ and $\IP(X, \frac{t}{2}(\bar{Y}+\bar{Y}'))$ are close to uniform, and therefore both $\IP(X, \bar{Y})$ and $\IP(X, (\bar{Y}+\bar{Y}'))$ are also close to uniform, by the Cauchy-Schwarz inequality. However the error increases according to the seed length. Calculations show that we can get seed length $d=O(\log n+\log (1/\e))$. 

\subsection{An optimal privacy amplification protocol for $k=\delta n$}
In \cite{DLWZ11}, the authors give a privacy amplification protocol for $k=\delta n$ with $C=\poly(1/\delta)$ rounds and entropy loss $\poly(1/\delta)s$, where $s$ is the security parameter. Here we want to somehow ``compress" the protocol into 2 rounds while still keeping the entropy loss to be $O(s)$. As in \cite{DLWZ11}, we first use the condenser in \cite{BarakKSSW05, Raz05, Zuc07} to convert the shared $(n,k)$ source $X$ into a somewhere rate-$0.9$ source $(X_1, \cdots, X_C)$ with $C=\poly(1/\delta)$ rows. Now the high-level idea of the protocol is as follows. In the first round, Alice samples a fresh random string $Y_1$ from her private random bits and sends it to Bob, where Bob receives a possibly modified version $Y'_1$. In the second round, Bob samples a fresh random string $W'$ from his private random bits and tries to send it to Alice, where Alice receives a possibly modified version $W$. We want a protocol such that if Eve does not change $Y_1$, then with high probability Bob can authenticate $W'$ to Alice and they can both output $\Ext(X, W')$ as the final outputs, by using a strong seeded extractor $\Ext$. If Eve does change $Y_1$, then with high probability Alice should be able to detect this and reject.

The first goal is relatively easy to achieve. At the end of the first round, Alice and Bob compute $Z=\Ext(X, Y_1)$ and $Z'=\Ext(X, Y'_1)$ respectively, using a strong extractor $\Ext$. If Eve does not change $Y_1$ then $Z=Z'$ and is private and uniform. Thus in the second round Bob can authenticate $W'$ to Alice by also sending a tag $T'$ produced by a standard MAC (message authentication code) with $Z$ as the key. We now focus on the second goal. If the extractor $\Ext$ in computing $Z$ and $Z'$ is non-malleable for entropy $k$ then this can be done by using the protocol proposed by Dodis and Wichs \cite{DW09}. However, we do not have explicit non-malleable extractors for $k=\delta n$. 

Nevertheless, we will still have Alice and Bob each produce a variable $V$ and $V'$ respectively. We will ensure that, if Eve changes $Y_1$ to a different $Y'_1$, then even given $T'$ and $V'$, with high probability Eve cannot come up with the correct $V$ for Alice. If this is true then in the second round we can have Bob also send $V'$ to Alice, where Alice receives a possibly modified version $\bar{V}$. Alice then checks both the tag $T$ and whether $V = \bar{V}$. If either of them fails, Alice rejects. This will give us a privacy amplification protocol.

The first problem with the above strategy is that now $V'$ may give information about $Z'$, thus now the MAC key may not be uniform. This is easy to solve since there are constructions of MACs that work as long as the key has entropy rate $>1/2$. Thus by limiting the size of $V'$ to be at most half the size of $Z'$, we can ensure that if Eve does not change $Y_1$, Bob can still authenticate $W'$ to Alice. We now explain how we produce the variables $V, V'$.

We actually have Alice produce $C$ variables $V=(V_1, \cdots, V_C)$. Similarly, Bob produces $V'=(V'_1, \cdots, V'_C)$. For this, we first choose a non-malleable extractor and have Alice and Bob each apply the extractor to the somewhere rate-$0.9$ source $(X_1, \cdots, X_C)$, using $Y_1$ and $Y'_1$ as the seeds respectively. Let the outputs be $(\bar{X}_1, \cdots, \bar{X}_C)$ and $(\bar{X'_1}, \cdots, \bar{X'_C})$. Note that one of the $X_i$'s, say $X_g$ is a rate $0.9$-source. Thus we can use the non-malleable extractors in \cite{DLWZ11, CRS11, Li12}. Now we fix $Y_1, Y'_1$, and we have that $\bar{X_g}$ is uniform and independent of $\bar{X'_g}$. Thus we can fix $\bar{X'_g}$ and $\bar{X_g}$ is still uniform. Next, we fix $Z'$. Since now $Z'$ is a deterministic function of $X$, as long as the size of $Z'$ is smaller than the size of $\bar{X_g}$, conditioned on this fixing $\bar{X_g}$ still has a lot of entropy left. We will now have Alice extract each $V_i$ from $\bar{X_i}$, and correspondingly, Bob will extract each $V'_i$ from $\bar{X'_i}$. Note that we can indeed ensure that the size of $\bar{X_g}$ is bigger than $Z'$, while the size of $Z'$ is bigger than the size of $(V'_1, \cdots, V'_C)$ just by limiting the size of each $V'_i$.

Ideally, we would want $V, V'$ be such that $V_g$ is close to uniform conditioned on $V'=(V'_1, \cdots, V'_C)$. However, we cannot achieve exactly this. Instead, what we can achieve is that $V_g$ is close to uniform conditioned on $(V'_1, \cdots, V'_g)$. Once we have this, we can limit the size of $(V'_{g+1}, \cdots, V'_C)$ to be smaller than the size of $V_g$. Thus $V_g$ still has a lot of entropy even conditioned on $V'=(V'_1, \cdots, V'_C)$. This will ensure that with high probability Eve cannot come up with the correct $V_g$. Since we do not know which one of $\{\bar{X_i}\}$ is $\bar{X_g}$, we will choose $(V_1, \cdots, V_C)$ such that the size of $V_C$ is say $2s$, and for any $i$ the size of $V_i$ is twice the size of $V_{i+1}$. In this way, no matter what $g$ is, the size of $(V'_{g+1}, \cdots, V'_C)$ is the size of $V_g$ minus $2s$. Thus $V_g$ still has $2s$ entropy left conditioned on $V'$. 

Finally we explain how we can achieve the above property. We achieve this by using the ``look-ahead" extractor in \cite{DW09} based on an alternating extraction protocol. Specifically, in the first round we also have Alice sample two other random strings $(Y_2, Y_3)$ and send them to Bob, where Bob receives $(Y'_2, Y'_3)$. Note that after we fix $(Y_1, Y'_1)$, $(Y'_2, Y'_3)$ is a deterministic function of $(Y_2, Y_3)$. Now pick a strong extractor $\Ext$ and Alice performs the following alternating extraction protocol: $S_1=Y_3, R_1=\Ext(X, S_1), S_2=\Ext(Y_2, R_1), R_2=\Ext(X, S_2), \cdots S_C=\Ext(Y_2, R_{C-1}), R_C=\Ext(X, S_C)$. Bob will perform the same protocol using $(Y'_2, Y'_3)$ and produces $\{S'_i, R'_i\}$. As long as the size of $(S_i, R_i)$ is limited, this protocol has the property that for any $i$, $(R_i, S_i)$ is uniform and independent of $\{S_j, R_j, S'_j, R'_j, j < i\}$. We now modify this protocol such that whenever $S_i$ is used to extract $R_i$, Alice also uses it to extract $V_i=\Ext(\bar{X_i}, S_i)$. Correspondingly, Bob extracts $V'_i=\Ext(\bar{X'_i}, S'_i)$. Now one can show that as long as the size of $V_i$ is also limited, we have that for any $i$, $(R_i, S_i)$ is uniform and independent of $\{S_j, R_j, V_j, S'_j, R'_j, V'_j,  j < i\}$. Specifically, we have that $S_g$ is uniform and independent of $\{S_j, R_j, V_j, S'_j, R'_j, V'_j,  j < g\}$. Moreover, we can show that conditioned on the fixing of $\{S_j, R_j, V_j, S'_j, R'_j, V'_j,  j < g\}$, both $S_g$ and $S'_g$ are deterministic functions of $(Y_2, Y_3)$ and are thus independent of $\bar{X_g}$. Furthermore, $\bar{X_g}$ still has a lot of entropy left. Now since $\Ext$ is a strong extractor, we have that $V_g=\Ext(\bar{X}_g, S_g)$ is uniform conditioned on $\{V_j, V'_j, j<g\}$ and $(S_g, S'_g)$. Note that we have fixed $(\bar{X'_g}, Z')$ before,  while $V'_g=\Ext(\bar{X'_g}, S'_g)$ and $T'$ is a function of $Z'$. Thus $V_g$ is still uniform even conditioned on $(\{V'_j, j \leq g\}, T')$. Thus we have achieved our goal.

One small problem with the above discussion is that in the first round Alice sends $(Y_1, Y_2, Y_3)$ to Bob and Bob receives $(Y'_1, Y'_2, Y'_3)$. Thus $Y'_1$ is a function of $(Y_1, Y_2, Y_3)$. Therefore fixing $(Y_1, Y'_1)$ may cause $(Y_2, Y_3)$ to lose entropy. Thus in the alternating extraction protocol $(Y_2, Y_3)$ may not be uniform. However, by making the size of $Y_1$ a constant times smaller than the size of $Y_3$, we can ensure that $Y_3$ has entropy rate $>2/3$ conditioned on $(Y_1, Y'_1)$. Thus we can add a step $0$ in the alternating extraction protocol: $S_0=Y_3, R_0=\Raz(S_0, X), S_1=\Ext(Y_2, R_0)$ and the following protocol remains the same. Here $\Raz$ is the two-source extractor in \cite{Raz05} that works as long as one of the source has entropy rate $>1/2$. This gives our whole privacy amplification protocol. Note that the entropy loss is $O(2^Cs)=2^{\poly(1/\delta)}s$. For any constant $\delta>0$ this is still $O(s)$. By using the improved non-malleable extractor in \cite{Li12} that has short seed length and large output size (In fact, it suffices to use the non-malleable condensers in \cite{Li12}, instead of extractors), we can achieve randomness complexity (the number of truly random bits needed) $O(Cs)=\poly(1/\delta)s$ and communication complexity $O(2^Cs)=2^{\poly(1/\delta)}s$. \\
      
\noindent{\bf Organization.} The rest of the paper is organized as follows. We give some preliminaries in \sectionref{sec:prelim}. In \sectionref{sec:nmtwo} we show that non-malleable extractors can be used to construct two-source extractors. In \sectionref{sec:twonm} we show that two-source extractors based on the inner product function can be used to construct non-malleable extractors. In \sectionref{sec:nmext} we give our new and improved constructions of non-malleable extractors. In \sectionref{sec:protocol} we give our privacy amplification protocol for arbitrarily linear min-entropy. We conclude with some open problems in \sectionref{sec:con}. The existence of generalized non-malleable extractors is proved in \appendixref{sec:nmexist} and an alternative construction of non-malleable extractors for entropy $(1/2-\delta)n$ is given in \appendixref{sec:alter}.
%However, in general the inner product function is not a non-malleable extractor even for weak sources with min-entropy $n-1$ and output just 1 bit. To see this, take $X$ to be the bit 0 concatenated with $U_{n-1}$.  Let $\adv(y)$ be $y$ with the first bit flipped.  Then for all $x$ in the support of $X$, one has $f(x,y) = f(x,\adv(y))$.

%\appendix

\section{Preliminaries} \label{sec:prelim}
We often use capital letters for random variables and corresponding small letters for their instantiations. Let $|S|$ denote the cardinality of the set~$S$.
Let $\dbZ_r$ denote the cyclic group $\dbZ/(r\dbZ)$,
and let $\F_q$ denote the finite field of size $q$.
All logarithms are to the base 2.

\subsection{Probability distributions}
\begin{definition} [statistical distance]Let $W$ and $Z$ be two distributions on
a set $S$. Their \emph{statistical distance} (variation distance) is
\begin{align*}
\Delta(W,Z) \eqdef \max_{T \subseteq S}(|W(T) - Z(T)|) = \frac{1}{2}
\sum_{s \in S}|W(s)-Z(s)|.
\end{align*}
\end{definition}

We say $W$ is $\eps$-close to $Z$, denoted $W \approx_\eps Z$, if $\Delta(W,Z) \leq \eps$.
For a distribution $D$ on a set $S$ and a function $h:S \to T$, let $h(D)$ denote the distribution on $T$ induced by choosing $x$ according to $D$ and outputting $h(x)$.
We often view a distribution as a function whose value at a sample point is the probability of that sample point.
Thus $\lone{W-Z}$ denotes the $\ell_1$ norm of the difference of the distributions specified by the random variables $W$ and $Z$, which equals $2\Delta(W,Z)$.

\begin{definition}

A function $\TExt : \bits^{n_1} \times \bits^{n_2} \rightarrow \bits^m$ is  a \emph{strong two source extractor} for min-entropy $k_1, k_2$ and error $\e$ if for every independent  $(n_1, k_1)$ source $X$ and $(n_2, k_2)$ source $Y$, 

\[ |(\TExt(X, Y), X)-(U_m, X)| < \e\]

and

\[ |(\TExt(X, Y), Y)-(U_m, Y)| < \e,\]
where $U_m$ is the uniform distribution on $m$ bits independent of $(X, Y)$. 
\end{definition}

\subsection{Somewhere Random Sources, Extractors and Condensers}

\begin{definition} [Somewhere Random sources] \label{def:SR} A source $X=(X_1, \cdots, X_t)$ is $(t \times r)$
  \emph{somewhere-random} (SR-source for short) if each $X_i$ takes values in $\bits^r$ and there is an $i$ such that $X_i$ is uniformly distributed.
\end{definition}

\BD
An elementary somewhere-k-source is a 	vector of sources $(X_1, \cdots, X_t)$, such that some $X_i$ is a $k$-source. A somewhere $k$-source is a convex combination of elementary somewhere-k-sources.
\ED

\BD
A function $C: \bits^n \times \bits^d \to \bits^m$ is a $(k \to l, \e)$-condenser if for every $k$-source $X$, $C(X, U_d)$ is $\e$-close to some $l$-source. When convenient, we call $C$ a rate-$(k/n \to l/m, \e)$-condenser.   
\ED

\BD
A function $C: \bits^n \times \bits^d \to \bits^m$ is a $(k \to l, \e)$-somewhere-condenser if for every $k$-source $X$, the vector $(C(X, y)_{y \in \bits^d})$ is $\e$-close to a somewhere-$l$-source. When convenient, we call $C$ a rate-$(k/n \to l/m, \e)$-somewhere-condenser.   
\ED

We are going to use condensers recently constructed based on the sum-product theorem. The following constructions are due to Zuckerman \cite{Zuc07}.

\BT [\cite{Zuc07}] \label{thm:tcondenser}
There exists a constant $\alpha>0$ such that for any constant $0<\delta<0.9$, there is an efficient family of rate-$(\delta \to (1+\alpha)\delta, \e=2^{-\Omega(n)})$-somewhere condensers $\scond: \bits^n \to (\bits^m)^2$ where $m=\Omega(n)$. 

\ET

\BT [\cite{BarakKSSW05, Raz05, Zuc07}] \label{thm:swcondenser}
For any constant $\beta, \delta>0$, there is an efficient family of rate-$(\delta \to 1-\beta, \e=2^{-\Omega(n)})$-somewhere condensers $\zuc: \bits^n \to (\bits^m)^D$ where $D=O(1)$ and $m=\Omega(n)$. 

\ET

%\iffalse
\subsection{Average conditional min-entropy}
\label{avgcase}

Dodis and Wichs originally defined non-malleable extractors with respect to average conditional min-entropy, a notion defined by
Dodis, Ostrovsky, Reyzin, and Smith \cite{dors}.

\begin{definition}
The \emph{average conditional min-entropy} is defined as
\[ \thinf(X|W)= - \log \left (\expect_{w \leftarrow W} \left [ \max_x \Pr[X=x|W=w] \right ] \right )
= - \log \left (\expect_{w \leftarrow W} \left [2^{-\hinf(X|W=w)} \right ] \right ).
\]
\end{definition}

Average conditional min-entropy tends to be useful for cryptographic applications.
By taking $W$ to be the empty string, we see that average conditional min-entropy is at least as strong as min-entropy.
In fact, the two are essentially equivalent, up to a small loss in parameters.

We have the following lemmas.

\begin{lemma} [\cite{dors}]
\label{entropies}
For any $s > 0$,
$\Pr_{w \leftarrow W} [\hinf(X|W=w) \geq \thinf(X|W) - s] \geq 1-2^{-s}$.
\end{lemma}

\BL [\cite{dors}] \label{lem:amentropy}
If a random variable $B$ has at most $2^{\ell}$ possible values, then $\thinf(A|B) \geq \hinf(A)-\ell$.
\EL

To clarify which notion of min-entropy and non-malleable extractor we mean, we use the term \emph{worst-case non-malleable extractor} when we refer to
our Definition~\ref{nmdef}, which is with respect to traditional (worst-case) min-entropy, and \emph{average-case non-malleable extractor} to refer to
the original definition of Dodis and Wichs, which is with respect to average conditional min-entropy.

\begin{corollary}
A $(k,\eps)$-average-case non-malleable extractor is a $(k,\eps)$-worst-case non-malleable extractor.
For any $s>0$, a $(k,\eps)$-worst-case non-malleable extractor is a $(k+s,\eps + 2^{-s})$-average-case non-malleable extractor.
\end{corollary}

Throughout the rest of our paper, when we say non-malleable extractor, we refer to the worst-case non-malleable extractor of Definition~\ref{nmdef}.
%\fi

\subsection{Fourier analysis}

We give some basic and standard facts about Fourier analysis here. We normalize as in \cite{DLWZ11}. 
For functions $f,g$ from a set $S$ to $\dbC$, we define the inner product $\angles{f,g} = \sum_{x \in S} f(x) \overline{g(x)}$.
Let $D$ be a distribution on $S$, sometimes we will also view it as a function from $S$ to $\dbR$.
Note that $\expect_D [f(D)] = \angles{f,D}$.
Now suppose we have functions $h:S \to T$ and $g:T \to \dbC$. 
Then
\[ \angles{g \circ h, D} = \expect_D [g(h(D))] = \angles{g,h(D)}.\]

Let $G$ be a finite abelian group, we say $\phi$ is a character of $G$ if it is a homomorphism from $G$ to $\dbC^\times$.
We call the character that maps all elements to 1 the trivial character.
Define the Fourier coefficient $\widehat{f}(\phi) = \angles{f,\phi}$, and let $\widehat{f}$ denote the vector with entries $\widehat{f}(\phi)$ for all $\phi$.
Note that for a distribution $D$, one has $\widehat{D}(\phi) = \expect_D[\phi(D)]$.

Since the characters divided by $\sqrt{|G|}$ form an orthonormal basis,
the inner product is preserved up to scale: $\angles{\widehat{f},\widehat{g}} = |G| \angles{f,g}$.  As a corollary, we obtain
Parseval's equality:
\[ \ltwo{\widehat{f}}^2=\angles{\widehat{f},\widehat{f}} = |G| \angles{f,f}=|G|\ltwo{f}^2.\]
Hence by Cauchy-Schwarz,
\begin{equation}
\label{fourier-bound}
\lone{f} \leq \sqrt{|G|} \ltwo{f} = \ltwo{\widehat{f}} \leq \sqrt{|G|} \linfty{\widehat{f}}.
\end{equation}

For functions $f,g:S \to \dbC$, we define the function $(f,g):S \times S \to \dbC$ by $(f,g)(x,y) = f(x)g(y)$.
Thus, the characters of the group $G \times G$ are the functions $(\phi,\phi')$, where $\phi$ and $\phi'$ range over all characters of $G$.
We abbreviate the Fourier coefficient $\widehat{(f,g)}((\phi,\phi'))$ by $\widehat{(f,g)}(\phi,\phi')$.  Note that
\[ \widehat{(f,g)}(\phi,\phi') = \sum_{(x,y) \in G \times G} f(x)g(y)\phi(x)\phi'(y) = \left ( \sum_{x \in G} f(x)\phi(x) \right ) \left ( \sum_{y \in G} g(x)\phi'(x) \right )
= \widehat{f}(\phi) \widehat{g}{(\phi')}. \]

In this paper, in the additive group of $\F_p$ we use the characters $e_r(s)=e^{2\pi i rs/p}$ for $r \in \F_p$. It is easy to verify that $\{e_r, r \in F_p\}$ indeed are characters and these characters divided by $\sqrt{p}$ form an orthonormal basis. Note that the trivial character corresponds to the case $r=0$.

We next generalize the characters to the additive group of the field $\F_{p^l}$. In this case, for any $r \in \F_{p^l}$, we use the character $e_r(s) = e^{2 \pi i (r \cdot s)/p}$, where $r$ and $s$ are viewed as vectors in $\F^l_p$ and $\cdot$ indicates the inner product function in $\F^l_p$. Again it is easy to verify that these indeed are characters and they form an orthonormal basis (up to a normalization factor of $p^{l/2}$).

\subsection{Non-uniform XOR lemma}
The following non-uniform XOR lemmas are proved in \cite{DLWZ11}. 
%Here we prove an extension of Vazirani's XOR lemma in the non-uniform setting. Our statement and proof parallels Rao \cite{Rao07}.

\begin{lemma}
\label{special-case}
Let $(W,W')$ be a random variable on $G \times G$ for a finite abelian group $G$, and suppose
that for all characters $\psi,\psi'$ on $G$ with $\psi$ nontrivial, one has
$$|\expect_{(W,W')}[\psi(W)\psi'(W')]| \leq \e.$$
Then the distribution of $(W,W')$ is $\e |G|$ close to $(U,W')$, where $U$ is the uniform distribution on $G$ independent of $W'$.
Moreover, for $f:G\times G \to \dbR$ defined as the difference of distributions $(W,W') - (U,W')$,
we have $\linfty{\widehat{f}} \leq \e$.
\end{lemma}

\iffalse
\begin{proof}
As implied in the lemma statement, the value of $f(a,b)$ is the probability assigned to $(a,b)$ by the distribution of $(W,W')$ minus
that assigned by $(U,W')$.

First observe that
\[
\widehat{f}(\psi,\psi') = \angles{f,(\psi,\psi')} = \expect_{(W,W')}[\psi(W)\psi'(W')] - \expect_{(U,W')}[\psi(U)\psi'(W')].
\]
Since $U$ and $W'$ are independent, this last term equals
\[\expect_{(U,W')}[\psi(U)]\expect_{(U,W')}[\psi'(W')]
= \expect_{U}[\psi(U)]\expect_{W'}[\psi'(W')] = 0,\]
since $\psi$ is nontrivial.
Therefore, by hypothesis, when $\psi$ is nontrivial, one finds that $|\widehat{f}(\psi,\psi')| \leq \e$.

When $\psi$ is trivial, we get
\[ \widehat{f}(\psi,\psi') = \expect_{(W,W')}[\psi'(W')] - \expect_{(U,W')}[\psi'(W')] = 0.\]

Hence $\lone{f} \leq \sqrt{|G \times G|} \linfty{\widehat{f}} \leq |G| \e$.
\end{proof}
\fi

\begin{lemma}\label{lem:noxor}
For every cyclic group $G=\Z_N$ and every integer $M \leq N$, there is an efficiently computable function $\sigma: \Z_N \to \Z_M=H$ such that the following holds. Let $(W,W')$ be a random variable on $G \times G$, and suppose that for all characters $\psi,\psi'$ on $G$ with $\psi$ nontrivial, one has
$$|\expect_{(W,W')}[\psi(W)\psi'(W')]| \leq \e.$$
Then the distribution $(\sigma(W), \sigma(W'))$ is $O(\e M \log N+M/N)$-close to the distribution $(U, W')$ where $U$ stands for the uniform distribution over $H$ independent of $W'$.
\end{lemma}

The following non-uniform XOR lemma is proved in \cite{CRS11}.
\BL \label{lem:noxor2}
Let $X$ be a random variable over $\bits^m$ and $Y$ be a random variable over $\bits^n$. For any subset $\sigma \subseteq [m]$ and $\tau \subseteq [n]$, let $X_{\sigma}=\oplus_{i \in \sigma} X_i$ and $Y_{\tau}=\oplus_{j \in \tau} Y_j$. Assume that for any non-empty $\sigma \subseteq [m]$ and any $\tau \subseteq [n]$, we have $X_{\sigma} \oplus Y_{\tau} \approx_{\e} U$. Then

\[|(X, Y)-(U_m, Y)| \leq ((2^m-1) \cdot 2^n)^{1/2} \e.\]
\EL

\subsection{Strong non-malleable extractor}
The following theorem is proved in \cite{Rao07}.

\BT \label{thm:stext} \cite{Rao07}
Let $\TExt: \bits^{n_1} \times \bits^{n_2} \to \bits^m$ be any two source extractor for min-entropy $k_1, k_2$ with error $\e$. Then if $X$ is an $(n_1, k_1)$ source and $Y$ is an independent $(n_1, k_2')$ source, we have

\[|(\TExt(X, Y), Y)-(U_m, Y)| \leq 2^m(2^{k_2-k_2'+1}+\e).\]
\ET

Here we prove a similar theorem that will enable our non-malleable extractor to be ``strong".

\BT \label{thm:snmext}
Let $\TExt: \bits^{n_1} \times \bits^{n_2} \to \bits^m$ be a two source extractor for min-entropy $k_1, k_2$ and $\adv_i:\bits^{n_2} \to \bits^{n_2}, i=1, \cdots, r$ be $r$ deterministic functions such that for any $(n_1, k_1)$ source $X$ and any independent $(n_2, k_2)$ source $Y$, 

\[|(\TExt(X, Y), \{\TExt(X, \adv_i(Y))\})-(U_m, \{\TExt(X, \adv_i(Y))\})| \leq \e.\]

 Then for any $(n_2, k'_2)$ source $Y'$ independent of $X$, 
 
 \[|(\TExt(X, Y'), \{\TExt(X, \adv_i(Y'))\}, Y')-(U_m, \{\TExt(X, \adv_i(Y'))\}, Y')| \leq 2^{(r+1)m}(2^{k_2-k_2'+1}+\e).\]
\ET

\begin{thmproof}
Let $W=\TExt(X, Y)$ and $W_i'=\TExt(X, \adv_i(Y))$ for any $i \in [r]$. Let $\bar{W}$ be the vector $(W_1', \cdots, W'_r)$. Let $\bar{z}$ be the vector $(z_1', \cdots, z_r') \in (\bits^m)^r $. For any $(z, \bar{z}) \in \bits^m \times (\bits^m)^r$, define the set of bad $y$'s for $(z, \bar{z})$ to be

\[B_{z, \bar{z}} = \{y: |\Pr[W=z, \bar{W}=\bar{z}]-2^{-m}\Pr[\bar{W}=\bar{z}]| > \e\}.\]

Then we must have 

\BCM
For every $(z, \bar{z}) $, $|B_{z, \bar{z}}| < 2 \cdot 2^{k_2}$.
\ECM

To see this, assume for the sake of contradiction that $|B_{z, \bar{z}}| \geq 2\cdot 2^{k_2}$ for some $(z, \bar{z}) $. Let 
\[B^{+}_{z, \bar{z}} = \{y: \Pr[W=z, \bar{W}=\bar{z}]-2^{-m}\Pr[\bar{W}=\bar{z}] > \e\}\] and

\[B^{-}_{z, \bar{z}} = \{y: \Pr[W=z, \bar{W}=\bar{z}]-2^{-m}\Pr[\bar{W}=\bar{z}] < -\e\}.\]

Then $|B_{z, \bar{z}}|=|B^{+}_{z, \bar{z}}|+|B^{-}_{z, \bar{z}} |$ and thus one of them must have size $\geq 2^{k_2}$. Without loss of generality assume that $|B^{+}_{z, \bar{z}}| \geq 2^{k_2}$. Then we can let $Y$ to be the uniform distribution over $|B^{+}_{z, \bar{z}}|$ and $Y$ is independent of $X$, but $|(W, \bar{W})-(U_m, \bar{W})| > \e$, which is a contradiction.

Let $B=\cup_{z, \bar{z}} B_{z, \bar{z}}$. We have $|B|<2^{(r+1)m} \cdot 2 \cdot 2^{k_2}=2^{(r+1)m+1}2^{k_2}$. Now we can bound $|(W, \bar{W}, Y')-(U_m, \bar{W}, Y')|$ when $Y'$ is an independent $(n_2, k'_2)$ source, as follows.

\begin{align*}
&|(W, \bar{W}, Y')-(U_m, \bar{W}, Y')| \\
\leq & \sum_{y \in \Supp(Y')} 2^{-{k'_2}} |(W, \bar{W})|_{Y'=y}-(U_m, \bar{W})|_{Y'=y}| \\ 
=& \sum_{y \in \Supp(Y') \cap B} 2^{-{k'_2}} |(W, \bar{W})|_{Y'=y}-(U_m, \bar{W})|_{Y'=y}| + \sum_{y \in \Supp(Y') \backslash B} 2^{-{k'_2}} |(W, \bar{W})|_{Y'=y}-(U_m, \bar{W})|_{Y'=y}| \\
<& 2^{-{k'_2}}2^{(r+1)m+1}2^{k_2} + 2^{(r+1)m} \e \\
=& 2^{(r+1)m}(2^{k_2-k_2'+1}+\e).
\end{align*}
\end{thmproof}

\iffalse
We'll need the following extension of Vazirani's XOR lemma.  We can't use traditional versions of the XOR lemma, because our output
may not be uniform.  Our statement and proof parallels Rao \cite{Rao07}.

\fi

\subsection{Basic properties of the inner product function}
Here we prove some basic properties of the inner product function.

\BL \label{lem:char1}
Let $\F_p$ be a field and $X, Y$ be two independent random variables over $\F^l_p$. Assume that $X$ has min-entropy $k_1$ and $Y$ has min-entropy $k_2$. Let $Z=\IP(X, Y)=X \cdot Y$ be the inner product function where the operation is in $\F_p$. For any non-trivial character $e_r$ where $r \in \F_p$, 

\[|E_{X, Y}[e_r(Z)]|^2 \leq p^l 2^{-(k_1+k_2)}.\]
\EL

\begin{proof}
Note that if a weak random source $W$ has min-entropy $k$, then $\linfty{W} \leq 2^{-k}$, and $\ltwo{W}^2 = \sum_{w} (\Pr[W=w])^2 \leq 2^{-k} \sum_{w} \Pr[W=w]=2^{-k}$. 

For a fixed $Y=y$,
 
\[E_X[e_r(x \cdot y) ] = E_X [e_{ry}(X)] =\angles{e_{ry}, X}=\overline{\widehat{X}(e_{ry})} . \]

Thus 

\[E_{X, Y}[e_r(Z)]=E_Y [E_X[e_r(x \cdot y) ] ]= E_Y[\overline{\widehat{X}(e_{ry})} ]=\angles{Y, \widehat{X}}.\]

Therefore by Cauchy-Schwartz, 

\begin{align*}
(E_{X, Y}[e_r(Z)])^2 & \leq \angles{Y, Y} \cdot \angles{\widehat{X}, \widehat{X}} \\
 & =\ltwo{Y}^2 \ltwo{\widehat{X}}^2 = p^l \ltwo{Y}^2 \ltwo{X}^2 \\
 & \leq p^l 2^{-k_1} 2^{-k_2} = p^l 2^{-(k_1+k_2)}.
\end{align*}
\end{proof}

Now for any weak random source $W$, we let $2W=W+W$ stand for the distribution that is obtained by first sampling $w_1, w_2$ from two independent and identical distributions according to $W$, and then computing $w_1+w_2$. Similarly $W-W$ is obtained by first sampling $w_1, w_2$ and then computing $w_1-w_2$. Similarly we define $cW$ to be the distribution by sampling $w_i$ from $c$ independent and identical distributions according to $W$, and then computing the sum. We now have the following lemma.

\BL \label{lem:char2}
Let $X, Y$ be two independent random variables over $\F^l_p$. For any two integers $c_1, c_2$, let $X_{c_1}=2^{c_1}X-2^{c_1}X$ and $Y_{c_2}=2^{c_2}Y-2^{c_2}Y$. Then for any non-trivial character $\psi$,

\[|E_{X, Y}[\psi(X \cdot Y)]| \leq |E_{X_{c_1}, Y_{c_2}}[\psi(X_{c_1} \cdot Y_{c_2})]|^{1/2^{c_1+c_2+2}}.\]
\EL

\begin{proof}
First note

\[|E_{X, Y}[\psi(X \cdot Y)]| = |E_Y [E_X [\psi(X \cdot Y) ] ]| \leq E_Y |E_X [\psi(X \cdot Y) ] |. \]

Note that $\psi(s)=e^{2 \pi i rs/p}$ for some $r \in \F_p$. Thus by Jensen's inequality, 

\begin{align*}
(E_{X, Y}[\psi(X \cdot Y)])^2 & \leq E_Y |E_X [\psi(X \cdot Y) ] |^2 = E_Y [E_X [\psi(X \cdot Y) ] \overline{E_X [\psi(X \cdot Y) ]}]\\
& = |E_Y \sum_{x_1, x_2} X(x_1) X(x_2) \psi((x_1-x_2) \cdot Y)| \\
& = |E_Y E_{X-X} [\psi((X-X) \cdot Y)]| \\
& = |E_{X_1, Y} [\psi(X_1 \cdot Y)]|
\end{align*}
where $X_1=X-X$.

Apply the above procedure again, we get that

\[(E_{X, Y}[\psi(X \cdot Y)])^4 \leq |E_{X_1, Y} [\psi(X_1 \cdot Y)]|^2 \leq |E_{X_2, Y} [\psi(X_2 \cdot Y)]|,\]
where $X_2=X_1-X_1=2X-2X$.

Repeat the procedure for $c_1$ times, we get that
\[(E_{X, Y}[\psi(X \cdot Y)])^{2^{c_1+1}} \leq |E_{X_{c_1}, Y} [\psi(X_{c_1} \cdot Y)]|,\]
where $X_{c_1}=2^{c_1}X-2^{c_1}X$.

similarly, we can apply the argument to $Y$ for another $c_2$ times, and we get

\[(E_{X, Y}[\psi(X \cdot Y)])^{2^{c_1+c_2+2}} \leq |E_{X_{c_1}, Y_{c_2}} [\psi(X_{c_1} \cdot Y_{c_2})]|,\]
where $X_{c_1}=2^{c_1}X-2^{c_1}X$ and $Y_{c_2}=2^{c_2}Y-2^{c_2}Y$. Thus the lemma is proved.
\end{proof}

\subsection{Incidence theorems}
We need the following theorems about point line incidences. For a field $\F$, we call a subset $\ell \subset F \times F$ a line if there exist $a, b \in \F$ such that $\ell = \{(x, ax+b)\}$ for all $x \in \F$. Let $P \subset F \times F$ be a set of points and $L$ be a set of lines, we say that a point $(x, y)$ has an incidence with a line $\ell$ if $(x, y) \in \ell$. The following theorem provides a bound on the number of incidences that can be generated from $K$ points and $K$ lines.

\BT \label{thm:incidence} \cite{BourgainKT03, Kon} There exist universal constants $\alpha>0, 0.1>\beta >0$ such that for any field $\F_q$ where $q$ is either prime or $2^p$ for $p$ prime, if $L, P$ are sets of $K$ lines and $K$ points respectively, with $K \leq q^{2-\beta}$, the number of incidences $I(P, L) \leq O(K^{3/2-\alpha})$.  
\ET

\subsection{BCH codes}
In this paper we will only focus on BCH codes over $\F_2$. Given two parameters $m, t \in \N$, a BCH code is a linear code with block length $n=2^m-1$, message length roughly $n-mt$ and distance $d \geq 2t+1$. Specifically, we have the following theorem.

\BT \label{thm:bch}
For all integers $m$ and $t$ there exists an explicit $[n, n-mt, 2t+1]$-BCH code\footnote{In fact, the message length may not be exactly $n-mt$, but for simplicity we will assume that it is exactly $n-mt$. The small error does not affect our analysis. Also, for small $t$ the message length is exactly $n-mt$.}, with $n=2^m-1$.
\ET

Since a BCH code is a linear code, we can take its parity check matrix. Note that this is a $mt \times n$ matrix. Let $\alpha$ be a primitive element in $\F^*_{2^m}$, the $i$'th column of the parity check matrix is of the form $(\alpha^i, (\alpha^i)^3, (\alpha^i)^5, \cdots, (\alpha^i)^{2t-1})$, for $i=0, 1, \cdots, n-1$. Since $\alpha$ is a generator in $\F^*_{2^m}$, equivalently, for $y \in \F^*_{2^m}$ we can think of the $y$'th column to be $(y, y^3, \cdots, y^{2t-1})$.

\section{From Non-Malleable Extractors to Two-Source Extractors} \label{sec:nmtwo}
In this section we show that non-malleable extractors can be used to construct two-source extractors. First we have the following lemmas.

\BL \label{lem:conv} 
Let $X$ be a probability distribution on $\bits^{n_1}$ and $Y, Y_1, \cdots, Y_m$ be probability distributions on $\bits^{n_2}$. Assume that there exists a function $f:\bits^{n_1} \to \bits^{n_2}$ and positive numbers $v_1, \cdots v_m$ with $\sum_i v_i=1$ such that $Y=f(X)$ and $Y=\sum_i v_i Y_i$. Then there exist probability distributions $X_1, \cdots, X_m \in \bits^{n_1}$ such that

\[X=\sum_i v_i X_i \text{ and } Y_i=f(X_i).\]
\EL

\begin{proof}
We define the distributions $\{X_i\}$ as follows. Let $S$ be the support of $Y$. For any $y \in S$, let $S_y=\{x \in \Supp(X): f(x)=y\}$, i.e., $S_y$ is the set of preimages of $y$. Let $p(y)=\Pr[Y=y]$ and $\forall i, p_i(y)=\Pr[Y_i=y]$. Thus we have

\[p(y)=\sum_i v_i p_i(y)\] and 

\[p(y)=\sum_{x \in S_y}\Pr[X=x].\]

Now for any $x \in \Supp(X)$, let $y=f(x)$. Let $\Pr[X_i=x]=q_i(x)=\frac{p_i(y)}{p(y)}\Pr[X=x]$. First note that 

\[\sum_{x \in S_y}q_i(x)=\sum_{x \in S_y}\frac{p_i(y)}{p(y)}\Pr[X=x]=\frac{p_i(y)}{p(y)}\sum_{x \in S_y}\Pr[X=x]=\frac{p_i(y)}{p(y)}p(y)=p_i(y).\]

Thus we have that $Y_i=f(X_i)$. Next note that 

\[\sum_{x}q_i(x) = \sum_y \sum_{x \in S_y} q_i(x)=\sum_y p_i(y)=1.\]

Therefore for any $i$, $X_i$ is indeed a probability distribution. Finally, note that

\[\sum_i v_i q_i(x)=\sum_i \frac{v_i p_i(y)}{p(y)}\Pr[X=x]=\Pr[X=x].\]

Thus we have that $X=\sum_i v_i X_i$.
\end{proof}

\BL \label{lem:decomp}
Let $X$ be a probability distribution on $\bits^{n_1}$. Assume that there exists a function $f:\bits^{n_1} \to \bits^{n_2}$ such that $Y=f(X)$ is an $(n_2, k)$-source. Then there exits an integer $m$ and (flat) $(n_1,k)$-sources $X_1, \cdots, X_m$, bijections $f_1, \cdots, f_m: \bits^{n_1} \to \bits^{n_2}$, positive numbers $v_1, \cdots v_m$ with $\sum_i v_i=1$ such that 

\[X=\sum_i v_i X_i \text{ and } f(X_i)=f_i(X_i).\]
%In other words, $X$ is a convex combination of $X_i$ and $Y$ is the same convex combination of $f_i(X_i)$. 
\EL

\begin{proof}
First, note that every $(n_2, k)$-source is a convex combination of flat $(n_2, k)$-sources. By \lemmaref{lem:conv}, if $Y=f(X)$ and $Y=\sum_i v_i Y_i$, then there exist probability distributions $X_1, \cdots, X_m$ on $\bits^{n_1}$ such that $X=\sum_i v_i X_i$ and $Y_i=f(X_i)$. Thus it suffices to consider only the case where $Y$ is a flat $(n_2,k)$-source.

Now let $Y$ be a flat $(n_2, k)$-source. For any $y \in \Supp(Y)$, let $S_y=\{x \in \Supp(X): f(x)=y\}$. For any $x \in \Supp(X)$, let $p(x)=\Pr[X=x]$ and for any $y \in \Supp(Y)$, let $p(y)=\Pr[Y=y]$. Thus we have $\sum_{x \in S_y} p(x)=p(y)=2^{-k}$. We now decompose $X$ into a convex combination and produce the bijections $f_1, \cdots, f_m$ as follows.

Let $i=1$. While $\cup_y S_y$ is not empty, do the following.

\begin{enumerate}
\item Pick $x$ from $\cup_y S_y$ such that $p(x)$ is the minimum. Assume that $x \in S_y$. Now for all $y' \in \Supp(Y), y' \neq y$, pick an arbitrary $x' \in S_{y'}$. Thus for any $x'$, $p(x') \geq p(x)$. We now let the source $X_i$ be the uniform distribution over the set of the chosen $x$'s, and we let the bijection $f_i$ be such that $f_i(x)=y$ and $f_i(x')=y'$ for all $y' \neq y$. This clearly satisfies the property that $f(X_i)=f_i(X_i)$. Next, we let $v_i$ be the total probability mass of $x$'s, i.e., $v_i=2^k p(x)$. 

\item We now want to subtract the probability mass from both $X$ and $Y$. Thus for any $x'$, we let $p(x')=p(x')-p(x)$ and if $p(x')=0$, remove $x'$ from $S_{y'}$. Specifically, we remove $x$ from $S_y$. Similarly, for any $y \in \Supp(Y)$, let $p(y)=p(y)-p(x)$. Note after this we still have that for any $y \in \Supp(Y)$, $\sum_{x \in S_y} p(x)=p(y)$.

\item Finally, let $i=i+1$.
\end{enumerate}

Note that in the above algorithm, in each iteration at least one element will be removed from $\cup_y S_y$. Thus the algorithm will terminate after finite steps, and we obtain $X_1, \cdots, X_m$, $f_1, \cdots, f_m$ and $v_1, \cdots v_m$. Note that in each iteration the $p(y)$'s are always the same. Thus in each step we can always obtain a flat $(n_1,k)$-source $X_i$ and a bijection $f_i$ such that $f_i(X_i)=f(X_i)$. Note that the algorithm terminates only when $\cup_y S_y$ is empty. Since we always have that for any $y \in \Supp(Y)$, $\sum_{x \in S_y} p(x)=p(y)$, when the algorithm terminates we must have that for any $y \in \Supp(Y)$, $p(y)=0$. Thus we have decomposed $X$ into a convex combination of flat sources $X_1, \cdots, X_m$. In other words, $\sum_i v_i=1$. 
\end{proof}

\BL  \label{lem:error}
Let $X$ be a probability distribution over $\bits^{n_1}$, $Y$ be a probability distribution over $\bits^{n_2}$ and $f: \bits^{n_1} \to \bits^{n_2}$ be any deterministic function. Assume that $|f(X)-Y| \leq \e$ for some $0<\e<1$, then there exists a probability distribution $X'$ over $\bits^{n_1}$ such that 

\[|X'-X| \leq \e \text{ and } Y=f(X').\]

\EL

\begin{proof}
For any $y \in \Supp(Y)$, let $p(y)=\Pr[f(X)=y]$ and $q(y)=\Pr[Y=y]$. Let $S_y=\{x \in \Supp(X): f(x)=y\}$. Thus we have that $p(y)=\sum_{x \in S_y} \Pr[X=x]$. Let $W=\{y \in \Supp(Y): p(y)>q(y)\}$ and $V=\{y \in \Supp(Y): p(y)<q(y)\}$. Thus we have that $\sum_{y \in W} |p(y)-q(y)| = \sum_{y \in V} |p(y)-q(y)| = \e$.

We now gradually change the probability distribution $X$ into $X'$, as follows. First let $X'$ be the same probability distribution as $X$, then, while $W$ is not empty or $V$ is not empty, do the following. 

\begin{enumerate}
\item Pick $y \in W \cup V$ such that $|p(y)-q(y)|=min\{|p(y')-q(y')|, y' \in W \cup V\}$. 

\item If $y \in W$, we decrease $p(y)$ to $q(y)$. Specifically, let $\delta = \tau=p(y)-q(y)$. We pick the elements $x \in S_y$ one by one in an arbitrary order and while $\tau>0$, do the following. Let $\tau'=min(\Pr[X'=x], \tau)$, $\Pr[X'=x]=\Pr[X'=x]-\tau'$ and $\tau=\tau-\tau'$. Note that since $p(y)=\delta+q(y) \geq \delta$, this process will indeed end when $\tau=0$ and now $p(y)=q(y)$. Now to ensure that $X'$ is still a probability distribution, we pick any $\bar{y} \in V$ and increase $p(\bar{y})$ to $p(\bar{y})+ \delta$. To do this, simply pick any $x \in S_{\bar{y}}$ and let $\Pr[X'=x]=\Pr[X'=x]+\delta$. Note that after this change we still have that $p(\bar{y}) \leq q(\bar{y})$. Finally, remove $y$ from $W$ and if $p(\bar{y}) = q(\bar{y})$, remove $\bar{y}$ from $V$.

\item If $y \in V$, we increase $p(y)$ to $q(y)$. Specifically, let $\delta=q(y)-p(y)$. Pick any $x \in S_y$ and let $\Pr[X'=x]=\Pr[X'=x]+\delta$. Now to ensure that $X'$ is still a probability distribution, we pick any $\bar{y} \in W$ and decrease $p(\bar{y})$ to $p(\bar{y})- \delta$. To do this, let $\tau=\delta$. We pick the elements $x \in S_{\bar{y}}$ one by one in an arbitrary order and while $\tau>0$, do the following. Let $\tau'=min(\Pr[X'=x], \tau)$, $\Pr[X'=x]=\Pr[X'=x]-\tau'$ and $\tau=\tau-\tau'$. Note that since $p_{\bar{y}} \geq \delta+q_{\bar{y}}$, this process will indeed end when $\tau=0$ and we still have $p(\bar{y}) \geq q(\bar{y})$. Finally, remove $y$ from $V$ and if $p(\bar{y}) = q(\bar{y})$, remove $\bar{y}$ from $W$.
\end{enumerate}

Note that in each iteration, at least one element will be removed from $W \cup V$. Thus the iteration will end after finite steps. When it ends, we have that $\forall y, p(y)=q(y)$. Thus $f(X')=Y$. Also, it is clear from the algorithm that $|X'-X| =\sum_{y \in W} |p(y)-q(y)| \leq \e$.
\end{proof}

We also need the following definition and theorem about non-malleable extractors with weak random seeds.

\begin{definition}\label{def:weak-seed}\cite{DLWZ11}
A function $\nm:[N] \times [D] \to [M]$ is a $(k,k',\eps)$-non-malleable extractor if,
for any source $X$ with $\hinf(X) \geq k$, any seed $Y$ with $\hinf(Y) \geq k'$,
and any function $\adv:[D] \to [D]$ such that $\adv(y) \neq y$ for all $y$,
the following holds:
\[
(\nm(X,Y),\nm(X,\adv(Y)),Y) \approx_\eps (U_{[M]},\nm(X,\adv(Y)),Y).
\]
\end{definition}

A non-malleable extractor with small error will remain to be non-malleable even if the seed is somewhat weak random.

\BL\cite{DLWZ11} \label{lem:weakseed}
A $(k, \e)$-non-malleable extractor $\nm: \zo^n \times \zo^d \to \zo^m$ is also a $(k, k', \e')$-non-malleable extractor with $\e'=2^{d-k'}\e$.
\EL

Now we can show how non-malleable extractors can be used to construct two-source extractors. In \cite{DW09}, Dodis and Wichs showed that non-malleable extractors for $(n,k)$-sources exist when when $k>2m+3\log(1/\eps) + \log d + 9$ and $d>\log(n-k+1) + 2\log (1/\eps) + 7$. We now show the following theorem.

\BT \label{thm:nmtwo1}
Assume that for any $\e>0$, we have explicit constructions of $(k, \e)$-non-malleable extractors $\nmExt$ with seed length $d=2\log(1/\e)+o(n)$ and output length $m$. Then there exists a constant $\delta>0$ and an explicit construction of two source extractors that take as input an $(n, (1/2-\delta)n)$ source and an independent $(n,k)$ source, and output $m$ bits with error $2^{-\Omega(n)}$.
\ET

\begin{thmproof}
Let $Y$ be an $(n, (1/2-\delta)n)$ source and $X$ be an independent $(n,k)$ source. We construct the two-source extractor $\TExt$ as follows. First we use the somewhere-condenser in \theoremref{thm:tcondenser} to convert $Y$ into a source $\bar{Y}=\scond(Y) $with two rows $(\bar{Y}_1, \bar{Y}_2)$. By \theoremref{thm:tcondenser}, $\bar{Y}$ is $2^{-\Omega(n)}$-close to a somewhere rate-$(1+\alpha)(1/2-\delta)$-source. We choose $\delta>0$ such that $(1+\alpha)(1/2-\delta)=(1/2+\delta)$. Thus now $\bar{Y}$ is $2^{-\Omega(n)}$-close to a somewhere rate-$(1/2+\delta)$-source. Note that each row of $\bar{Y}$ has $l=\Omega(n)$ bits. 

Now let $Y_1=(\bar{Y}_1 \circ 0)$ and $Y_2=(\bar{Y}_2 \circ 1)$. The two-source extractor is defined as

\[\TExt(X, Y)=\nmExt(X, Y_1) \oplus \nmExt(X, Y_2).\]

We now show that this is indeed a two-source extractor for $X$ and $Y$. First note that $\scond(Y)$ is $2^{-\Omega(n)}$-close to a somewhere rate-$(1/2+\delta)$-source. \lemmaref{lem:error} implies that there exists another source $Y'$ such that $|Y-Y'| \leq 2^{-\Omega(n)}$ and $\scond(Y')$ is a somewhere rate-$(1/2+\delta)$-source. In the following analysis we will treat $Y$ as $Y'$, and this will add at most $2^{-\Omega(n)}$ to the error. 

Thus we now have that $\bar{Y}=\scond(Y)$ is a somewhere rate-$(1/2+\delta)$-source, and we want to show that $\TExt(X, Y)$ is close to uniform. Note that a somewhere rate-$(1/2+\delta)$-source is a convex combination of elementary somewhere rate-$(1/2+\delta)$-sources. \lemmaref{lem:conv} now implies that there exist sources $Y^1, \cdots, Y^t$ such that $Y$ is a convex combination of $Y^1, \cdots, Y^t$ and for each $Y^i$, $\scond(Y^i)$ is an elementary somewhere rate-$(1/2+\delta)$-source. Thus we only need to show that for each $Y^i$, $\TExt(X, Y^i)$ is close to uniform. Equivalently and for simplicity, we can assume without loss of generality that $\scond(Y)$ is an elementary somewhere rate-$(1/2+\delta)$-source.

Now, again without loss of generality we assume that $\bar{Y}_1$ is a $(l, (1/2+\delta)l)$-source. Consider the function $f(Y)=\bar{Y}_1$. \lemmaref{lem:decomp} implies that there exist flat $(n, (1/2+\delta)l)$ sources $Y^1, \cdots, Y^t$ and bijections $f_1, \cdots, f_t$ such that $Y$ is a convex combination of $Y^1, \cdots, Y^t$, and for each $i \in [t]$, $f(Y^i)=f_i(Y^i)$. Thus we only need to show that for each $Y^i$, $\TExt(X, Y^i)$ is close to uniform. Consider such a $Y^i$. Since $f(Y^i)=f_i(Y^i)$ and $f_i$ is a bijection, $\bar{Y^i_1}$ is a $(l, (1/2+\delta)l)$-source. Moreover, we can take $g_i$ to be the inverse function of $f_i$, and now $Y^i=g_i(\bar{Y^i_1})$. Note that $\bar{Y^i_2}$ is a deterministic function of $Y^i$, thus we have that now $\bar{Y^i_2}$ is a deterministic function of $\bar{Y^i_1}$. Finally, note that $Y^i_1 \leftrightarrow \bar{Y^i_1}$ and $Y^i_2 \leftrightarrow \bar{Y^i_2}$ are both bijections, we thus have the following claim.

\BCM
$Y^i_1$ is a $(l+1, (1/2+\delta)l)$ source, $Y^i_2=h_i(Y^i_1)$ where $h_i$ is a deterministic function, and $\forall y \in \Supp(Y^i_1), h_i(y) \neq y$.
\ECM

In other words, $Y^i_2$ can be viewed exactly as the seed modified by an adversary in a non-malleable extractor. Now since we are using a non-malleable extractor with seed length $d=l+1=2\log(1/\e)+o(n)=\Omega(n)$ and the seed has min-entropy $(1/2+\delta)l$, by \lemmaref{lem:weakseed} the error of the non-malleable extractor is 

\[2^{d-k'}\e=2^{l+1-(1/2+\delta)l}2^{(o(n)-l-1)/2}=2^{-(\delta l-o(n))}=2^{-\Omega(n)}.\] 

Therefore, we have that $|(\nmExt(X, Y^i_1), \nmExt(X, Y^i_2))-(U_m, \nmExt(X, Y^i_2))| \leq 2^{-\Omega(n)}$. Thus $\TExt(X, Y^i)=\nmExt(X, Y^i_1) \oplus \nmExt(X, Y^i_2)$ is $2^{-\Omega(n)}$-close to uniform. So $\TExt(X, Y)$ is also $2^{-\Omega(n)}$-close to uniform.
\end{thmproof}

We can generalize this theorem to work for sources with smaller min-entropy. For this we need $(r, k, \e)$-non-malleable extractors with $r>1$. We will prove the following theorem in \appendixref{sec:nmexist}.

\BT
For any constant $r \geq 1$, there exists a $(r, k,\eps)$-non-malleable extractor as long as 

\[d > \frac{3}{2}\log(n-k)+3\log(1/\e)+O(1)\]

\[k > (r+1)m+\frac{d}{3}+2 \log (1/\e)+\log(d)+O(1). \]

\ET

We can also define $(r, k, \e)$-non-malleable extractors with weak seed.

\begin{definition}
A function $\bits^n \times \bits^d \to \bits^m$ is a $(r, k,k',\eps)$-non-malleable extractor if,
for any source $X$ with $\hinf(X) \geq k$, any seed $Y$ with $\hinf(Y) \geq k'$,
and any $r$ function $\adv_i:\bits^d \to \bits^d, i=1, \cdots, r$ such that $\adv_i(y) \neq y$ for all~$i$ and $y$,
the following holds:
\[
(\nm(X,Y),\{\nm(X,\adv_i(Y))\},Y) \approx_\eps (U_m,\{\nm(X,\adv_i(Y))\},Y).
\]
\end{definition}

Similarly we have the following lemma.

\begin{lemma}
\label{lem:gweakseed}
A $(r, k,\eps)$-non-malleable extractor $\nm:\bits^n \times \bits^d \to \bits^m$
is also a $(r, k,k',\eps')$-non-malleable extractor with $\eps' = 2^{d-k'} \eps$.
\end{lemma}

\begin{proof}
For $y \in \bits^d$, let $\eps_y = \Delta((\nm(X,y),\{\nm(X,\adv_i(y))\},y), (U_m,\{\nm(X,\adv_i(y))\},y)).$
Then for $Y$ chosen uniformly from $\bits^d$,
\[
\eps \geq \Delta((\nm(X,Y),\{\nm(X,\adv_i(Y))\},Y), (U_m,\{\nm(X,\adv_i(Y))\},Y)) = \frac{1}{2^d} \sum_{y \in \bits^d} \eps_y.
\]
Thus, for $Y'$ with $\hinf(Y') \geq k'$, we get
\begin{align*}
\Delta((\nm(X,Y'),\{\nm(X,\adv_i(Y'))\},&Y'), (U_m,\{\nm(X,\adv_i(Y'))\},Y'))\\
&=\sum_{y \in \bits^d} \Pr[Y=y] \eps_y \leq 2^{-k'}\sum_{y \in \bits^d} \eps_y \leq 2^{d-k'} \eps.
\end{align*}
\end{proof}

Now we have the following theorem.

\BT \label{thm:nmtwo2}
For any constant $b>2$ and any constant $0<\delta<1$, there exists a constant $C=C(\delta)=\poly(1/\delta)$ such that the following holds. Assume that for any $\e>0$ there exists an explicit construction of $(C, k, \e)$-non-malleable extractors $\nmExt$ with seed length $d=b\log(1/\e)+o(n)$ and output length $m$. Then there exists an explicit construction of two source extractors that take as input an $(n, \delta n)$ source and an independent $(n,k)$ source, and output $m$ bits with error $2^{-\Omega(n)}$.
\ET

\begin{thmproof}
Let $Y$ be an $(n, \delta n)$ source and $X$ be an independent $(n,k)$ source. We construct the two-source extractor $\TExt$ as follows. First we use the somewhere-condenser in \theoremref{thm:swcondenser} to convert $Y$ into a source $\bar{Y}=\zuc(Y) $ with $D=\poly(1/\delta)$ rows $(\bar{Y}_1, \cdots, \bar{Y}_D)$ such that $\bar{Y}$ is $2^{-\Omega(n)}$-close to a somewhere rate-$\frac{b}{b+1}$-source. Note that each row of $\bar{Y}$ has $l=\Omega(n)$ bits. 

Now for each row $j$ we let $Y_j$ be $\bar{Y}_j$ concatenated with the binary expression of $j-1$. The two-source extractor is defined as

\[\TExt(X, Y)=\bigoplus_j \nmExt(X, Y_j).\]

Let $C=D-1$. Again, we want to show that $\TExt(X, Y)$ is close to uniform. The proof is similar to the proof in \theoremref{thm:nmtwo1}. Specifically, we can assume that $Y$ is such that $\zuc(Y)$ is indeed a somewhere rate-$\frac{b}{b+1}$-source. This only adds $2^{-\Omega(n)}$ to the error. Next, we can assume that $\zuc(Y)$ is an an elementary somewhere rate-$\frac{b}{b+1}$-source.

Now, without loss of generality we assume that $\bar{Y}_1$ is a $(l, \frac{b}{b+1}l)$-source. Consider the function $f(Y)=\bar{Y}_1$. \lemmaref{lem:decomp} implies that there exist flat $(n, \frac{b}{b+1}l)$ sources $Y^1, \cdots, Y^t$ and bijections $f_1, \cdots, f_t$ such that $Y$ is a convex combination of $Y^1, \cdots, Y^t$, and for each $i \in [t]$, $f(Y^i)=f_i(Y^i)$. Thus we only need to show that for each $Y^i$, $\TExt(X, Y^i)$ is close to uniform. Consider such a $Y^i$. Since $f(Y^i)=f_i(Y^i)$ and $f_i$ is a bijection, $\bar{Y^i_1}$ is a $(l, \frac{b}{b+1}l)$-source. Moreover, we can take $g_i$ to be the inverse function of $f_i$, and now $Y^i=g_i(\bar{Y^i_1})$. Note that for any $j, j \neq 1$, $\bar{Y^i_j}$ is a deterministic function of $Y^i$, thus we have that now $\bar{Y^i_j}$ is a deterministic function of $\bar{Y^i_1}$. Finally, note that for any $j$, $Y^i_j \leftrightarrow \bar{Y^i_j}$ is a bijection, we thus have the following claim.

\BCM
$Y^i_1$ is a $(l+O(1), \frac{b}{b+1}l)$ source, for any $j, j \neq 1$,  $Y^i_j=h_{ij}(Y^i_1)$ where $h_{ij}$ is a deterministic function, and $\forall j, j \neq 1, \forall y, h_{ij}(y) \neq y$.
\ECM

Note that now we have $C$ modified seeds for the non-malleable extractor. Since we are using a non-malleable extractor with seed length $d=l+O(1)=b\log(1/\e)+o(n)=\Omega(n)$ and the seed has min-entropy $\frac{b}{b+1}l$, by \lemmaref{lem:gweakseed} the error of the non-malleable extractor is 

\[2^{d-k'}\e=2^{l+O(1)-\frac{b}{b+1}l}2^{(o(n)-l-O(1))/b} \leq 2^{-(\frac{1}{b(b+1)} l-o(n))}=2^{-\Omega(n)}.\] 

Therefore, we have that $|(\nmExt(X, Y^i_1), \{\nmExt(X, Y^i_j)\})-(U_m, \{\nmExt(X, Y^i_j)\})| \leq 2^{-\Omega(n)}$. Thus $\TExt(X, Y^i)=\bigoplus_j \nmExt(X, Y^i_j)$ is $2^{-\Omega(n)}$-close to uniform. So $\TExt(X, Y)$ is also $2^{-\Omega(n)}$-close to uniform.
\end{thmproof}

\section{From Two-Source Extractors to Non-Malleable Extractors} \label{sec:twonm}
In this section we show how to use a certain kind of two-source extractors to construct non-malleable extractors. %The two-source extractors we will use are those that are constructed based on the inner product function. More specifically, we will consider two-source extractors with the form $\TExt=\IP(f(X), Y)$, where $\IP$ is the inner product function over $\F_2$ and $f(X)$ stands for some function (encoding) of the source $X$.

We first define the following encoding of a string $y \in \bits^d$. 

\BD Given an integer $s$, we choose a BCH code with $t=2$ and $m=s+1$, thus the block length is $n=2^{s+1}-1$ and the parity check matrix is a $mt \times n$ matrix. For any $y \in \bits^s$, let $S_Y$ stand for the integer whose binary expression is $y$. We encode $y$ to $\bar{y}$ such that $\bar{y}$ is the $S_y$'th column in the parity check matrix (i.e., $\Enc(y)=\bar{y}=(y, y^3)$ when $y$ is viewed as an element in $\F^*_{2^{s+1}}$). 

\ED

We have the following theorem.

\BT \label{thm:twonm1}
Assume that we have a two-source extractor $\TExt=\IP(f(X), W)$ such that when given an $(n_1, k)$-source $X$ and an independent $(n_2, n_2/2-\ell)$-source $W$, $\TExt$ outputs 1 bit with error $\e$. Let $n'_2=\lfloor \frac{n_2}{2} \rfloor-1$ and let $Y$ be the uniform distribution over $\bits^{n'_2}$. Define a seeded extractor 

\[\nmExt(X, Y)=\IP(f(X), \Enc(Y)).\]

Then $\nmExt$ is a $(k, \e')$-non-malleable extractor with error $\e'=O(2^{-\ell}+\e)$.
\ET

\begin{thmproof}
First let $Y'$ be a source over $\bits^{n'_2}$ with min-entropy $n_2/2-\ell+1$. Let $\adv: \bits^{n'_2} \to \bits^{n'_2}$ be any deterministic function such that $\forall y, \adv(y) \neq y$.

Note that the BCH code has distance $2t+1=5>4$, thus any 4 columns in the parity check matrix must be linearly independent. This in particular implies that every two different columns must be different. Thus $\Enc(Y')=\bar{Y'}$ has min-entropy $n_2/2-\ell+1$. Therefore by the assumption we have that
\[\nmExt(X, Y') \approx_{\e} U.\]

Next, note that 

\begin{align*}
\nmExt(X, Y') \oplus \nmExt(X, \adv(Y')) &=\IP(f(X), \Enc(Y')) \oplus \IP(f(X), \Enc(\adv(Y'))) \\ &=\IP(f(x), \overline{Y'} + \overline{\adv(Y')}).
\end{align*}

For two different $y_1, y_2$, if $\overline{y_1}+\overline{\adv(y_1)}=\overline{y_2}+\overline{\adv(y_2)}$, then $\overline{y_1}, \overline{\adv(y_1)}, \overline{y_2}, \overline{\adv(y_2)}$ are linearly dependent. Note that $\overline{\adv(y_1)}=\Enc(\adv(y_1))$ and $\overline{\adv(y_2)}=\Enc(\adv(y_2))$ are also some columns of the parity check matrix. Since $\adv(y_1) \neq y_1$ and $\adv(y_2) \neq y_2$, we have that $\overline{\adv(y_1)} \neq \overline{y_1}$ and $\overline{\adv(y_2)} \neq \overline{y_2}$. Thus we must have $\overline{\adv(y_1)}=\overline{y_2}$ and $\overline{\adv(y_2)}=\overline{y_1}$.

Therefore, the min-entropy of $\overline{Y'} + \overline{\adv(Y')}$ is at least  $n_2/2-\ell$ since the probability of getting any particular element in the support is at most $2 \cdot 2^{-(n_2/2-\ell+1)}=2^{-(n_2/2-\ell)}$. Thus by the assumption we have

\[\nmExt(X, Y') \oplus \nmExt(X, \adv(Y')) \approx_{\e} U.\]

Thus by the non-uniform XOR lemma, \lemmaref{special-case}, we have

\[|(\nm(X, Y'), \nm(X, \adv(Y')))-(U, \nm(X, \adv(Y')))| \leq 2\e.\]

Now note that $Y$ has min-entropy $n'_2=\lfloor \frac{n_2}{2} \rfloor-1$, thus by \theoremref{thm:snmext}, 

\[|(\nm(X, Y), \nm(X, \adv(Y)), Y)-(U, \nm(X, \adv(Y)), Y)| \leq 2^2(2^{-(\ell-4)}+2\e)=O(2^{-\ell}+\e).\]
\end{thmproof}

Similarly, we can generalize the above theorem to the case of $(r, k,\eps)$-non-malleable extractors. We have the following definition and theorem.

\BD Given two integers $r, s$, we choose a BCH code with $t=r+1$ and $m=s+1$, thus the block length is $n=2^{s+1}-1$ and the parity check matrix is a $mt \times n$ matrix. For any $y \in \bits^s$, let $S_Y$ stand for the integer whose binary expression is $y$. We encode $y$ to $\bar{y}$ such that $\bar{y}$ is the $S_y$'th column in the parity check matrix (i.e., $\Enc_r(y)=\bar{y}=(y, y^3, \cdots, y^{2r+1})$ when $y$ is viewed as an element in $\F^*_{2^{s+1}}$). 
\ED

\BT \label{thm:twonm2}
Given two integers $r, \ell$ such that $\ell > r$. Assume that we have a two-source extractor $\TExt=\IP(f(X), W)$ such that when given an $(n_1, k)$-source $X$ and an independent $(n_2, n_2/(r+1)-\ell)$-source $W$, $\TExt$ outputs 1 bit with error $\e$. Let $n'_2=\lfloor \frac{n_2}{r+1} \rfloor-1$ and let $Y$ be the uniform distribution over $\bits^{n'_2}$. Define a seeded extractor 

\[\nmExt(X, Y)=\IP(f(X), \Enc_r(Y)).\]

Then $\nmExt$ is a $(r, k, \e')$-non-malleable extractor with error $\e'=O(r2^{r-\ell}+2^{\frac{3r}{2}}\e)$.
\ET

\begin{thmproof}
First let $Y'$ be a source over $\bits^{n'_2}$ with min-entropy $n_2/(r+1)-\ell+\log(r+1)$. Let $\adv_i: \bits^{n'_2} \to \bits^{n'_2}, i=1, \cdots, r$ be $r$ deterministic function such that for any $i$, $\forall y, \adv_i(y) \neq y$.

Note that the BCH code has distance $2t+1=2r+3>2r+2$, thus any $2r+2$ columns in the parity check matrix must be linearly independent. This in particular implies that every two different columns must be different. Thus $\Enc(Y')=\bar{Y'}$ has min-entropy $n_2/(r+1)-\ell+\log(r+1)$. Therefore by the assumption we have that
\[\nmExt(X, Y') \approx_{\e} U.\]

Next, choose any non-empty subset $S \subseteq [r]$, note that

\begin{align*}
\nmExt(X, Y') \oplus \bigoplus_{i \in S} \nmExt(X, \adv_i(Y')) &=\IP(f(X), \Enc(Y')) \oplus \bigoplus_{i \in S} \IP(f(X), \Enc(\adv_i(Y'))) \\ &=\IP(f(x), \overline{Y'} + \sum_{i \in S} \overline{\adv_i(Y')}).
\end{align*}

For two different $y_1, y_2$, if $\overline{y_1}+\sum_{i \in S} \overline{\adv_i(y_1)}=\overline{y_2}+\sum_{i \in S} \overline{\adv_i(y_2)}$, then $\overline{y_1}, \{\overline{\adv_i(y_1)}, i \in S\}, \overline{y_2}, \{\overline{\adv_i(y_2)}, \\ i \in S\}$ are linearly dependent. Without loss of generality we assume that all $\adv_i(y_1)$ are different and all $\adv_i(y_2)$ are different, since if not then some of them sum to 0 and this only decreases the number of items. Note that $\overline{\adv_i(y_1)}=\Enc_r(\adv_i(y_1))$ and $\overline{\adv_i(y_2)}=\Enc_r(\adv_i(y_2))$ are also some columns of the parity check matrix. Since the columns are $2r+2$-wise linearly independent, and the total number of items here is at most $2r+2$, we must have that the items in $\{\overline{y_1}, \{\overline{\adv_i(y_1)}, i \in S\}\}$ and the items in $\{\overline{y_2}, \{\overline{\adv_i(y_2)}, i \in S\}\}$ form a perfect matching, where an edge in the matching is of the form $\overline{y_1}=\overline{\adv_j(y_2)}$, $\overline{\adv_i(y_1)}=\overline{y_2}$ or $\overline{\adv_i(y_1)}=\overline{\adv_j(y_2)}$.

Now we claim that for any $y$, there are at most $r$ different $y_j$'s such that $\overline{y}+\sum_{i \in S} \overline{\adv_i(y)}=\overline{y_j}+\sum_{i \in S} \overline{\adv_i(y_j)}$. To see this, assume for the sake of contradiction that there are $r+1$ such different $y_j$'s. Then by the above discussion we see that for each $j$, there exists an $i \in S$ such that $\overline{y_j}=\overline{\adv_i(y)}$. Since $|S| \leq r$ we must have two different $y_j$ and $y_l$ and an $i \in S$ such that $\overline{y_j}=\overline{\adv_i(y)}=\overline{y_l}$. Note that $\Enc_r$ is injective, thus we have $y_j=y_l$, a contradiction.

Therefore, the min-entropy of $\overline{Y'} + \sum_{i \in S} \overline{\adv_i(Y')}$ is at least $n_2/(r+1)-\ell+\log(r+1)-\log(r+1)=n_2/(r+1)-\ell$. Therefore by the assumption we have that

\[\nmExt(X, Y') \oplus \bigoplus_{i \in S} \nmExt(X, \adv_i(Y')) \approx_{\e} U.\]

Thus by the non-uniform XOR lemma, \lemmaref{lem:noxor2}, we have

\[|(\nm(X, Y'), \{\nm(X, \adv_i(Y'))\})-(U, \{\nm(X, \adv_i(Y'))\})| \leq 2^{\frac{r}{2}}\e.\]

Now note that $Y$ has min-entropy $n'_2=\lfloor \frac{n_2}{r+1} \rfloor-1$, thus by \theoremref{thm:snmext}, 

\begin{align*}
&|(\nm(X, Y),  \{\nm(X, \adv_i(Y))\}, Y)-(U, \nm(X, \{\nm(X, \adv_i(Y))\}, Y)|  \\ &\leq 2^{r+1}(2^{-(\ell-\log(r+1)-3)}+2^{\frac{r}{2}}\e)=O(r2^{r-\ell}+2^{\frac{3r}{2}}\e).
\end{align*}
\end{thmproof}

\section{Improved Constructions of Non-Malleable Extractors} \label{sec:nmext}
Now we can use the above theorems to construct new non-malleable extractors. As a warm up, we first give a new construction of non-malleable extractors for min-entropy rate $>1/2$.

\subsection{A Non-Malleable Extractor for Entropy Rate $> 1/2$}\label{sec:half}
For this purpose, simply notice that the inner product function itself is a two-source extractor for an $(n, (1/2+\delta)n)$ source and another independent source on $n$ bits with min-entropy slightly below $n/2$. Specifically, we have

\BT \cite{ChorG88, Vazirani85} \label{thm:ip} 
For every constant $\delta>0$, if $X$ is an $(n, k_1)$ source, $Y$ is an independent $(n, k_2)$ source and $k_1+k_2 \geq (1+\delta) n$, then 

\[\IP(X, Y) \approx_\e U\]
with $\e=2^{-\Omega(\delta n)}$. 

\ET

Thus we can use the following construction. Given an $(n,k)$-source $X$ with $k = (1/2+\delta)n$, take an independent uniform seed $Y \in \bits^{n/2-1}$ and encode $Y$ to $\bar{Y}$ such that $\bar{Y}=\Enc(Y)=(Y, Y^3)$ when $Y$ is viewed as an element in $\F^*_{2^{n/2}}$. Our non-malleable extractor is now defined as 

\[\nmExt(X, Y)=\IP(X, \Enc(Y))=\IP(X, \bar{Y}),\]
where $\IP$ is the inner product function over $\F_2$.

\BT \label{thm:nmhalf}
For any constant $\delta>0$, the function $\nmExt$ defined as above is a $((1/2+\delta)n, 2^{-\Omega(\delta n)})$ non-malleable extractor.
\ET

\begin{thmproof}
By \theoremref{thm:ip}, $\IP(X, Y)$ is a two-source extractor for an $(n, (1/2+\delta)n)$ source and an independent $(n, (1/2-\delta/2)n)$ source with error $2^{-\Omega(\delta n)}$. Thus by \theoremref{thm:twonm1}, $\nmExt$ is a $((1/2+\delta)n, \e)$ non-malleable extractor with $\e=O(2^{-\delta n/2}+2^{-\Omega(\delta n)})=2^{-\Omega(\delta n)}$.
\end{thmproof}

\subsection{Non-Malleable Extractors for Entropy Rate $<1/2$}\label{sec:below}
In this section we give one of our main constructions, namely a non-malleable extractor for weak sources with min-entropy rate $1/2-\delta$ for some universal constant $\delta>0$. We have the following construction.

Given an $(n,k)$-source $X$ with $k = (1/2-\delta)n$, we first pick a prime $p$ that is close to $n$. By Bertrand's postulate and \cite{BHP01}, there exists $n_0 \in \N$ such that for every $n \geq n_0$, there exists a prime between $n$ and $n+O(n^{0.525})$. We will pick a prime $p$ in this range. Note that the prime can be found in polynomial time in $n$. Take the field $\F_q$ where $q=2^p$ and let $g$ be a generator in $\F^*_q$. The construction is as follows.

\BI
\item Treat $X$ as an element in $\F^*_q$ and encode $X$ such that $\Enc(X)=(X, g^X)$.
\item Take an independent and uniform seed $Y \in \bits^{p-1}$ and encode $Y$ to $\bar{Y}$ such that $\bar{Y}=(Y, Y^3)$ when $Y$ is viewed as an element in $\F^*_{2^p}$.
\item Output $\nm(X, Y)=\IP(\Enc(X), \bar{Y})$ where $\IP$ is the inner product function over $\F_2$.
\EI

To prove our construction is a non-malleable extractor, we are going to use \theoremref{thm:twonm1}. To this end, we first prove the following lemma.

\BL \label{lem:smallerror} There exists a constant $\delta>0$ such that for any $(n,k)$-source $X$ with $k = (1/2-\delta)n$, and any independent $(2p, k_2)$ source $Y$ with $k_2 \geq (1-\delta)p$, 
 \[|\IP(\Enc(X), Y)-U| \leq \e, \]
where $\e=2^{-\Omega(n)}$.
\EL

\begin{proof}
We think of $X$ as a distribution in $\F^*_q$ that has min-entropy $k$. This increases the error by at most $2^{-k}$ (for the element 0). By the XOR lemma, we only need to show that for the only non-trivial character $\psi$ (since we only output 1 bit), 

\[|E_{X, Y}[\psi(\IP(\Enc(X), Y))]| \leq 2^{-\Omega(n)}.\]

Let $X'=4\Enc(X)-4\Enc(X)$, by \lemmaref{lem:char2} we have 

\[|E_{X, Y}[\psi(\IP(\Enc(X), Y))]| \leq |E_{X', Y}[ \psi(X' \cdot Y)]|^{\frac{1}{8}}.\]

We next bound $|E_{X', Y}[ \psi(X' \cdot Y)]|$. First we show that $X'$ is close to a source with min-entropy rate $>1/2$. We have the following claim.

\BCM \label{clm:grow}
There is a universal constant $\delta>0$ such that if $X$ is any weak source with min-entropy $(1/2-\delta) n$, $3 \Enc(X)$ is $2^{-\Omega(n)}$-close to a source with min-entropy $(1/2+\delta) (2p)$.
\ECM

\begin{proof}[Proof of the claim]
Note that $k=(1/2-\delta) n$ and $p$ is between $n$ and $n(1+\frac{1}{2 \ln^2 n})$. Thus for sufficiently large $n$ we have that $k \geq (1/2-1.01\delta)p$. Note that we choose the field $\F_q$ where $q=2^p$. Thus the sum of $\Enc(X)+\Enc(X)$ when viewing $\Enc(X)$ as a vector in $\F^{2p}_2$ is the same as when viewing $\Enc(X)$ as a vector in $\F^2_q$. In the following we will view $\Enc(X)$ as a vector in $\F^2_q$. We show that $3 \Enc(X)$ has a larger min-entropy rate. 

First consider the distribution $2\Enc(X)$. Note that the distribution is of the form $(X+X, g^X+g^X)$. Let $\bar{X}=g^X$ and note that $g^x$ is a bijection in $F^*_q$. Thus $\bar{X}$  has the same min-entropy as X. Now the support of $2\Enc(X)$ is of the form $(\log_g(\bar{x}_1\bar{x}_2), \bar{x}_1+\bar{x}_2)$. For any $(b, a)$ in this support, we have that $\bar{x}_1\bar{x}_2=g^b$ and $\bar{x}_1+\bar{x}_2=a$. Thus there are at most 2 different pairs of $(\bar{x}_1, \bar{x}_2)$ that satisfy both equations. Therefore the min-entropy of $2 \Enc(X)$ is at least $2 H_\infty(X)-1$. We can also assume that $a \neq 0$ since this only increases the error by at most $2^{-H_\infty(X)}$. Now let $k=H_\infty(X)-1$, we have that $\Enc(X)$ has min-entropy at least $k$ and $2\Enc(X)$ has min-entropy at least $2k$.

Now consider $3\Enc(X)$. Every element in the support of $3\Enc(X)$ has the form $(\log_g(\bar{x}_1\bar{x}_2\bar{x}_3), \bar{x}_1+\bar{x}_2+\bar{x}_3)$, which determines the point $(\bar{x}_1\bar{x}_2\bar{x}_3, \bar{x}_1+\bar{x}_2+\bar{x}_3)$. Let $a=\bar{x}_1+\bar{x}_2$ and $b=\bar{x}_1\bar{x}_2$, this point is 

\[(b\bar{x}_3, a+\bar{x}_3).\]

Let $\tilde{x}_3=a+\bar{x}_3$, then 

\[(a+\bar{x}_3, b\bar{x}_3)=(\tilde{x}_3, b\tilde{x}_3-ab).\]

For a fixed $(a=\bar{x}_1+\bar{x}_2, b=\bar{x}_1\bar{x}_2)$ define the line 

\[\ell_{a,b}=\{(x, bx-ab)|x \in \F_q\}.\]

Thus we have a set of lines $L=\{\ell_{a, b}\}$. Note that $a\neq 0$ and $b \neq 0$. Thus for different $(a, b)$, the line $\ell_{a, b}$ is also different. Note that $x_3$ is sampled from $X_3$, which has min-entropy $k$ and $(a, b)$ is sampled from $\Enc(X_1)+\Enc(X_2)$, which has min-entropy $2k$. Further note that these two distributions are independent. Since every weak source with min-entropy $k$ is a convex combination of flat $k$ sources, without loss of generality we can assume that $X_3$ and $\Enc(X_1)+\Enc(X_2)$ are both flat sources. Thus $L$ has size $2^{2k}$. 

Now let $\alpha, \beta$ be the two constants in \theoremref{thm:incidence}. Assume that $3 \Enc(X)$ is $\e$-far from any source with min-entropy $(1+\alpha/2)2k$. Since $3 \Enc(X)$ determines the distribution $(A+\bar{X}_3, B\bar{X}_3)$, this distribution is also $\e$-far from any source with min-entropy $(1+\alpha/2)2k$. Thus there must exist some set $M$ of size at most $2^{(1+\alpha/2)2k}$ such that 

\[\Pr_{(a, b) \leftarrow 2\Enc(X), x_3 \leftarrow X}[(a+\bar{x}_3, b\bar{x}_3) \in M] \geq \e.\]  

Note that whenever $(a+\bar{x}_3, b\bar{x}_3) \in M$, this point has an incidence with the line $\ell_{a,b}$. Further note that whenever $(a, b)$ is different or $x_3$ is different, the incidence is also different. Thus by the above inequality the number of incidences between the set of points $M$ and the set of lines $L$ is  at least

\[\Pr_{(a, b) \leftarrow 2\Enc(X), x_3 \leftarrow X}[(a+\bar{x}_3, b\bar{x}_3) \in M] 2^k 2^{2k} \geq \e2^{3k}.\]

On the other hand, since $L$ has size $2^{2k}$ and $M$ has size $2^{(1+\alpha/2)2k} \leq 2^{(1+\alpha/2)2(1/2-\delta)p} < 2^{(1+\alpha/2)p} \leq q^{2-\beta}$, by \theoremref{thm:incidence}, the number of incidences between $M$ and $L$ is at most $O(2^{(3/2-\alpha)(2+\alpha)k}) < 2^{3k(1-\alpha/6)}=2^{-\alpha k/2}2^{3k}$.

Thus we must have $\e < 2^{-\alpha k/2}$.

Thus we have shown that $3\Enc(X)$ is $2^{-\alpha k/2}$-close to having min-entropy $(1+\alpha/2)2k$. By choosing $\delta$ appropriately, we get that $3\Enc(X)$ is $2^{-\Omega(n)}$-close to having min-entropy $(1/2+\delta)2p$.
\end{proof}

Now note that $Y$ is a weak source over $\bits^{2p}$ with min-entropy $k_2 \geq (1-\delta)p$. Also note that the min-entropy of $X'$ is at least the min-entropy of $3\Enc(X)$. Thus by \lemmaref{lem:char1} we have that

\[|E_{X', Y}[ \psi(X' \cdot Y)]| \leq 2^{2p}2^{-(1/2+\delta)2p}2^{-(1-\delta)p}+2^{-\Omega(n)}=2^{-\Omega(n)}.\]

Therefore \[|E_{X, Y}[\psi(\IP(\Enc(X), Y))]| \leq 2^{-\Omega(n)}.\]
\end{proof}

Now we can prove our construction is a non-malleable extractor.

\BT
For any $(n,k)$-source $X$ with $k = (1/2-\delta)n$, the function $\nm$ defined above is a $(k, \e)$-non-malleable extractor with $\e=2^{-\Omega(n)}$.
%an independent seed $Y$ and any deterministic function $\adv$ such that $\forall y \in \bits^{p-1}, \adv(y) \neq y$, 
%\[|(\nm(X, Y), \nm(X, \adv(Y)), Y)-(U, \nm(X, \adv(Y)), Y)| \leq 2^{-\Omega(n)}\]
\ET

\begin{thmproof}
By \lemmaref{lem:smallerror}, $\IP(\Enc(X), Y)$ is a two-source extractor for an $(n, k)$ source and an independent $(p, (1-\delta)p)$ source with error $2^{-\Omega(n)}$. Therefore by \theoremref{thm:twonm1}, $\nm$ is a $(k, \e)$-non-malleable extractor with error $\e=O(2^{-\delta p}+2^{-\Omega(n)})=2^{-\Omega(n)}$.
\end{thmproof}

\subsection{Achieving Even Smaller Min-Entropy}\label{sec:constant}
In this section we show that we can construct non-malleable extractors for even smaller min-entropy rate (potentially any constant arbitrarily close to 0), if we assume that we have affine extractors with large enough output size, and the Approximate Duality Conjecture (or the Polynomial Freiman-Ruzsa Conjecture) as in \cite{BenSZ11}. 

 Recall the definition of an affine extractor. 

\BD An $[n, m, \rho, \e]$ affine extractor is a deterministic function $f: \bits^n \to \bits^m$ such that whenever $X$ is the uniform distribution over some affine subspace over $\F^n_2$ with dimension $\rho n$, we have that for every $z \in \bits^m$, 

\[|\Pr[f(X)=z]-2^{-m}| < \e.\]
\ED

Now we define the duality measure of two sets as in \cite{BenSZ11}.

\BD  \cite{BenSZ11}
Given two sets $A, B \subseteq \F^n_2$, their duality measure is defined as

\[\mu^{\perp} (A, B) = \left | E_{a \in A, b \in B}[(-1)^{\angles{a,b}}]\right |.\]
\ED

The following conjecture is introduced in \cite{BenSZ11} and is shown in that paper to be implied by the well-known Polynomial Freiman-Ruzsa Conjecture in additive combinatorics.

\begin{conjecture} (Approximate Duality (ADC)) \cite{BenSZ11} \label{conj:adc}
For every pair of constants $\alpha, \delta>0$ there exist a constant $\zeta>0$ and an integer $r$, both depending on $\alpha$ and $\delta$ such that the following holds for sufficiently large $n$. If $A, B \subseteq \F^n_2$ satisfy $|A|, |B| \geq 2^{\alpha n}$ and $\mu^{\perp} (A, B) \geq 2^{-\zeta n}$, then there exists a pair of subsets

\[A' \subseteq A, \text{} A' \geq \frac{|A|}{2^{\delta n+1}} \text{ and } B' \subseteq B, \text{} |B'| \geq \left (\frac{\mu^{\perp} (A, B)}{2} \right )^r \cdot \frac{|B|}{2^{\delta n}}\]
such that $\mu^{\perp} (A', B')=1$.
\end{conjecture}

We now have the following construction.
%We will let $\lambda$ denote the entropy loss rate, i.e., $\lambda=1-\frac{m}{\rho n}$. In this paper we will focus on $[n, (1-\lambda) \frac{2}{3}n, \frac{2}{3}, 2^{-m}]$ affine extractors and ideally we would like $\lambda$ to be as small as possible (e.g., close to 0). We note that it is straightforward to show by the probabilistic method that such extractors exist for any constant $\lambda>0$. However the current state of art constructions only achieve $\lambda \approx \frac{3}{4}$. We now have the following construction.

\begin{construction}
Given any $(n,k)$ source $X$ and a constant $0<\lambda<1$, let $f: \bits^{n'} \to \bits^{m'}$ be an $[n', m'=(1-\lambda) \frac{2}{3}n', \frac{2}{3}, 2^{-m'}]$ affine extractor such that $n=n'-m'$. For any $z \in \bits^{m'}$, let $f^{-1}(z)=\{x: f(x)=z\}$. Then there exists $z \in \bits^{m'}$ such that $|f^{-1}(z)| \geq 2^n$. Let $F: \bits^n \to f^{-1}(z)$ be (any) injective map. Now take an independent uniform seed $Y \in \bits^{n'/2-1}$ and encode $Y$ to $\bar{Y}$ such that $\bar{Y}=\Enc(Y)=(Y, Y^3)$ when $Y$ is viewed as an element in $\F^*_{2^{n'/2}}$. Our non-malleable extractor is now defined as 

\[\nmExt(X, Y)=\IP(F(X), \Enc(Y))=\IP(F(X), \bar{Y}),\]
where $\IP$ is the inner product function taken over $\F_2$. 
\end{construction}

\begin{remark}
Note that here the function $F$ may not be efficiently computable (in time $\poly(n)$). However, the time to compute $F$ is polynomial in the length of the truth table of our final extractor.
\end{remark}

% $W=\nmExt(X, Y)$ and $W'=\nmExt(X, Y')$ where $Y'=\adv(Y)$ such that $\forall y, \adv(y) \neq y$. Again we will show the non-malleability of our extractor by showing that $(W, Y) \approx_{\e} (U, Y)$ and $(W \oplus W', Y) \approx_{\e} (U, Y)$ for some $\e$. 

Again, we will show that our construction is a non-malleable extractor by using \theoremref{thm:twonm1}. To this end, we first show the lemma.

\BL \label{lem:disperser} 
For any $(n,k)$ source $X$ with $k =\frac{2.5 \lambda}{1+2\lambda} n$ and any independent $(n', \frac{n'}{3}+1)$ source $Y$, $\IP(F(X), Y)$ is non-constant.
\EL

\begin{proof}
As usual we can assume without loss of generality that $X$ and $Y$ are flat sources. If $\IP(F(X), Y)$ is a constant, then $\Supp(F(X))$ and $\Supp(Y)$ must be contained in two affine subspaces with dimension $d_1, d_2$ such that $d_1+d_2 \leq n'$. Note that  $d_2 >\frac{n'}{3}$ since $Y$ has min-entropy $\frac{n'}{3}+1$. We next show that $d_1 > \frac{2}{3}n'$ and thus reach a contradiction. 

To see this, let $S=\Supp(F(X))$. It suffices to show that $S$ is not contained in any affine subspace of dimension $\frac{2}{3}n'$. Let $A$ be such an affine subspace. We have
\[|A \cap S| \leq |A \cap f^{-1}(z)| < 2 \cdot 2^{-m'} 2^{\frac{2}{3}n'} =2^{\frac{2\lambda}{3}n'+1},\]
where the last inequality follows from the fact that $f$ is an affine extractor. Now note that $|S| = 2^{\frac{2.5 \lambda}{1+2\lambda} n}=2^{\frac{2.5\lambda}{3}n'}$. Thus we have that $|A \cap S| < |S|$ and therefore $S$ cannot be contained in $A$.
\end{proof}

%Now we define the duality measure of two sets as in \cite{BenSZ11}.

%\BD  \cite{BenSZ11}
%Given two sets $A, B \subseteq \F^n_2$, their duality measure is defined as

%\[\mu^{\perp} (A, B) = \left | E_{a \in A, b \in B}[(-1)^{\angles{a,b}}]\right |.\]
%\ED

%The following conjecture is introduced in \cite{BenSZ11} and is shown in that paper to be implied by the Polynomial Freiman-Ruzsa Conjecture in additive combinatorics.

Now we have the following lemma. 

\BL \label{lem:conext}
There exists a constant $\zeta=\zeta(\lambda)$ such that for any $(n,k)$ source $X$ with $k =\frac{3 \lambda}{1+2\lambda} n$ and any independent $(n', \frac{7n'}{15})$ source $Y$, $\IP(F(X), Y)$ is $2^{-\zeta n}$-close to uniform.
\EL

\begin{proof}
Let $v=\frac{n}{n'}=\frac{1+2\lambda}{3}$, $\alpha=min\{\lambda, \frac{1}{3}\}$ and $\delta = \frac{\lambda}{8}$. Let $\zeta'$ and $r$ be the constant and the integer guaranteed by conjecture~\ref{conj:adc} for $\alpha$ and $\delta$. Let $\zeta=min \{\frac{\zeta'}{v}, \frac{1-\lambda}{8r}\}$. We will prove the lemma by way of contradiction.

Let $X$ and $Y$ be two independent sources as in the statement of the lemma. Again we assume without loss of generality that both $X$ and $Y$ are flat sources. Let $A=\Supp(X)$ and $\bar{A}=\{F(a)| a \in A\} \subseteq \F^{n'}_2$. Let $\bar{B}=\Supp(Y) \subseteq \F^{n'}_2$. Note that $F$ is an injective function. Thus $|\bar{A}| =2^{\frac{3\lambda}{1+2\lambda}n} = 2^{\lambda n'} \geq 2^{\alpha n'}$ and $|\bar{B}| =2^{\frac{7n'}{15}} > 2^{\frac{n'}{3}} \geq 2^{\alpha n'}$. 

Assume for the sake of contradiction that the error of $\IP(F(X), Y)$, which is equal to $\frac{1}{2} \mu^{\perp} (\bar{A}, \bar{B})$, is greater than $2^{-\zeta n} \geq 2^{-\zeta' n'}$. Then by the ADC conjecture (conjecture~\ref{conj:adc}) there exist $A' \subseteq \bar{A}$ and $B' \subseteq \bar{B}$ such that

\[|A'| \geq \frac{|\bar{A}|}{2^{\delta n'+1}} \geq  2^{\frac{5\lambda}{6} n'} = 2^{\frac{2.5 \lambda}{1+2\lambda} n} \text{ and } |B'| \geq \frac{|\bar{B}|}{2^{\delta n'+r\zeta n}} \geq \frac{2^{\frac{7n'}{15}}}{2^{\frac{n'}{8}}}> 2^{\frac{n'}{3}+1},\]
and $\mu^{\perp} (A', B')=1$.

Let $A''$ be the preimages of $A'$ under $F$. Since $F$ is injective, we must have $|A''| \geq 2^{\frac{2.5 \lambda}{1+2\lambda} n}$. Thus if we let $X'$ and $Y'$ be the uniform distribution over $A''$ and $B'$ respectively, we get two independent sources that satisfy the conditions in \lemmaref{lem:disperser}. However $\IP(F(X'), Y')$ is a constant, which contradicts \lemmaref{lem:disperser}. Thus we must have that $\IP(F(X), Y)$ is $2^{-\zeta n}$-close to uniform.
\end{proof}

Now we can prove the following theorem.

\BT \label{thm:nmconst}
$\nm$ is a $(k, \e)$-non-malleable extractor with $k =\frac{3 \lambda}{1+2\lambda} n$, seed length $d=\frac{3}{2+4\lambda}n-1$ and $\e=2^{-\Omega(n)}$.

%\[|(W, W', Y)-(U, W', Y)| \leq 2^{-\Omega(\zeta n)}.\]
\ET

\begin{thmproof}
The seed length is clearly $d=n'/2-1=\frac{3}{2+4\lambda}n-1$. Let $\zeta=\zeta(\lambda)$ be as in \lemmaref{lem:conext}. By \lemmaref{lem:conext}, $\IP(F(X), Y)$ is a two-source extractor for an $(n, k =\frac{3 \lambda}{1+2\lambda} n)$ source and an independent $(n', \frac{7n'}{15})$ source with error $2^{-\zeta n}$. Thus by \theoremref{thm:twonm1}, $\nm$ is a $(k, \e)$-non-malleable extractor with error $\e=O(2^{-n'/30}+2^{-\zeta n})=2^{-\Omega(n)}$.
\end{thmproof}

Similarly, we can generalize our construction to $(r, k, \e)$-non-malleable extractors. We have the following construction and theorem.

\begin{construction}
Given a constant integer $r$, any $(n,k)$ source $X$ and a constant $0<\lambda<1$, let $f: \bits^{n'} \to \bits^{m'}$ be an $[n', m'=(1-\lambda) \frac{r+1}{r+2}n', \frac{r+1}{r+2}, 2^{-m'}]$ affine extractor such that $n=n'-m'$. For any $z \in \bits^{m'}$, let $f^{-1}(z)=\{x: f(x)=z\}$. Then there exists $z \in \bits^{m'}$ such that $|f^{-1}(z)| \geq 2^n$. Let $F: \bits^n \to f^{-1}(z)$ be (any) injective map. Now take an independent uniform seed $Y \in \bits^{n'/(r+1)-1}$ and encode $Y$ to $\bar{Y}$ such that $\bar{Y}=\Enc_r(Y)=(Y, Y^3, \cdots, Y^{2r+1})$ when $Y$ is viewed as an element in $\F^*_{2^{n'/(r+1)}}$. Define a seeded extractor

\[\nmExt(X, Y)=\IP(F(X), \Enc_r(Y))=\IP(F(X), \bar{Y}),\]
where $\IP$ is the inner product function taken over $\F_2$. 
\end{construction}

\BT
$\nm$ is a $(r, k, \e)$-non-malleable extractor with $k =\frac{(r+2) \lambda}{1+(r+1)\lambda} n$, seed length $d=\frac{r+2}{r+1+(r+1)^2\lambda}n-1$ and $\e=2^{-\Omega(n)}$.
\ET

By using \theoremref{thm:twonm2} instead of \theoremref{thm:twonm1}, the proof of this theorem is very similar to the proof of \theoremref{thm:nmconst}. We thus omit the proof here. 

\subsection{Increasing the Output Size and Reducing the Seed Length}\label{sec:seed}
In this section we show that we can increase the output size and reduce the seed length  for the constructions in \subsectionref{sec:half}, \subsectionref{sec:below} and \subsectionref{sec:constant}. All these constructions share the same pattern: the seed $Y$ is encoded using the parity check matrix of a BCH code, and then the output is the inner product function of the encoded source and the encoded seed over $\F_2$.

We only discuss the construction in \subsectionref{sec:below}, but the method can be applied to all the other constructions in the same way. We start by showing how to increase the output size to $m=\Omega(n)$.

\subsubsection{Increasing the output size}
Recall that in the construction we used a field $\F_{2^p}$ for a prime $p$. Given the finite filed $\F_{2^p}$, the elements of this field form a vector space of dimension $p$ over $\F_2$. Let $b_1, \cdots, b_p \in \F_{2^p}$ be a basis for this vector space. Now recall that in the construction we encode the seed $Y$ to $\bar{Y}=(Y, Y^3)$, when viewing $Y$ as an element in $\F^*_{2^p}$. Now for each $b_i$, let $\bar{Y^i}=(b_i Y, b_i Y^3)$ and define one bit $Z_i=\IP(\Enc(X), \bar{Y^i})$. We now show that $\{Z_i\}$ satisfy the conditions of a non-uniform XOR lemma.

\BL \label{lem:xoroutput}
Given any $(n,k)$-source $X$ with $k=(1/2-\delta)n$ and an independent seed $Y \in \bits^{p-1}$ with min-entropy $(1-\delta)p+2$, let $\adv: \bits^{p-1} \to \bits^{p-1}$ be any deterministic function such that $\forall y, \adv(y) \neq y$. For any $i$, let $Z'_i=\IP(\Enc(X), \bar{Y^{i'}})$, where $\bar{Y^{i'}}=(b_i Y', b_i Y'^3)$ and $Y'=\adv(Y)$. Then for any non-empty subset $S_1 \subseteq [p]$ and any subset $S_2 \subseteq [p]$, we have that

\[|\bigoplus_{i \in S_1} Z_i \oplus \bigoplus_{j \in S_2} Z'_j-U | \leq 2^{-\Omega(n)}.\]
\EL

\begin{proof}
Note that 
\[\bigoplus_{i \in S_1} Z_i = \IP(\Enc(X), \sum_{i \in S_1} \bar{Y^i})=\IP(\Enc(X), t_1(Y, Y^3) ),\] 
where $t_1=\sum_{i \in S_1} b_i \in \F_{2^p}$, and

\[\bigoplus_{j \in S_2} Z'_j = \IP(\Enc(X), \sum_{j \in S_2} \bar{Y^{j'}})=\IP(\Enc(X), t_2(Y', Y'^3) ),\] 
where $t_2=\sum_{j \in S_2} b_j \in \F_{2^p}$.

Thus

\[\bigoplus_{i \in S_1} Z_i \oplus \bigoplus_{j \in S_2} Z'_j=\IP(\Enc(X), \tilde{Y}),\]
where $\tilde{Y}=t_1(Y, Y^3)+t_2(Y', Y'^3)$.

We now bound the min-entropy of $\tilde{Y}$ and have the following claim.

\BCM
$H_{\infty}(\tilde{Y}) > (1-\delta)p$.
\ECM

\begin{proof}
We have two cases.

\textbf{Case 1:} $S_2 = \phi$. In this case $\tilde{Y}=t_1(Y, Y^3)$. Since $S_1 \neq \phi$, we have $t_1 \neq 0$. Thus $\tilde{Y}$ has the same min-entropy as $Y$, which is $(1-\delta)p+2>(1-\delta)p$. 

\textbf{Case 2:} $S_2 \neq \phi$. In this case we have $t_1 \neq 0$ and $t_2 \neq 0$. We need to bound the min-entropy of $\tilde{Y}=t_1(Y, Y^3)+t_2(Y', Y'^3)$. Again, if for every two different $y_1, y_2$, we have $t_1(y_1, y_1^3)+t_2(y_1', y_1'^3) \neq t_1(y_2, y_2^3)+t_2(y_2', y_2'^3)$, then $\tilde{Y}$ will have the same min-entropy of $Y$. We now show that any element in $\Supp(\tilde{Y})$ can come from at most 3 different elements in $\Supp(Y)$.

To show this, assume for the sake of contradiction that there are 4 different $y_1, y_2, y_3, y_4$ such that $t_1(y_i, y_i^3)+t_2(y_i', y_i'^3)$ are the same for $i=1, 2, 3, 4$. First consider $y_1, y_2$, we have $t_1(y_1, y_1^3)+t_2(y_1', y_1'^3) = t_1(y_2, y_2^3)+t_2(y_2', y_2'^3)$. Since $t_1 \neq 0$, let $r=t_2/t_1 \in \F_{2^p}$. Thus $r \neq 0$ and we have $(y_1, y_1^3)+r(y_1', y_1'^3) = (y_2, y_2^3)+r(y_2', y_2'^3)$. We first consider the case where $r =1$. In this case, the vectors $(y_1, y_1^3)$, $(y_1', y_1'^3)$, $(y_2, y_2^3)$ and $(y_2', y_2'^3)$ are linearly dependent over $\F_2$. However we know that the columns of the parity check matrix of the BCH code are 4-wise linearly independent. Thus we must have $y_1'=y_2$ and $y_2'=y_1$. Thus in this case the element in $\Supp(\tilde{Y})$ comes from at most 2 different elements in $\Supp(Y)$. Now if $r \neq 1$, we have

\[y_1 +r y_1' = y_2 +r y_2'\] and

\[(y_1)^3 +r (y_1')^3 = (y_2)^3 +r (y_2')^3.\]

Hence we get

\[ y_1-y_2 = r (y_2'-y_1')\] and 

\[ (y_1^2+y_1y_2+y_2^2)(y_1-y_2) = r(y_2'^2+y_1'y_2'+y_1'^2)(y_2'-y_1').\]

Since $y_1 \neq y_2$ and $r \neq 0$, we must have that $y_1' \neq y_2'$. Thus we get 

\[y_1^2+y_1y_2+y_2^2 = y_2'^2+y_1'y_2'+y_1'^2.\]

Similarly we can get 

\[y_1^2+y_1y_3+y_3^2 = y_3'^2+y_1'y_3'+y_1'^2.\]

Thus

\[(y_1+y_2+y_3)(y_2-y_3)=(y_1'+y_2'+y_3')(y_2'-y_3').\]

Also, from $y_2 +r y_2' = y_3 +r y_3'$ we get 

\[y_2-y_3=r(y_3'-y_2').\]

Since $y_2 \neq y_3$, we have

\[y_1'+y_2'+y_3'=-r(y_1+y_2+y_3).\]

Similarly we can get 

\[y_1'+y_2'+y_4'=-r(y_1+y_2+y_4).\]

Therefore

\[y_4'-y_3'=r(y_3-y_4).\]

On the other hand, from $y_3 +r y_3' = y_4 +r y_4'$ we get

\[y_4'-y_3'=1/r(y_3-y_4).\]

Thus

\[(r^2-1)(y_3-y_4)=0.\] 

Since $r \neq 1$, $r^2-1 \neq 0$. Thus we have $y_3=y_4$, a contradiction.

Therefore, the min-entropy of $\tilde{Y}$ is at least $H_{\infty}(Y)-\log 3=(1-\delta)p+2-\log 3>(1-\delta)p$. 
\end{proof}

Now, by \lemmaref{lem:smallerror}, the lemma follows.
\end{proof}

Now we have the following theorem.

\BT 
There exists a constant $0< \delta<1$ such that for any $n \in \N$, $k = (1/2-\delta)n$, there exists an explicit $(k, \e)$-non-malleable extractor $\nm: \bits^n \times \bits^n \to \bits^m$ with $m=\Omega(n)$ and $\e=2^{-\Omega(n)}$.
\ET

\begin{thmproof}
By \lemmaref{lem:xoroutput} and \lemmaref{special-case}, we can choose $m=\Omega(n)$ bits from $\{Z_i\}$ such that when we have $\nm$ output $Z_1 \circ \cdots \circ Z_m$, we get 

\[|(\nm(X, Y), \nm(X, \adv(Y)))-(U_m, \nm(X, \adv(Y)))| \leq 2^{-\Omega(n)}.\]

Note that in \lemmaref{lem:xoroutput} the seed $Y$ only has min-entropy $(1-\delta)p+2$. Thus if we use a uniform seed $Y  \in \bits^{p-1}$, by \theoremref{thm:snmext} we have that 

\[|(\nm(X, Y), \nm(X, \adv(Y)), Y)-(U_m, \nm(X, \adv(Y)), Y)| \leq \e,\]
where $\e=2^{2m}(2^{-\delta p+4}+2^{-\Omega(n)})=2^{-\Omega(n)}$ when $m=\Omega(n)$ and is small enough.
\end{thmproof}

\subsubsection{Reducing the seed length}
In the constructions mentioned above, we use a BCH code with distance 5. Thus the columns of the parity check matrix are 4-wise linearly independent. To reduce the seed length, we are going to use a BCH code with larger distance. Specifically, we will choose a $[2^{\ell}-1, 2^{\ell}-1-2t{\ell}, 4t+1]$-BCH code with ${\ell}=p/t$ for some parameter $t$ to be chosen later. Note that the parity check matrix is a $2p \times (2^{\ell}-1)$ matrix\footnote{Actually $p$ is not divisible by $t$, thus $\ell t <p$. However for simplicity we will assume that the matrix has $2p$ rows. For example we can add $0$'s in the end, the small error does not affect our analysis.}. Thus the columns of the matrix are $D=4t$-wise linearly independent. The detailed construction is as follows.

\BI
\item Given an $(n,k)$-source $X$ with $k = (1/2-\delta)n$, pick a prime $p$ such that $n \leq p \leq n(1+\frac{1}{2 \ln^2 n})$.
\item Let $q=2^p$ and $g$ be a generator in $\F^*_q$. Treat $X$ as an element in $\F^*_q$ and encode $X$ such that $\Enc(X)=(X, g^X)$.
\item Let ${\ell}=p/t$. Take the parity check matrix of a $[2^{\ell}-1, 2^{\ell}-1-2t{\ell}, 4t+1]$-BCH code. Note that it is a $2p \times (2^{\ell}-1)$ matrix. Take an independent and uniform seed $Y \in \bits^{{\ell}-1}$ and let $S_Y$ stand for the integer whose binary expression is $Y$. We encode $Y$ to $\bar{Y}$ such that $\bar{Y}$ is the $S_Y$'th column in the parity check matrix.
\item Output $\nm(X, Y)=\IP(\Enc(X), \bar{Y})$ where $\IP$ is the inner product function taken over $\F_2$.
\EI

As in \subsectionref{sec:below}, we have \claimref{clm:grow}. We now want to argue about the min-entropy of $t\bar{Y}$ and $t(\bar{Y}+\bar{\adv(Y)})$.

\BL
Assume $Y$ has min-entropy $k_2$, then $t\bar{Y}$ is $t^22^{-(k_2+1)}$-close to having min-entropy $t(k_2-\log t)$, and $t(\bar{Y}+\bar{Y'})$ is $t^22^{-(k_2+2)}+t(t2^{-\frac{2}{3}k_2})^{\log t}$-close to having min-entropy $t((1-\frac{\log t}{3t})k_2-3\log t)$.
\EL

\begin{proof}
Without loss of generality assume that $Y$ is a flat source. Let $K=2^{k_2}$. First consider $t\bar{Y}$. Note that $\bar{Y}$ has the same min-entropy as $Y$ and is also a flat source, since every two columns of the parity check matrix are different. The support of $t\bar{Y}$ has the form $\bar{y}_1+ \cdots +\bar{y}_t$. Consider the case where all $\bar{y}_i$'s are different. This takes up a probability mass of 

\[\frac{K!}{(K-t)!} \cdot K^{-t}=1 \cdot (1-\frac{1}{K}) \cdots (1-\frac{t-1}{K}) > 1-\sum_{i=1}^{t-1}\frac{i}{K}>1-\frac{t^2}{2K}.\]

Since the columns of the parity check matrix are $4t$-wise linearly independent. For every two different sets $\{\bar{y}_i\}$'s, their sum cannot be the same. Therefore, the probability mass of getting a particular value is at most $t!K^{-t} \leq 2^{-t(k_2-\log t)}$. Thus $t\bar{Y}$ is $t^22^{-(k_2+1)}$-close to having min-entropy $t(k_2-\log t)$.

Next consider $t(\bar{Y}+\bar{\adv(Y)})$. Let $\adv(Y)=Y'$ and $Y''=\bar{Y}+\bar{Y'}$. Note that for every $s \in \Supp(Y'')$, $s \neq 0$ since $\forall y, \adv(y) \neq y$. Also note that $Y''$ has min-entropy at least $k_2-1$ since if $\bar{y}_1+\bar{y'}_1=\bar{y}_2+\bar{y'}_2$ for $y_1 \neq y_2$, then we must have $\bar{y'}_1=\bar{y}_2$ and $\bar{y'}_2=\bar{y}_1$. Without loss of generality assume that $Y''$ is a flat source with min-entropy $k_2-1$. Let $K_2=2^{k_2-1}$. Note that now in the support of $Y''$ there are no two different $y_1, y_2$ such that $\bar{y}_1+\bar{y'}_1=\bar{y}_2+\bar{y'}_2$ (since this will be absorbed into the same element).

We now consider $tY''$. An element in its support has the form $\sum_{i=1}^{t}(\bar{y}_i+\bar{y'}_i)$. We first get rid of those elements in $\Supp(tY'')$ such that some of the $\{\bar{y}_i+\bar{y'}_i\}$'s are the same. By the same argument as above this takes up a probability mass of at most $\frac{t^2}{2K_2}$. Now, for a particular set $\{\bar{y}_i+\bar{y'}_i\}_{i \in [t]}$, we consider how many different sets can have the same sum. 

Since the columns of the parity check matrix are $4t$-wise linearly independent, if the sum of two different sets $\{\bar{y}_i+\bar{y'}_i\}_{i \in [t]}$ are the same, then except those $\bar{y}_i+\bar{y'}_i$'s that are common in both sets, the rest of $\bar{y}_i+\bar{y'}_i$'s must form cycles. By cycle we mean a set of $l$ elements such that $\bar{y'}_1=\bar{y}_2, \bar{y'}_2=\bar{y}_3, \cdots, \bar{y'}_l=\bar{y}_1$ so that the sum is 0. Note that $l \geq 3$ since the support of $Y''$ has no 2-cycles. Let $S_1, S_2$ be the two sets $\{\bar{y}_i+\bar{y'}_i\}_{i \in [t]}$. Now, the elements in a cycle can come from both sets or just from one set. If the elements from a cycle comes only from $S_2$, then this cycle can be replaced by any other cycle with the same length, and the sum of $S_1$ and $S_2$ are still the same. On the other hand, if the elements of a cycle comes from both $S_1$ and $S_2$, then the elements in this cycle are completely determined by $S_1$ since cycles are disjoint. Therefore, let $r$ be the number of common elements in $S_1, S_2$, and let $l$ be the total length of cycles whose elements only come from the rest elements of $S_2$, and note that cycles have length at least 3, we have that if $l \geq \log t$, then the total probability mass of these elements in $\Supp(tY'')$ is at most

\[\sum_{\log t \leq l \leq t} \binom{t}{l}\binom{\frac{K_2}{3}}{\frac{l}{3}}l!(K_2)^{-l} \leq \sum_{\log t \leq l \leq t} t^l (\frac{K_2}{3})^{\frac{l}{3}}(K_2)^{-l} < \sum_{\log t \leq l \leq t}(t(K_2)^{-\frac{2}{3}})^l < t(t2^{-\frac{2}{3}k_2})^{\log t}.\]

On the other hand, if $l < \log t$, then the probability that $tY''$ gets a particular value is at most

\[\sum_{0 \leq l \leq \log t}\sum_{0 \leq r \leq t-l} \binom{t}{l} \binom{t-l}{r}\binom{\frac{K_2}{3}}{\frac{l}{3}}t! (K_2)^{-t}< t\log t \cdot t^{2t} (\frac{K_2}{3})^{\frac{\log t}{3}}(K_2)^{-t}< t\log t(t^2 (K_2)^{-(1-\frac{\log t}{3t})})^t.\]

Thus the min-entropy is at least $t((1-\frac{\log t}{3t})k_2-3\log t)$.
\end{proof}

Now for an $(n, k)$-source $X$ with $k=(1/2-\delta)n$, we know that $3\Enc(X)$ is $2^{-\Omega(n)}$-close to having min-entropy $(1/2+\delta) (2p)$. Assume that we want our non-malleable extractor to have error $\e \leq 1/n$. We'll choose a parameter $t<n/C\log n$ for a sufficiently large constant $C>1$. When $Y$ is uniform over $\ell=p/t$ bits, $t\bar{Y}$ is close to having min-entropy $t(k_2-\log t)> (1-1/C)p>(1/2-\delta/2) (2p)$, and $t(\bar{Y}+\bar{Y'})$ is close to having min-entropy $t((1-\frac{\log t}{3t})k_2-3\log t)> (1-1/C)p>(1/2-\delta/2) (2p)$. When $t\bar{Y}$ and $t(\bar{Y}+\bar{Y'})$ indeed have this min-entropy, by \lemmaref{lem:char1} we have that both $\IP(3\Enc(X), t\bar{Y})$ and $\IP(3\Enc(X), t(\bar{Y}+\bar{Y'}))$ are $2^{-\Omega(n)}$-close to uniform. Thus we can take $t=\Omega(n/(\log (1/\e)))$ and by \lemmaref{lem:char2} and \theoremref{thm:snmext} we have that the error of the non-malleable extractor is at most $\e$, and the seed length is roughly $p/t = O(n/t) =O(\log(1/\e))$. Thus we have the following theorem.

\BT
There exists a universal constant $\delta>0$ such that for every $n \in \N$ and $\e$ such that $2^{-\Omega(n)} \leq \e \leq 1/\poly(n)$, there exists an explicit $(k, \e)$ non-malleable extractor $\nm: \bits^n \times \bits^d \to \bits$ for $k=(1/2-\delta)n$ and seed length $d=O(\log n+\log (1/\e))$. 
\ET

\section{An Optimal Privacy Amplification Protocol for Arbitrarily Linear Entropy}\label{sec:protocol}
In this section we present our privacy amplification protocol for $(n,k)$-sources $X$ with $k=\delta n$ for any constant $\delta>0$. Following \cite{kr:agree-close} and \cite{DLWZ11}, we define a privacy amplification protocol $(P_A, P_B)$. The protocol is executed by two parties Alice and Bob, who share a secret $X\in \bits^n$. An active, computationally unbounded adversary Eve might have some partial information $E$ about $X$ satisfying $\thinf(X|E)\ge k$. Since Eve is unbounded, we can assume without loss of generality that she is deterministic.
Informally, we want the protocol to be such that whenever a party (Alice or Bob) does not reject, the key $R$ output by this party is random and independent of Eve's view. Moreover, if both parties do not reject, they must output the same keys $R_A=R_B$ with overwhelming probability.

More formally, we assume that Eve has full control of the communication channel between the two parties. This means that Eve can arbitrarily insert, delete, reorder or modify messages sent by Alice and Bob to each other. In particular, Eve's strategy $P_E$ defines two correlated executions $(P_A,P_E)$ and $(P_E,P_B)$ between Alice and Eve, and Eve and Bob, called ``left execution'' and ``right execution'', respectively. Alice and Bob are assumed to have fresh, private and independent random bits $Y$ and $W$, respectively. $Y$ and $W$ are not known to Eve. In the protocol we use $\perp$ as a special symbol to indicate rejection. At the end of the left execution $(P_A(X,Y),P_E(E))$, Alice outputs a key $R_A\in \bits^m \cup \{\perp\}$. Similarly, Bob outputs a key $R_B \in \bits^m \cup \{\perp\}$ at the end of the right execution $(P_E(E),P_B(X,W))$. We let $E'$ denote the final view of Eve, which includes $E$ and the communication transcripts of both executions $(P_A(X,Y),P_E(E))$ and $(P_E(E),P_B(X,W)$. We can now define the security of $(P_A,P_B)$. 

\BD An interactive protocol $(P_A, P_B)$, executed by Alice and Bob on a communication channel fully controlled by an active adversary Eve, is a $(k, m, \e)$-\emph{privacy amplification protocol} if it satisfies the following properties whenever $\thinf(X|E) \geq k$:
\begin{enumerate}
\item \underline{Correctness.} If Eve is passive, then $\Pr[R_A=R_B \land~ R_A\neq \perp \land~ R_B\neq \perp]=1$.
\item \underline{Robustness.} We start by defining the notion of {\em pre-application} robustness, which states that even if Eve is active, $\Pr[R_A \neq R_B \land~ R_A \neq \perp \land~ R_B \neq \perp]\le \e$.

The stronger notion of {\em post-application} robustness is defined similarly, except Eve is additionally given the key $R_A$ the moment she completed the left execution $(P_A,P_E)$, and the key $R_B$ the moment she completed the right execution $(P_E,P_B)$. For example, if Eve completed the left execution before the right execution, she may try to use $R_A$ to force Bob to output a different key $R_B\not\in\{R_A,\perp\}$, and vice versa.
\item \underline{Extraction.} Given a string $r\in \bits^m\cup \{\perp\}$, let $\purify(r)$ be $\perp$ if $r=\perp$, and otherwise replace $r\neq \perp$ by a fresh $m$-bit random string $U_m$:  $\purify(r)\leftarrow U_m$. Letting $E'$ denote Eve's view of the protocol, we require that

\begin{align*}
\Delta((R_A, E'),(\purify(R_A), E')) \leq \e
~~~~\mbox{and}~~~~ 
\Delta((R_B, E'),(\purify(R_B), E')) \leq \e
\end{align*}

Namely, whenever a party does not reject, its key looks like a fresh random string to Eve.
\end{enumerate}
The quantity $k-m$ is called the \emph{entropy loss} and the quantity $\log (1/\e)$ is called the \emph{security parameter} of the protocol.
\ED

\subsection{Prerequisites from previous work}
One-time message authentication codes (MACs) use a shared random key to authenticate a message in the information-theoretic setting.
\begin{definition} \label{def:mac}
A function family $\{\mac_R : \bits^{d} \to \bits^{v} \}$ is a $\e$-secure one-time MAC for messages of length $d$ with tags of length $v$ if for any $w \in \bits^{d}$ and any function (adversary) $A : \bits^{v} \to \bits^{d} \times \bits^{v}$,

\[\Pr_R[\mac_R(W')=T' \wedge W' \neq w \mid (W', T')=A(\mac_R(w))] \leq \e,\]
where $R$ is the uniform distribution over the key space $\bits^{\ell}$.
\end{definition}

\begin{theorem} [\cite{kr:agree-close}] \label{thm:mac}
For any message length $d$ and tag length $v$,
there exists an efficient family of $(\lceil  \frac{d}{v} \rceil 2^{-v})$-secure
$\mac$s with key length $\ell=2v$. In particular, this $\mac$ is $\eps$-secure when
$v = \log d + \log (1/\e)$.\\
More generally, this $\mac$ also enjoys the following security guarantee, even if Eve has partial information $E$ about its key $R$.
Let $(R, E)$ be any joint distribution.
Then, for all attackers $A_1$ and $A_2$,

\begin{align*}
\Pr_{(R, E)} [&\mac_R(W')=T' \wedge W' \neq W \mid W = A_1(E), \\ &~(W', T') = A_2(\mac_R(W), E)] \leq \left \lceil  \frac{d}{v} \right \rceil 2^{v-\thinf(R|E)}.
\end{align*}

(In the special case when $R\equiv U_{2v}$ and independent of $E$, we get the original bound.)
\end{theorem}

\begin{remark}
Note that the above theorem indicates that the MAC works even if the key $R$ has average conditional min-entropy rate $>1/2$.
\end{remark}

Sometimes it is convenient to talk about average case seeded extractors, where the source $X$ has average conditional min-entropy $\thinf(X|Z) \geq k$ and the output of the extractor should be uniform given $Z$ as well. The following lemma is proved in \cite{dors}.

\BL \cite{dors}
For any $\delta>0$, if $\Ext$ is a $(k, \e)$ extractor then it is also a $(k+\log(1/\delta), \e+\delta)$ average case extractor.
\EL

For a strong seeded extractor with optimal parameters, we use the following extractor constructed in \cite{GuruswamiUV09}.

\BT [\cite{GuruswamiUV09}] \label{thm:optext} 
For every constant $\alpha>0$, and all positive integers $n,k$ and any $\e>0$, there is an explicit construction of a strong $(k,\e)$-extractor $\Ext: \bits^n \times \bits^d \to \bits^m$ with $d=O(\log n +\log (1/\e))$ and $m \geq (1-\alpha) k$. It is also a strong $(k, \e)$ average case extractor with $m \geq (1-\alpha) k-O(\log n+\log (1/\e))$.
\ET

We need the following construction of strong two-source extractors in \cite{Raz05}.
\begin{theorem}[\cite{Raz05}] \label{thm:razweakseed} For any
  $n_1,n_2,k_1,k_2,m$ and any $0 < \delta < 1/2$ with

\begin{itemize}
\item $n_1 \geq 6 \log n_1 + 2 \log n_2$ \item $k_1 \geq (0.5 +
\delta)n_1 + 3 \log n_1 + \log n_2$ \item $k_2 \geq 5 \log(n_1 -
k_1)$ \item $m \leq \delta \min[n_1/8,k_2/40] - 1$
\end{itemize}

There is a polynomial time computable strong 2-source extractor
$\Raz : \bits^{n_1} \times \bits^{n_2} \rightarrow \bits^m$ for
min-entropy $k_1, k_2$ with error $2^{-1.5 m}$.
\end{theorem}

\BT \cite{DLWZ11, CRS11, Li12} \label{thm:enmext}
%For any constant $\delta>0$, there is a constant $c$ such that for any $\epsilon > 0$
%there is an explicit $(k=(1/2+\delta)n,  \eps)$-non-malleable extractor
%$\nm:\zo^n \times \zo^d \\ \to \zo^m$
%with $d=c(m+ \log \eps^{-1} + \log n)$. 
For every constant $\delta>0$, there exists a constant $\beta>0$ such that for every $n, k \in \N$ with $k \geq (1/2+\delta)n$ and $\e>2^{-\beta n}$ there exists an explicit $(k, \e)$ non-malleable extractor with seed length $d=O(\log n+\log \e^{-1})$ and output length $m=\Omega(n)$.
\ET

The following standard lemma about conditional min-entropy is implicit in \cite{NisanZ96} and explicit in \cite{MW97}.

\begin{lemma}[\cite{MW97}] \label{lem:condition} 
Let $X$ and $Y$ be random variables and let ${\cal Y}$ denote the range of $Y$. Then for all $\e>0$, one has
\[\Pr_Y \left [ H_{\infty}(X|Y=y) \geq H_{\infty}(X)-\log|{\cal Y}|-\log \left( \frac{1}{\e} \right )\right ] \geq 1-\e.\]
\end{lemma}

\subsection{The privacy amplification protocol}
We first define the following alternating extraction protocol.

\begin{figure}[htb]
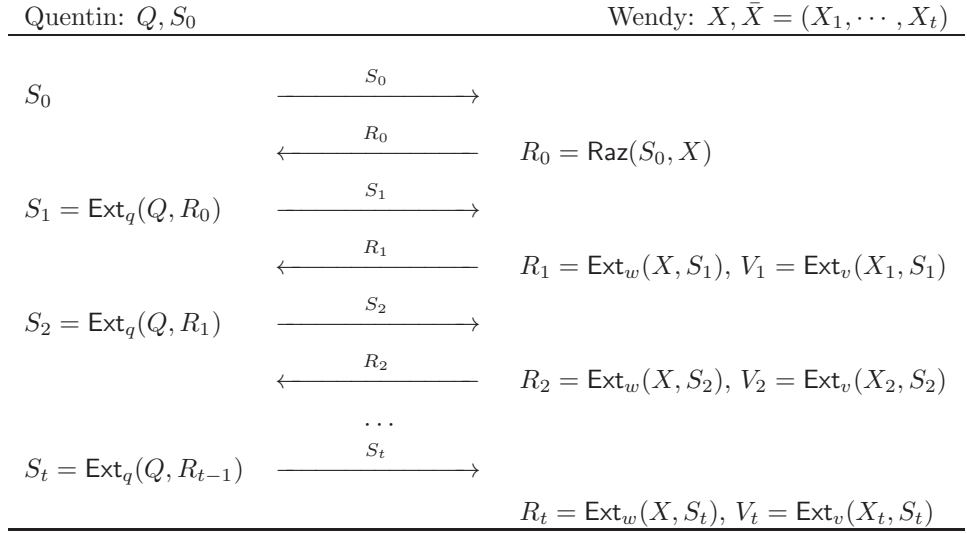

\begin{center}
\begin{small}
\begin{tabular}{l c l}
Quentin:  $Q, S_0$ & &~~~~~~~~~~Wendy: $X, \bar{X}=(X_1, \cdots, X_t)$ \\

%{\tt input:} $W, \mu_A$ && {\tt input:} $W$ \\
\hline\\
$S_0$ & $\llrightarrow[\rule{2.5cm}{0cm}]{S_0}{} $ & \\
& $\llleftarrow[\rule{2.5cm}{0cm}]{R_0}{} $ & $R_0=\Raz(S_0, X)$ \\
$S_1=\Ext_q(Q, R_0)$ & $\llrightarrow[\rule{2.5cm}{0cm}]{S_1}{} $ & \\
& $\llleftarrow[\rule{2.5cm}{0cm}]{R_1}{} $ & $R_1=\Ext_w(X, S_1)$, $V_1=\Ext_v(X_1, S_1)$ \\
$S_2=\Ext_q(Q, R_1)$ & $\llrightarrow[\rule{2.5cm}{0cm}]{S_2}{} $ & \\
& $\llleftarrow[\rule{2.5cm}{0cm}]{R_2}{} $ & $R_2=\Ext_w(X, S_2)$, $V_2=\Ext_v(X_2, S_2)$ \\
& $\cdots$ & \\
$S_t=\Ext_q(Q, R_{t-1})$ & $\llrightarrow[\rule{2.5cm}{0cm}]{S_t}{} $ & \\
& & $R_t=\Ext_w(X, S_t)$, $V_t=\Ext_v(X_t, S_t)$ \\
\hline
\end{tabular}
\end{small}
\caption{\label{fig:altext}
Alternating Extraction.
}
\end{center}
\end{figure}

\textbf{Alternating Extraction.} Assume that we have two parties, Quentin and Wendy. Quentin has a source $Q$,  Wendy has a source $X$ and a source $\bar{X}=(X_1\circ \cdots \circ X_t)$ with $t$ rows. Also assume that Quentin has a weak source $S_0$ with entropy rate $>1/2$ (which may be correlated with $Q$). Suppose that $(Q, S_0)$ is kept secret from Wendy and $(X, \bar{X})$ is kept secret from Quentin.  Let $\Ext_q$, $\Ext_w$, $\Ext_v$ be strong seeded extractors with optimal parameters, such as that in \theoremref{thm:optext}. Let $\Raz$ be the strong two-source extractor in \theoremref{thm:razweakseed}. Let $s, d$ be two integer parameters for the protocol. The \emph{alternating extraction protocol} is an interactive process between Quentin and Wendy that runs in $t+1$ steps. 

In the $0$'th step, Quentin sends $S_0$ to Wendy, Wendy computes $R_0=\Raz(S_0, X)$ and replies $R_0$ to Quentin, Quentin then computes $S_1=\Ext_q(Q, R_0)$. In this step $R_0, S_1$ each outputs $d$ bits. In the first step, Quentin sends $S_1$ to Wendy, Wendy computes $V_1=\Ext_v(X_1, S_1)$ and $R_1=\Ext_w(X, S_1)$. She sends $R_1$ to Quentin and Quentin computes $S_2=\Ext_q(Q, R_1)$. In this step $V_1$ outputs $2^{t-1}s$ bits, and $R_1, S_2$ each outputs $d$ bits. In each subsequent step $i$, Quentin sends $S_i$ to Wendy, Wendy computes $V_i=\Ext_v(X_i, S_i)$ and $R_i=\Ext_w(X, S_i)$. She replies $R_i$ to Quentin and Quentin computes $S_{i+1}=\Ext_q(Q, R_i)$. In step $i$, $V_i$ outputs $2^{t-i}s$ bits, and $R_i, S_{i+1}$ each outputs $d$ bits. Therefore, this process produces the following sequence: 

\begin{align*}
S_0, R_0=\Raz(S_0, X), & S_1=\Ext_q(Q, R_0), V_1=\Ext_v(X_1, S_1), R_1=\Ext_w(X, S_1), \cdots, \\ &S_t=\Ext_q(Q, R_{t-1}), V_t=\Ext_v(X_t, S_t), R_t=\Ext_w(X, S_t).
\end{align*}

\textbf{Look-Ahead Extractor.} Now we can define our look-ahead extractor. Let $Y=(Q, S_0)$ be a seed, the look-ahead extractor is defined as 

\[\laext((X, \bar{X}),  Y)=\laext((X, \bar{X}), (Q, S_0)) \eqdef V_1, \cdots, V_t.\]

Note that the look-ahead extractor can be computed by each party (Alice or Bob) alone in our final protocol. Now we give our protocol for privacy amplification.

\subsubsection{The protocol}

Now we give our privacy amplification protocol for the setting when $\thinf(X|E) = k \ge \delta n$.
We assume that the error $\e$ we seek satisfies $2^{-\Omega(\delta n)}< \e < 1/n$. In the description below, it will be convenient to introduce an ``auxiliary'' security parameter $s$. Eventually, we will set $s=\log(C/\e)+O(1)=\log(1/\e)+O(1)$, so that $O(C)/2^s<\e$, for a sufficiently large $O(C)$ constant related to the number of ``bad'' events we will need to account for. We will need the following building blocks:

%, for now using a relatively general setting our parameters, and instantiating them later with %concrete values.

\begin{itemize}
\item Let $\Cond:\zo^n \rightarrow (\zo^{n'})^C$ be a rate-$(\delta \to 0.9, 2^{-s})$-somewhere-condenser. Specifically, we will use the one from \theoremref{thm:swcondenser}, where $C=\poly(1/\delta)=O(1)$, $n'=\poly(\delta)n=\Omega(n)$ and
$2^{-s} \gg 2^{-\Omega(\delta n)}$.

\item Let $\nmExt:\bits^{n'}\times \bits^{d'}\rightarrow \bits^{m'}$ be a $(0.8n',2^{-s})$-non-malleable extractor. Specifically, we will use the one from \theoremref{thm:enmext} and set the output length $m' = 6\cdot 2^Cs$.
%which is legal since $s\ll n'$).

\item Let $\Ext, \Ext_q, \Ext_w, \Ext_v$ be seeded extractors with error $2^{-s}$, seed length $d=O(\log n+s)$ and optimal entropy loss $O(s)$ as in \theoremref{thm:optext}. $\Ext_q, \Ext_w, \Ext_v$ will be used in $\laext$. 

\item Let $\Raz$ be the strong two-source extractor in \theoremref{thm:razweakseed}. This will be used in $\laext$.
%\item Let $\mac$ be the one-time, $2^{-s}$-secure MAC for $d$-bit messages, whose key length $\ell'=m'$ (the output length of $\nmExt$). Using the construction from \theoremref{thm:mac},
%we set the tag length $v' = s + \log d \le 2s$ (since $d\le n \le 1/\e \le 2^s$), which
%means that the key length $\ell' = m' = 2v' \le 4s$.

\item Let $\lrmac$ be a one-time (``leakage-resilient'') MAC for $d$-bit messages, with key length $2^C(6s)$ and tag length $2^C(3s)$. We will later use the second part of \theoremref{thm:mac} to argue good security of this MAC even when some bits of partial information about its key is leaked to the attacker. 
\end{itemize}
Using the above building blocks, the protocol is given in Figure~\ref{fig:AKA}. To emphasize the presence of Eve, we will use `prime' to denote all the protocol values seen or generated by Bob; e.g., Bob picks $W'$, but Alice sees potentially different $W$, etc. %Also, for any random variable $G$ used in describing our protocol, we use the notation $G=\perp$ to indicate that $G$ was never assigned any value, because the party who was supposed to assign $G$ rejected earlier. The case of final keys $R_A$ and $R_B$ becomes a special case of this convention.

\begin{figure}[htb]
\begin{center}
\begin{small}
\begin{tabular}{l c l}
Alice:  $X$ & Eve: $E$ & ~~~~~~~~~~~~Bob: $X$ \\

%{\tt input:} $W, \mu_A$ && {\tt input:} $W$ \\
\hline\\
 $(X_1,\ldots X_C) = \Cond(X)$. &  &  $(X_1,\ldots X_C) = \Cond(X)$.\\
Sample random $Y=(Y_1, Y_2, Y_3)$ && \\
such that& &\\
$|Y_1|=max\{d, d'\}, |Y_3|=30max\{d, d'\}+3s$, &&\\
$|Y_2|=4Cd+31max\{d, d'\}+4s$ & & \\
 & $(Y_1, Y_2, Y_3) \llrightarrow[\rule{0.1cm}{0cm}]{} (Y'_1, Y'_2, Y'_3)$ & \\
&& Sample random $W'$ with $d$ bits.\\
&&  $Z' = \Ext(X;Y_1')$ with $2^C(6s)$ bits.\\
&&  $\bar{X}'=(\bar{X}'_1, \ldots, \bar{X}'_C)$, \\
&&  where $\bar{X}'_i=\nmExt(X_i, Y'_1)$. \\
&&  $V'=(V'_1, \ldots, V'_C)=\laext((X, \bar{X}'), (Y'_2, Y'_3))$ \\
&&  with parameters $(2s, d)$. \\
&&  $T' = \lrmac_{Z'}(W')$.\\
&&  Set final $R_B = \Ext(X;W')$.\\
 & $(W,T, \bar{V}) \llleftarrow[\rule{0.1cm}{0cm}]{} (W',T', V')$ & \\
 $Z = \Ext(X;Y_1)$ with $2^C(6s)$ bits. &&\\
 $\bar{X}=(\bar{X}_1, \ldots, \bar{X}_C)$, && \\
 where $\bar{X}_i=\nmExt(X_i, Y_1)$. &&\\
 $V=(V_1, \ldots, V_C)=\laext((X, \bar{X}), (Y_2, Y_3))$ &&\\
 with parameters $(2s, d)$. && \\
{\bf If} $T \neq \lrmac_{Z}(W)$ or &&\\
$ V \neq  \bar{V}$ {\em reject}.&&\\
Set final $R_A = \Ext(X;W)$.&&\\
\hline
\end{tabular}
\end{small}
\caption{\label{fig:AKA}
$2$-round Privacy Amplification Protocol for $\thinf(X|E) \geq \delta n$.
}
\end{center}
\end{figure}

\BT
For any constant $\delta>0$, the above protocol is a privacy amplification protocol with security parameter $\log(1/\e)$, entropy loss $2^{\poly(1/\delta)}\log(1/\e)$, randomness complexity $\poly(1/\delta)\log(1/\e)$ and communication complexity $2^{\poly(1/\delta)}\log(1/\e)$.
\ET

\begin{thmproof}
The proof can be divided into two cases: whether the adversary changes $Y_1$ or not. Note that $Y_1, Y_2, Y_3$ and $W$ all have size $O(s)$.

\textbf{Case 1:} The adversary does not change $Y_1$. In this case, note that $Z=Z'$ and is $2^{-s}$-close to uniform in Eve's view (even conditioned on $Y_1, Y_2, Y_3$). Note that the size of $(V'_1, \ldots, V'_C)$ is at most $\sum_i 2^{C-i}(2s) < 2^C (2s)$, and the size of $Z$ is $2^C (6s)$. Therefore, by \lemmaref{lem:amentropy} even if conditioned on $(V'_1, \ldots, V'_C)$, the average conditional min-entropy of  $Z$ is at least $2^C (6s)-2^C (2s)=2^C (4s)$. Therefore by theorem~\ref{thm:mac} the probability that Eve can change $W'$ to a different $W$ without causing Alice to reject is at most 

\[\left \lceil  \frac{O(s)}{2^C (3s)} \right \rceil 2^{2^C (3s)-2^C (4s)}+2^{-s} \leq O(2^{-s}).\]

When $W=W'$, by theorem~\ref{thm:optext} $R_A=R_B$ and is $2^{-s}$-close to uniform in Eve's view. 

\textbf{Case 2:} The adversary does change $Y_1$. In this case, first note that by \theoremref{thm:swcondenser}, $\Cond(X)=(X_1,\ldots X_C)$ is $2^{-s}$-close to a somewhere rate-$0.9$-source with $C$ rows, and each row has length $\Omega(n)$. In the following we will simply treat it as a somewhere rate-$0.9$-source, since this only adds $2^{-s}$ to the error. We assume that $X_g, 1 \leq g \leq C$ is a rate $0.9$-source \footnote{In general a somewhere rate-$0.9$-source is a convex combination of elementary somewhere rate-$0.9$-sources, but without loss of generality we can assume it is an elementary somewhere rate-$0.9$-source.}. 

Now since the adversary changes $Y_1$ to $Y'_1 \neq Y_1$, by \theoremref{thm:enmext} we have that 

\[(\bar{X}_g, \bar{X}'_g, Y_1) \approx_{2^{-s}} (U_{m'},  \bar{X}'_g, Y_1).\]

As the first step for the following analysis, we now fix $Y_1$, $Y'_1$ and $Z', \bar{X}'_g$. Note that $Y'_1$ is a deterministic function of $(Y_1, Y_2, Y_3)$, and after fixing $Y'_1$, $(Z', \bar{X}'_g)$ is a deterministic function of $X$. Thus by \lemmaref{lem:amentropy} we have the following claim.

\BCM \label{clm:condition0}
After the fixings of $(Y_1, Y'_1, Z', \bar{X}'_g)$, $\bar{X}_g$ is a deterministic function of $X$ and is $2^{-s}$ close to a source with average conditional min-entropy $m'-2^C(6s)$. 
\ECM

Note that by \lemmaref{lem:amentropy}, after this fixing, the average conditional min-entropy of $X$ is at least $k-m'-2^C(6s)$. Now we analyze the sequences $(V_1, \ldots, V_C)$ produced by $\laext$, or equivalently, the alternating extraction process. Note here $(Q, S_0)=(Y_2, Y_3)$ and $(Q', S'_0)=(Y'_2, Y'_3)$. First we have the following claim.

\BL \label{lem:condition1}
In step $0$, we have 
\[(R_0, S_0, S'_0) \approx_{2^{-s}} (U_d, S_0, S'_0)\]
and
\[(S_1, R_0, S_0, R'_0, S'_0) \approx_{5 \cdot 2^{-s}} (U_d, R_0, S_0, R'_0, S'_0).\]
Moreover, conditioned on $(S_0, S'_0)$, $(R_0, R'_0)$ are both deterministic functions of $X$; conditioned on $(R_0, S_0, R'_0, S'_0)$, $(S_1, S'_1)$ are both deterministic functions of $Q$.
\EL

\begin{proof}[Proof of the claim.] 
Note that previously we have fixed $Y_1, Y'_1$. Since $Y_1$ is independent of $Y_3$ and $Y'_1$ is a deterministic function of $Y$, by \lemmaref{lem:condition} we have that $S_0=Y_3$ is $2^{-s}$-close to a source with min-entropy $29max\{d, d'\}+2s$. Note that $Y_3$ and $X$ are still independent. Thus by \theoremref{thm:razweakseed} we have that 

\[(R_0, S_0) \approx_{2^{-s}} (U_d, S_0).\]

Since conditioned on $S_0$, $R_0$ is a deterministic function of $X$, which is independent of $Y$, we also have that 

\[(R_0, S_0, S'_0) \approx_{2^{-s}} (U_d, S_0, S'_0).\]

Now we fix $(S_0, S'_0)$ and $(R_0, R'_0)$ are both deterministic functions of $X$. Note that $S_0=Y_3$ is independent of $Y_2$ and $S'_0=Y'_3$ is a deterministic function of $Y$. Thus by \lemmaref{lem:condition} we have that conditioned on these fixings $Q=Y_2$ is $2^{-s}$-close to a source with entropy $4Cd$. Since $R_0, R'_0$ are both deterministic functions of $X$, they are independent of $Q$. Therefore by \theoremref{thm:optext} we have 

\[(S_1, R_0, R'_0) \approx_{2^{-s}} (U_d, R_0, R'_0).\]

Thus altogether we have that 

\[(S_1, R_0, S_0, R'_0, S'_0) \approx_{5 \cdot 2^{-s}} (U_d, R_0, S_0, R'_0, S'_0)\]
Moreover, conditioned on $(R_0, S_0, R'_0, S'_0)$, $(S_1, S'_1)$ are both deterministic functions of $Q$.
\end{proof}

Now we fix $(R_0, S_0, R'_0, S'_0)$. Note that after this fixing, $S_1, S'_1$ are both functions of $Q=Y_2$. Note that $Q$ has min-entropy at least $4Cd$.

For $i=0, \cdots, C$, let $View_i=(S_0, \cdots, S_{i}, R_0, \cdots, R_{i}, V_1, \cdots, V_i)$. Similarly define $View'_i$ to be the corresponding variables produced by Bob. We now have the following lemma.

\BL \label{lem:condition2}
For any $i$, we have that 
\[(R_i, View_{i-1}, View'_{i-1}, S_i, S'_i) \approx_{(2i+4)2^{-s}} (U_d, View_{i-1}, View'_{i-1}, S_i, S'_i)\]

and

\[(S_{i+1}, View_i, View'_i) \approx_{(2i+5)2^{-s}} (U_d, View_i, View'_i).\]
Moreover, conditioned on $(View_{i-1}, View'_{i-1}, S_i, S'_i)$, $(R_i, R'_i, V_i, V'_i)$ are all deterministic functions of $X$; conditioned on $(View_i, View'_i)$, $(S_{i+1}, S'_{i+1})$ are both deterministic functions of $Q$.
\EL

\begin{proof}
We prove the lemma by induction on $i$. When $i=0$, the statements are already proved in \lemmaref{lem:condition1}. Now we assume that the statements hold for $i=j$ and we prove them for $i=j+1$.

We first fix $(View_j, View'_j)$. Since now $(S_{j+1}, S'_{j+1})$ are both deterministic functions of $Q$, they are independent of $X$. Moreover $S_j$ is $(2j+5)2^{-s}$-close to uniform. Note that the average conditional min-entropy of $X$ is at least $k-m'-2^C(6s)-2^C(4s)-2Cd=k-m'-2^C(10s)-2Cd$. Therefore by \theoremref{thm:optext} we have that

\[(R_{j+1}, View_j, View'_j, S_{j+1}, S'_{j+1}) \approx_{(2j+6)2^{-s}} (U_d, View_j, View'_j, S_{j+1}, S'_{j+1}).\] 

Moreover, conditioned on $(View_j, View'_j, S_{j+1}, S'_{j+1})$,  $(R_{j+1}, R'_{j+1}, V_{j+1}, V'_{j+1})$ are all deterministic functions of $X$.

Next, since conditioned on $(View_j, View'_j, S_{j+1}, S'_{j+1})$,  $(R_{j+1}, R'_{j+1})$ are both deterministic functions of $X$, they are independent of $Q$. Moreover $R_{j+1}$ is $(2j+6)2^{-s}$-close to uniform. Note that the average conditional min-entropy of $Q$ is at least $4Cd-2Cd =2Cd$. Therefore by \theoremref{thm:optext} we have that

\begin{align*}
& (S_{j+2}, View_j, View'_j, S_{j+1}, S'_{j+1}, R_{j+1}, R'_{j+1}, V_{j+1}, V'_{j+1}) \\ \approx_{(2j+7)2^{-s}} & (U_d, View_j, View'_j, S_{j+1}, S'_{j+1}, R_{j+1}, R'_{j+1}, V_{j+1}, V'_{j+1}).
\end{align*}

Namely, 

\[(S_{j+2}, View_{j+1}, View'_{j+1}) \approx_{(2(j+1)+5)2^{-s}} (U_d, View_{j+1}, View'_{j+1}).\] 
Moreover, conditioned on $(View_{j+1}, View'_{j+1})$,  $(S_{j+2}, S'_{j+2})$ are deterministic functions of $Q$.
\end{proof}

Now we consider step $g$ in the alternating extraction protocol. By \claimref{clm:condition0}, conditioned on $(Y_1, Y'_1, Z', \bar{X}'_g)$, $\bar{X}_g$ is a deterministic function of $X$ and is $2^{-s}$ close to a source with average conditional min-entropy $m'-2^C(6s)$. Since $m'=\Omega(\delta n)$, when $s$ is significantly smaller than $\delta n$ (but we can still achieve up to $s=\Omega(\delta n)$), we have that $m' \geq 2^C(12s)+2Cd$ and thus $m'-2^C(6s) \geq 2^C(6s)+2Cd$. 

We now have the following lemma.

\BL
Conditioned on the fixing of $(View_{g-1}, View'_{g-1})$, $\bar{X}_g$ is a deterministic function of $X$, and the average conditional min-entropy of $\bar{X}_g$ is at least $2^C(2s)$.
\EL

\begin{proof}
We first use induction to prove that for any $i$, conditioned on the fixing of $(View_i, View'_i)$, $\bar{X}_g$ is a deterministic function of $X$. When $i=0$, we first fix $(S_0, S'_0)$, which is independent of $\bar{X}_g$. After this fixing $(R_0, R'_0)$ is a deterministic function of $X$. Thus we can now fix $(R_0, R'_0)$ and conditioned on this fixing, $\bar{X}_g$ is still a deterministic function of $X$. Thus the statement holds for $i=0$.

Now assume that the statement holds for $i=j$, we show that it also holds for $i=j+1$. Specifically, when $(View_j, View'_j)$ is fixed, $(S_{j+1}, S'_{j+1})$ is a deterministic function of $Q$, which is independent of $X$. Thus we can fix $(S_{j+1}, S'_{j+1})$. Now after this fixing, $(R_{j+1}, R'_{j+1}, V_{j+1}, V'_{j+1})$ is a deterministic function of $X$. Thus we can now fix $(R_{j+1}, R'_{j+1}, V_{j+1}, V'_{j+1})$ and conditioned on this fixing, $\bar{X}_g$ is still a deterministic function of $X$. Thus the statement holds for $i=j+1$. Therefore, 
for any $i$, conditioned on the fixing of $(View_i, View'_i)$, $\bar{X}_g$ is a deterministic function of $X$.

Finally, note that only the fixings of $(R_j, R'_j, V_j, V'_j)$ can cause $\bar{X}_g$ to lose entropy. Since the total size of $\{(R_j, R'_j, V_j, V'_j)\}$ is at most $\sum_{i=1}^{C} 2^{C-i}(4s)+2Cd < 2^C(4s)+2Cd$, by \lemmaref{lem:amentropy} the average conditional min-entropy of $\bar{X}_g$ is at least $2^C(2s)$.
\end{proof}

Now by \lemmaref{lem:condition2}, conditioned on the fixing of $(View_{g-1}, View'_{g-1})$, $S_g \approx_{(2g+3)2^{-s}} U_d$ and $S_g, S'_g$ are both deterministic functions of $Q$, which is independent of $X$. Thus $S_g$ and $\bar{X}_g$ are independent. Therefore by \theoremref{thm:optext} we have 

\[(V_g, S_g, S'_g) \approx_{2^{-s}} (U_{2^{C-g}(2s)}, S_g, S'_g).\]

Adding back all the errors, and note that we have fixed $(Y_1, Y'_1, Z', \bar{X}'_g)$ and $(View_{g-1}, View'_{g-1})$, we have that

\[(V_g, S_g, S'_g, View_{g-1}, View'_{g-1}, Z', \bar{X}'_g) \approx_{O(C2^{-s})} (U_{2^{C-g}(2s)}, S_g, S'_g, View_{g-1}, View'_{g-1}, Z', \bar{X}'_g).\]

In particular, note that $V'_g=\Ext_v(\bar{X}'_g, S'_g)$ and for a fixed message $w'$, $T'=\lrmac_{Z'}(w')$ is a function of $Z'$. Thus we have that 

\[(V_g, View'_{g-1}, T', V'_g) \approx_{O(C2^{-s})} (U_{2^{C-g}(2s)}, View'_{g-1}, T', V'_g).\]

This implies that 

\[(V_g, T', V'_1, \cdots, V'_g) \approx_{O(C2^{-s})} (U_{2^{C-g}(2s)}, T', V'_1, \cdots, V'_g).\]

Now note the size of $(V'_{g+1}, \cdots, V'_C)$ is at most $\sum_{i=g+1}^{C}2^{C-i}(2s) = 2^{C-g}(2s)-2s$, and that $V_g$ has size $2^{C-g}(2s)$. Therefore, if $V_g$ is uniform conditioned on $(T', V'_1, \cdots, V'_g)$, then by \lemmaref{lem:condition} we have that with probability $1-2^{-s}$ over the fixings of $(T', V'_1, \cdots, V'_C)$, $V_g$ is a source with min-entropy $s$. Thus the probability that Eve can come up with the correct $V_g$ is at most $2 \cdot 2^{-s}$. Adding back the error, we have that in the case that Eve changes $Y_1$, the probability that Alice does not reject is at most $O(C2^{-s})$. For an appropriately chosen $s=\log (1/\e)+O(1)$ this is at most $\e$.

Finally, note that the entropy loss of the protocol is $O(2^Cs)=2^{\poly(1/\delta)}\log (1/\e)=O(\log (1/\e))$, the randomness complexity is $O(Cd+s)=O(Cs)=\poly(1/\delta)\log (1/\e)=O(\log (1/\e))$ and the communication complexity is $O(2^Cs)=2^{\poly(1/\delta)}\log (1/\e)=O(\log (1/\e))$.
\end{thmproof}

\section{Conclusions and Open Problems}\label{sec:con}
In this paper we present a strong connection between non-malleable extractors and two-source extractors. First, we show that non-malleable extractors can be used to construct two-source extractors. If the non-malleable extractor works for small min-entropy and has a short seed length with respect to the error, then the resulted two-source extractor beats the best known construction of two-source extractors. Second, we show that two-source extractors of the form $\IP(f(X), Y)$ can be used to construct non-malleable extractors. Using this connection, we give the first explicit constructions of non-malleable extractors for min-entropy $k<n/2$.

The most important message from this part is perhaps that, non-malleable extractors and two-source extractors, although seemingly different, are closely related. Thus, future research should probably consider these two kinds of extractors together, as improvements in one kind may lead to improvements in the other. Our connection also suggests that it may be hard to construct non-malleable extractors for small entropy. However, strictly speaking, our result only shows that it may be hard to construct non-malleable extractors for small entropy with short seed length with respect to the error. It is totally possible that we can get explicit non-malleable extractors for small entropy with large seed length. Moreover, the weaker notion of non-malleable condensers introduced in \cite{Li12} is a hopeful alternative.

We also give the first privacy amplification protocol for $k=\delta n$ that simultaneously achieves optimal round complexity (2 rounds), asymptotically optimal entropy loss and communication complexity. However, our entropy loss is $2^{\poly(1/\delta)}s$, which has a large hidden constant for small $\delta$. As a comparison, the protocol in \cite{DLWZ11} runs in $\poly(1/\delta)$ rounds but only has entropy loss $\poly(1/\delta)s$. Thus for practical purposes it is interesting to see if we can reduce the hidden constant. In particular, it remains an interesting open problem to construct non-malleable extractors or non-malleable condensers for arbitrarily linear min-entropy. 

\section*{Acknowledgments}

We are grateful to David Zuckerman for many valuable discussions, especially for his suggestion to use the BCH code.

\bibliographystyle{alpha}

\bibliography{refs}

\appendix

\section{The Existence of Generalized Non-Malleable Extractors}\label{sec:nmexist}
In this section we prove the existence of non-malleable extractors with more than one adversarial seeds. First, we have the following definition.

\BD
A function $\nm:\bits^n \times \bits^d \to \bits^m$ is a $(r, k,\eps)$-non-malleable extractor if,
for any source $X$ with $\hinf(X) \geq k$ and any $r$ function $\adv_i:\bits^d \to \bits^d, i=1, \cdots, r$ such that $\adv_i(y) \neq y$ for all~$i$ and $y$,
the following holds.
When $Y$ is chosen uniformly from $\bits^d$ and independent of $X$,
\[
(\nm(X,Y),\{\nm(X,\adv_i(Y))\},Y) \approx_\eps (U_m,\{\nm(X,\adv_i(Y))\},Y).
\]
\ED

We will prove the following theorem.

\BT \label{thm:nmexist}
For any constant $r \geq 1$, there exists a $(r, k,\eps)$-non-malleable extractor as long as 

\[d > \frac{3}{2}\log(n-k)+3\log(1/\e)+O(1)\]

\[k > (r+1)m+\frac{d}{3}+2 \log (1/\e)+\log(d)+O(1) \]
\ET

We prove the theorem by using the probabilistic method, similar to the existence proof in \cite{DW09}. A function $f: \bits^n \times \bits^d \to \bits^m$ is a $(r, k,\eps)$-non-malleable extractor if for all $(n,k)$ sources $X$, all adversarial functions $\{\adv_i\}$ and all distinguishers $\calD$, we have

\[|\Pr[\cd(f(X,Y),\{f(X,\adv_i(Y))\},Y)=1]-\Pr[\cd(U_m,\{f(X,\adv_i(Y))\},Y)=1]| \leq \e.\]

As usual, it suffices to consider flat sources $X$. For the purpose of a union bound, we fix some $\cd, \{\adv_i\}$ and a source $X$ which is uniformly distributed on some subset $\Supp(X) \subseteq \bits^n$ with $|\Supp(X)|=2^k$. We use $\df$ to denote the uniform distribution over the space of all functions $f: \bits^n \times \bits^d \to \bits^m$.

For each $u \in \bits^{rm}$ and $y \in \bits^d$, define

\[\mc(u, y)=|\{u_2 \in \bits^m: \cd(u_2, u, y)=1\}  |.\]

For each $x \in \Supp(X)$, $y \in \bits^d$, define the following random variables (where the randomness comes from $\df$).

\[\dl(x, y)=\cd(\df(x, y), \{\df(x, \adv_i(y))\}, y)\] 

\[\dr(x,y)=\frac{\mc(\{\df(x, \adv_i(y))\}, y)}{2^m}\]

\[\dq(x, y)=\dl(x, y)-\dr(x, y)\]

and let 

\[\dbq=\frac{\sum_{x, y} \dq(x, y)}{2^{k+d}}.\]

Thus, $\dbq$ is essentially the quantity 

\[p(f)=\Pr_{f \leftarrow \df}[\cd(f(X,Y),\{f(X,\adv_i(Y))\},Y)=1]-\Pr_{f \leftarrow \df}[\cd(U_m,\{f(X,\adv_i(Y))\},Y)=1].\]

Therefore, we want to upper bound 

\[\Pr[|\dbq|>\e] = \Pr_{f \leftarrow \df}[|p(f)| > \e].\]

Again, it is easy to notice that for any $(x,y)$, $\dbE(\dl(x, y))=\dbE(\dr(x, y))$ and thus $\dbE(\dq(x, y))=0$ and $\dbE(\dbq)=0$. However, the variables $\dq(x, y)$ are not necessarily independent of each other (in particular, the adversarial seeds can form cycles), so we cannot use a simple Chernoff bound here. Now we represent the functions $\{\adv_i\}$ as a directed graph $G=(V, E)$ where the vertex set is $V=\bits^d$ and there is an edge from $y$ to $y'$ iff $\exists i, \adv_i(y)=y'$. Note that each $\adv_i$ is a function, thus the out-degree of every vertex is exactly $r$. The following lemma is proved in \cite{CRS11}.

\BL \label{lem:graph}
Let $G=(V, E)$ be a directed graph without self-loops. Assume that the out-degree of each vertex is at most $r$, where parallel edges are allowed. Then, there exists a subset of the vertices $V' \subseteq V$, such that the induced graph $H=(V', E')$ of $G$ is acyclic, and $|V'| \geq |V|/(r+1)$.
\EL

Let $s=1/(r+1)$. We now use \lemmaref{lem:graph} to decompose $G$ into $t+1$ subgraphs $H_j$ as follows. In each step $j, 1 \leq j \leq t$, we use the lemma to pick a $s$ fraction of vertices from the remaining vertices to form a subset $V_j$ and an induced graph $H_j$, and delete these vertices from $G$. After $t$ steps the remaining graph is $H_{t+1}$. Thus we have that $|V_j|=s(1-s)^{j-1}|V|$ for $j \leq t$ and $|V_{t+1}|=(1-s)^t |V|$.
 
The following lemma is proved in \cite{DW09}.

\BL \label{lem:martg}
For $V' \subseteq V$, let $H$ be the restriction of $G$ to the vertices $V'$ and assume that the graph $H$ is acyclic. Then the set $\{\dq(x, y)_{x \in \Supp(X), y \in V'}\}$ of random variables can be enumerated by $\dq_1, \cdots, \dq_{\ell}$ for $l=|V'|2^k$ such that $\dbE[\dq_i|\dq_1, \cdots, \dq_{i-1}]=0$ for all $1 \leq i \leq l$.
\EL

We can now prove the theorem.

\begin{thmproof}[Proof of \theoremref{thm:nmexist}.] By \lemmaref{lem:graph} and \lemmaref{lem:martg}, we can partition $\{\dq(x, y)\}$ into $t$ (enumerated) sets $\{\dq^j_1, \cdots, \dq^j_{l_j}\}$ where $l_j=|V_j|2^k$ for $j=1, \cdots, t$ and a remaining set $\{\dq(x, y)_{x \in \Supp(X), y \in V_{t+1}}\}$. For each $j \leq t$, $1 \leq i \leq l_j$, we have that $\dbE[\dq^j_i|\dq^j_1, \cdots, \dq^j_{i-1}]=0$. Now for each $j \leq t$, $1 \leq i \leq l_j$, define $S^j_i=\sum_{\ell=1}^{i} \dq^j_i$. Then for any $j \leq t$, $S^j_1, \cdots, S^j_{l_j}$ is a martingale. 

We first show that if for any $j \leq t$, $|S^j_{l_j}| \leq \frac{\e}{2} l_j$ and $|V_{t+1}| \leq \frac{\e}{4}2^d$, then $|\dbq| \leq \e$. Indeed, note that for any $(x, y)$, $\dq(x, y) \leq 2$. Thus in this case we have that 

\begin{align*}
|\dbq| &=\left |\frac{\sum_{j=1}^{t} S^j_{l_j}+\sum_{x \in \Supp(X), y \in V_{t+1}}\dq_{x, y}}{2^{k+d}} \right | \leq \frac{1}{2^{k+d}} \left (\sum_{j=1}^{t} \left |S^j_{l_j} \right |+\left |\sum_{x \in \Supp(X), y \in V_{t+1}}\dq_{x, y} \right | \right) \\ & \leq \frac{1}{2^{k+d}} \left (\sum_{j=1}^{t} \frac{\e}{2} l_j+2 \cdot \frac{\e}{4}2^{d+k} \right ) =  \frac{1}{2^{k+d}} \left (\frac{\e}{2}\left (\sum_{j=1}^{t} l_j \right )+\frac{\e}{2}2^{d+k}\right ) \\ & \leq \frac{1}{2^{k+d}} \left (\frac{\e}{2} 2^{d+k}+\frac{\e}{2}2^{d+k} \right ) = \e.
\end{align*}

Next, by Azuma's inequality we have that for any $j \leq t$, 

\[\Pr[|S^j_{l_j}| > \frac{\e}{2} l_j] \leq e^{-\frac{1}{32}l_j \e^2} \leq e^{-\frac{1}{32}s(1-s)^{t-1}2^{d+k} \e^2}\]
since $l_j =|V_j|2^k \geq s(1-s)^{t-1}2^{d+k}$.

Therefore, by the union bound, 

\[\Pr[\exists j \leq t, |S^j_{l_j}| > \frac{\e}{2} l_j] \leq te^{-\frac{1}{32}s(1-s)^{t-1}2^{d+k} \e^2}.\]

Now we will apply the union bound to all possible $X, \{\adv_i\}, \calD$. Let $N=2^n, K=2^k, D=2^d, M=2^m$. Then there are $\binom{N}{K}$ possible sources $X$, there are $D^{rD}$ possible adversaries $\{\adv_i\}$ and there are $2^{DM^{r+1}}$ possible distinguishers $\calD: \bits^{m} \times (\bits^{m})^r \times \bits^d \to \bits$. Thus to ensure that there exists a function that is a $(r, k, \e)$-non-malleable extractor, all we need is to satisfy the following inequalities:

\[\binom{N}{K}D^{rD}2^{DM^{r+1}} te^{-\frac{1}{32}s(1-s)^{t-1}2^{d+k} \e^2} <1\]

and 

\[|V_{t+1}|=(1-s)^t 2^d \leq \frac{\e}{4}2^d.\]

Choose $t$ such that $(1-s)^t=2^{-\frac{d}{3}}$. Now it is easy to check that both of these conditions are satisfied when the statements in the theorem hold.
\end{thmproof}

\section{Another Construction of Non-Malleable Extractors for Entropy Rate $1/2-\delta$}\label{sec:alter}
Here we give another construction of non-malleable extractors for $(n, k)$ sources with $k=(1/2-\delta)n$ for some constant $\delta>0$. 

Given an $(n,k)$-source $X$ with $k = (1/2-\delta)n$, we first pick a prime $p$ such that $2^n < p < 2^{n+1}$. By Bertrand's postulate, there is always such a prime. Now treat $X$ as an element in the field $\F_p$. Next we take an independent and uniform seed $Y \in \bits^n$ and again treat $Y$ as an element in $\F_p$. Encode $X, Y$ such that $\Enc(X)=(X, X^2)$ and $\Enc(Y)=(Y, Y^2)$. The operations are in $\F_p$. Our non-malleable extractor is defined as 

\[\nmExt(X, Y) = \IP(\Enc(X), \Enc(Y)) \text{ mod } M\]
for some integer $M=2^m$ that we will choose later. Note that $\Enc(X)$ and $\Enc(Y)$ are vectors in $\F^2_p$ and $\IP$ is the inner product function taken over $\F_p$.

%To show that our construction is a non-malleable extractor, we are going to use the non-malleable XOR lemma. 
Again, we show that for any weak source $X$ with min-entropy $(1/2-\delta) n$, $3\Enc(X)$ is close to a weak source that has min-entropy $(1/2+\delta) \log (p^2)$. 

\BL \label{lem:grow} 
Let $\F=\F_p$ for $p$ prime and $X$ be a random variable over $\F$. There is a universal constant $\delta>0$ such that if $X$ is any weak source with min-entropy $(1/2-\delta) n$, $3 \Enc(X)$ is $p^{-\Omega(1)}$-close to a source with min-entropy $(1/2+\delta) \log (p^2)$.
\EL

\begin{proof}
Note that $X$ has min-entropy $(1/2-\delta) n > (1/2-\delta) \log p-1$. First consider the distribution $2\Enc(X)$. Note that the distribution is of the form $(X+X, X^2+X^2)$. For any $(a, b)$ in the support of $2\Enc(X)$, we have that $a=x_1+x_2$ and $b=x^2_1+x^2_2$. Thus there are at most 2 different pairs of $(x_1, x_2)$ that satisfy both equations. Therefore the min-entropy of $2 \Enc(X)$ is at least $2 H_\infty(X)-1$. Now let $k=H_\infty(X)-1$, we have that $\Enc(X)$ has min-entropy at least $k$ and $2\Enc(X)$ has min-entropy at least $2k$. We now have the following claim.

\BCM Let $\alpha, \beta$ be the two constants in \theoremref{thm:incidence}. Then $3 \Enc(X)$ is $2^{-\Omega(k)}$-close to a source with min-entropy $(1+\alpha/2)2k$.
\ECM

\begin{proof}[Proof of the claim.]
Note that an element in the support of $3 \Enc(X)$ has the form $(x_1+x_2+x_3, x^2_1+x^2_2+x^2_3)$. This determines the point

\begin{align*}
&(x_1+x_2+x_3, (x_1+x_2+x_3)^2-(x^2_1+x^2_2+x^2_3)) \\
= & ((x_1+x_2)+x_3, 2(x_1+x_2)x_3+(x_1+x_2)^2-(x^2_1+x^2_2))
\end{align*}

Let $a=x_1+x_2$ and $b=x^2_1+x^2_2$, this point is 

\[(a+x_3, 2ax_3 +a^2-b).\]

Let $\bar{x}_3=a+x_3$, then 

\[(a+x_3, 2ax_3 +a^2-b)=(\bar{x}_3, 2a\bar{x}_3-a^2-b).\]

For a fixed $(a=x_1+x_2, b=x^2_1+x^2_2)$ define the line 

\[\ell_{a,b}=\{(x, 2ax-a^2-b)|x \in \F\}.\]

Note that for different $(a, b)$, the line $\ell_{a, b}$ is also different. Thus we have a set of lines $L=\{\ell_{a, b}\}$. Note that $x_3$ is sampled from $X_3$, which has min-entropy $k$ and $(a, b)$ is sampled from $\Enc(X_1)+\Enc(X_2)$, which has min-entropy $2k$. Further note that these two distributions are independent. Since every weak source with min-entropy $k$ is a convex combination of flat $k$ sources, without loss of generality we can assume that $X_3$ and $\Enc(X_1)+\Enc(X_2)$ are both flat sources. Thus $L$ has size $2^{2k}$. 

Now assume that $3 \Enc(X)$ is $\e$-far from any source with min-entropy $(1+\alpha/2)2k$. Since $3 \Enc(X)$ determines the distribution $(A+X_3, 2AX_3 +A^2-B)$, this distribution is also $\e$-far from any source with min-entropy $(1+\alpha/2)2k$. Thus there must exist some set $M$ of size at most $2^{(1+\alpha/2)2k}$ such that 

\[\Pr_{(a, b) \leftarrow 2\Enc(X), X_3 \leftarrow X}[(a+x_3, 2ax_3 +a^2-b) \in M] \geq \e.\]  

Note that whenever $(a+x_3, 2ax_3 +a^2-b) \in M$, this point has an incidence with the line $\ell_{a,b}$. Further note that whenever $(a, b)$ is different or $x_3$ is different, the incidence is also different. Thus by the above inequality the number of incidences between the set of points $M$ and the set of lines $L$ is  at least

\[\Pr_{(a, b) \leftarrow 2\Enc(X), X_3 \leftarrow X}[(a+x_3, 2ax_3 +a^2-b) \in M] 2^k 2^{2k} \geq \e2^{3k}.\]

On the other hand, since $L$ has size $2^{2k}$ and $M$ has size $2^{(1+\alpha/2)2k} \leq 2^{(1+\alpha/2)2(1/2-\delta)\log p} < 2^{(1+\alpha/2)\log p} \leq p^{2-\beta}$, by \theoremref{thm:incidence}, the number of incidences between $M$ and $L$ is at most $O(2^{(3/2-\alpha)(2+\alpha)k}) < 2^{3k(1-\alpha/6)}=2^{-\alpha k/2}2^{3k}$.

Thus we must have $\e < 2^{-\alpha k/2}$.
\end{proof}

By choosing $\delta$ appropriately and noting that $k \geq (1/2-\delta)\log p-2$, the lemma is proved.
\end{proof}

Now we can use the non-uniform XOR lemma to argue that our extractor is non-malleable. Specifically, we have the following lemma.

\BL \label{lem:main}
Let $\delta$ be the constant in \lemmaref{lem:grow}. Given any $(n,k)$-source $X$ with $k = (1/2-\delta)n$, and $Y$ an independent source over $\bits^n$ with min-entropy $(1-\delta)n$, let $W=\IP(\Enc(X), \Enc(Y))$ and $W'=\IP(\Enc(X), \Enc(Y'))$ where $Y'=\adv(Y)$ and $\forall y \in \bits^n, \adv(y) \neq y$. For any two characters $\psi(s)=e^{2 \pi i ts/p}$ and $\psi'(s)=e^{2 \pi i t's/p}$ where $t, t' \in \F_p$ and $t \neq 0$, 

\[|E_{W, W'}[ \psi(W) \psi'(W')] | \leq 2^{-\Omega(n)}.\]
\EL

\begin{proof}
Note that $W, W'$ are deterministic functions of $X, Y$. Thus 

\[E_{W, W'}[ \psi(W) \psi'(W')] = E_{X, Y}[ \psi(W) \psi'(W')].\] 
Depending on whether $\psi'$ is trivial, we have two cases.

\textbf{Case 1: $t'=0$}. This corresponds to the case where $\psi'$ is the trivial character. In this case $\psi'(W')$ is always 1. Thus 

\[E_{W, W'}[ \psi(W) \psi'(W')] = E_{X, Y}[ \psi(W)]=E_{X, Y}[ \psi(\Enc(X) \cdot \Enc(Y))].\]

Note that $\Enc(Y)$ has the same min-entropy as $Y$, which is $(1-\delta)n$. Now consider $\Enc(X)$. Since $X$ has min-entropy $(1/2-\delta)n$, by \lemmaref{lem:grow} $3\Enc(X)$ is $p^{-\Omega(1)}$-close to having min-entropy $(1/2+\delta) \log (p^2)$. Now note that the min-entropy of $4\Enc(X)-4\Enc(X)$ is at least the min-entropy of $4\Enc(X)$, and which in turn is at least the min-entropy of $3\Enc(X)$. Thus $4\Enc(X)-4\Enc(X)$ is  $p^{-\Omega(1)}$-close to having min-entropy $(1/2+\delta) \log (p^2)$. Since $(1/2+\delta) \log (p^2)+(1-\delta)n> (1+2\delta) \log p+(1-\delta)(\log p-1) > (2+\delta)\log p-1$, by \lemmaref{lem:char2} we have

\[|E_{W, W'}[ \psi(W) \psi'(W')]| = |E_{X, Y}[ \psi(\Enc(X) \cdot \Enc(Y))] | \leq (p^2 2^{1-(2+\delta)\log p})^{1/16}+p^{-\Omega(1)} =2^{-\Omega(n)}.\]

\textbf{Case 2: $t' \neq 0$}. This corresponds to the case where $\psi'$ is non-trivial. In this case, note that

\[\psi(W) \psi'(W') = e^{2 \pi i t (\Enc(X) \cdot \Enc(Y))} e^{2 \pi i t' (\Enc(X) \cdot \Enc(Y'))} = e^{2 \pi i t (\Enc(X) \cdot (\Enc(Y) + r \Enc(Y'))},\]
where $r=t'/t \in \F_p$ and $r \neq 0$ since $t \neq 0$ and $t' \neq 0$.

Let $\widetilde{\Enc(Y)}=\Enc(Y) + r \Enc(Y')$, then 

\[E_{W, W'}[ \psi(W) \psi'(W')] = E_{X, Y}[ \psi(W) \psi'(W')]=E_{X, Y}[\psi(\Enc(X) \cdot \widetilde{\Enc(Y)})].\]

Now again by the same argument as above we have that $4\Enc(X)-4\Enc(X)$ is  $p^{-\Omega(1)}$-close to having min-entropy $(1/2+\delta) \log (p^2)$. Now we only need to bound the min-entropy of $\widetilde{\Enc(Y)}$.

If for every two different $y_1, y_2$, we have that $\Enc(y_1) + r \Enc(y_1') \neq \Enc(y_2) + r \Enc(y_2')$, then obviously $\widetilde{\Enc(Y)}$ will have the same min-entropy as $Y$. Now assume that for some two different $y_1, y_2$, we have $\Enc(y_1) + r \Enc(y_1') = \Enc(y_2) + r \Enc(y_2')$.

This gives us

\[y_1 +r y_1' = y_2 +r y_2'\] and

\[(y_1)^2 +r (y_1')^2 = (y_2)^2 +r (y_2')^2.\]

Hence we get

\[ y_1-y_2 = r (y_2'-y_1')\] and 

\[ (y_1+y_2)(y_1-y_2) = r(y_2'+y_1')(y_2'-y_1').\]

Since $y_1 \neq y_2$ and $r \neq 0$, we must have that $y_1' \neq y_2'$. Thus we get 

\[y_1+y_2 = y_2' + y_1'.\]

Therefore we can completely solve the equations and get 

\[y_1' = ((r+1)y_2+(r-1)y_1)/2r, \text{ }  y_2' = ((r+1)y_1+(r-1)y_2)/2r.\]

Thus any element in $\Supp(\widetilde{\Enc(Y)})$ can come from at most 2 elements in $\Supp(Y)$. To see this, assume for the sake of contradiction that we have $\Enc(y_1) + r \Enc(y_1') = \Enc(y_2) + r \Enc(y_2')=\Enc(y_3) + r \Enc(y_3')$ for three different $y_1, y_2, y_3$. Thus by above we have

\[y_1' = ((r+1)y_2+(r-1)y_1)/2r\] and

\[y_1' = ((r+1)y_3+(r-1)y_1)/2r.\]

Note that $r \neq -1$ since otherwise this would imply that $y_1' = y_1$ which contradicts the assumption that $\forall y, \adv(y) \neq y$. Thus we get $y_2=y_3$, another contradiction.

Therefore the min-entropy of $\widetilde{\Enc(Y)}$ is at least $H_{\infty}(Y)-1=(1-\delta)n-1$. Now since $(1/2+\delta) \log (p^2)+(1-\delta)n-1> (1+2\delta) \log p+(1-\delta)(\log p-1)-1 > (2+\delta)\log p-2$, by \lemmaref{lem:char2} we have

\[|E_{W, W'}[ \psi(W) \psi'(W')]| = |E_{X, Y}[ \psi(\Enc(X) \cdot \widetilde{\Enc(Y)})] | \leq (p^2 2^{2-(2+\delta)\log p})^{1/16}+p^{-\Omega(1)} =2^{-\Omega(n)}.\]
\end{proof}

Now we can prove the following theorem.

\BT
Let $\delta$ be the constant from \lemmaref{lem:grow}. Given any $(n,k)$ source $X$ with $k=(1/2-\delta)n$ and an independent uniform seed $Y \in \bits^n$, as well as any deterministic function $\adv: \bits^n \to \bits^n$ such that $\forall y, \adv(y) \neq y$, 

\[|(\nmExt(X, Y), \nmExt(X, \adv(Y)), Y)-(U_m, \nmExt(X, \adv(Y)), Y)| \leq \e, \]
where $\e=2^{-\Omega(n)}$ and output size $m=\Omega(n)$.
\ET

\begin{thmproof}
Let $Z=\nmExt(X, Y)$ and $Z'=\nmExt(X, \adv(Y))$. By \lemmaref{lem:main} and \lemmaref{lem:noxor}, we can choose an $m=\Omega(n)$ and $M=2^m$ such that when $\nmExt(X, Y)=\IP(\Enc(X), \Enc(Y)) \text{ mod } M$ and $Y$ is an $(n, (1-\delta)n)$ source independent of $X$, we have 

\[|(Z, Z')-(U_m, Z')| \leq \e',\]
where $\e'=O(n 2^m 2^{-\Omega(n)}+2^{m-n})=2^{-\Omega(n)}$.

Therefore when $Y$ is an independent uniform distribution over $\bits^n$, by \theoremref{thm:snmext} we have

\[|(Z, Z', Y)-(U_m, Z', Y)| \leq \e,\]
where $\e=2^{2m}(2^{1-\delta n}+\e')$.

Note that $\e'=O(n 2^m 2^{-\Omega(n)}+2^{m-n})$. Thus we can take $m=\Omega(n)$ and $\e=2^{2m}(2^{1-\delta n}+\e')=2^{-\Omega(n)}$. Thus the theorem is proved.
\end{thmproof}

\end{document}